\documentclass[12pt,a4paper]{article}

\usepackage[truedimen,margin=30mm]{geometry} 

\usepackage{mathrsfs}
\usepackage{amssymb}
\usepackage{amsmath}
\usepackage{ascmac}
\usepackage{amsthm}
\usepackage[dvipdfmx]{graphicx}
\usepackage{natbib}
\usepackage{setspace}

\usepackage{url}
\usepackage{color}
\usepackage{comment}
\usepackage{times}

\usepackage{titlesec}
\titleformat*{\section}{\large\bfseries}
\titleformat*{\subsection}{\it}

\newtheorem{thm}{Theorem}
\newtheorem{lem}{Lemma}

\newtheorem{exm}{Example}
\newtheorem{algo}{Algorithm}
\newtheorem{as}{Assumption}
\def\ep{{\varepsilon}}

\title{{\bf Scalable Estimation of Crossed Random Effects Models via Multi-way Discretization}}

\date{}

\begin{document}

\maketitle
\doublespacing

\vspace{-1.5cm}
\begin{center}
{\large Shota Takeishi$^1$ and Shonosuke Sugasawa$^2$}

\medskip
%\today

\medskip
\noindent
$^1$Department of Statistics and Data Science, Washington University in St. Louis\\
$^2$Faculty of Economics, Keio University
\end{center}

\vspace{0.3cm}
\begin{center}
{\bf \large Abstract}
\end{center}

\vspace{-0cm}
Cross-classified data frequently arise in scientific fields such as education, healthcare, and social sciences. A common modeling strategy is to introduce crossed random effects within a regression framework. However, this approach often encounters serious computational bottlenecks, particularly for non-Gaussian outcomes. In this paper, we propose a scalable and flexible method that approximates the distribution of each random effect by a discrete distribution, effectively partitioning the random effects into a finite number of representative values. This approximation allows us to express the model as a multi-way discrete structure, which can be efficiently estimated using a simple and fast iterative algorithm. The proposed method accommodates a wide range of outcome models and remains applicable even in settings with more than two-way cross-classification. We theoretically establish the consistency and asymptotic normality of the estimator under general settings of classification levels. Through simulation studies and real data applications, we demonstrate the practical performance of the proposed method in logistic, Poisson, and ordered probit regression models involving cross-classified structures.


\newpage
%--------------------------------------%
%          Introduction                %
%--------------------------------------%
\section{Introduction}\label{sec-intro}

Cross-classified data are found in various scientific fields, including marketing and ecology \citep[e.g.][]{allenby1998marketing,browne2005variance}. For example, sales data might be recorded for both customers and products; observations in this dataset can be classified by the customer making the purchase and by the product sold. The standard modeling strategy for such data is to use multi-way crossed random effects models, in which the random effect unique to each observation is additively decomposed into components corresponding to different ways.

Despite their widespread usage and methodological simplicity, crossed random effects models suffer from high computational costs involving the evaluation of multiple integrals. This issue is particularly severe when the outcome variable is non-Gaussian and the sample size is large. 
To address these issues, several scalable methods have been proposed, including composite likelihood approaches \citep{bellio2023scalable,xu2023gaussian}, variational approximations \citep{jeon2017variational,torabi2012likelihood}, iterative algorithms such as backfitting and block-wise updates \citep{ghosh2022backfitting,ghosh2022scalable,ghandwani2023scalable}, and Bayesian methods \citep{papaspiliopoulos2020scalable,ghosh2021convergence}. 
While these methods improve computational efficiency, they often require model-specific derivations tailored to the form of the likelihood and are limited in their applicability to a broader class of outcome distributions.

To address the aforementioned challenges, this study proposes a novel and unified modeling strategy, and a scalable estimation algorithm for crossed random effects models applicable to various outcome types. Inspired by the grouped fixed effects models for panel data analysis \citep[e.g.,][]{hahn2010panel, bonhomme2015grouped}, our approach assumes a discrete distribution for each way of the random effects. Such a discrete structure enables us to estimate the parameters, including the unknown discrete random effects, with a simple iterative algorithm that avoids computationally intensive procedures such as numerical integration. We also propose a post-processing procedure to obtain smoother estimates of random effects using quasi-posterior probabilities.

Regarding the theoretical properties of the proposed method, we derive the asymptotic distribution of the estimators for the relevant parameters in a reasonably broad class of models. In particular, we establish the asymptotic normality of the estimators for parameters induced by concave objective functions, which is applicable to typical regression settings such as normal, logistic, and Poisson models. In the context of crossed random effects models, it is surprising that the asymptotic distribution of the (restricted) maximum likelihood estimators has only recently been derived. \cite{Jiang2025aos} establishes the asymptotic normality for generalized linear mixed models with normal random effects, building on his prior consistency result \citep{Jiang2013aos}. Meanwhile, \cite{LyuSissonWelsh2024aos} derive the asymptotic normality for linear mixed models with potentially non-normal random effects and errors. Compared to these works, the theoretical results of the present study accommodate more general multi-way random effect models without interaction in the sense that the parameter of interest is only assumed to be the maximizer of concave objective functions under regularity conditions. Note that, outside the context of extremum estimators, the estimating equation-based approach  \citep[e.g.,][]{Jiang1998jasa} can be applicable to crossed random effects models, and the asymptotic distribution for some estimators has been derived \citep[e.g.,][]{JiangZhang2001robust}.

As mentioned above, the idea of discretization is borrowed from the econometric literature. Seminal works by \cite{hahn2010panel} and \cite{bonhomme2015grouped} introduce the discretized structure for one-way fixed effect models in panel data analysis. This framework has been applied and extended in various directions \citep[e.g.,][]{GuVolgushev2019joe, ito2023grouped, liu2020identification, sugasawa2021grouped}. However, previous studies have been limited to discrete structures for only one-way random or fixed effects. The present study, in contrast, establishes asymptotic theory for models with discrete structures on multi-way random effects, which turns out to be a nontrivial theoretical extension. \cite{fernandez2016individual} develops asymptotic theory for fixed effect estimators for a general model similar to ours but without the discrete structure. However, their limiting distribution incurs an asymptotic bias that ours does not, and their sample size requirement is stronger than ours.

The remainder of this paper is organized as follows. Section~\ref{sec:model} introduces the discretized crossed random effects models and their estimation algorithms. Section~\ref{sec:theory} establishes the asymptotic theory, including asymptotic normality results for the proposed estimators. In Section~\ref{sec:sim}, we evaluate the finite sample performance of the proposed method through simulation studies. Section~\ref{sec:app} illustrates an application to real movie rating data. Finally, Section~\ref{sec:disc} concludes with discussions. All proofs and some algorithmic details are provided in the supplementary material. 
%The R code implementing the proposed method is available at the GitHub repository (\url{https://github.com/sshonosuke/CGE}).

\section{Scalable Estimation of Crossed Random Effects Model}
\label{sec:model}

\subsection{Models and penalized likelihood}
For simplicity, we first consider two-way cross-classified data, but the extension to the case with a higher number of classification factors will be given in Section~\ref{sec:multi-way}. 
Let $y_{ij}$ be the response variable for $i=1,\ldots,n$ and $j=1,\ldots, m$, and $x_{ij}$ be a $p$-dimensional vector of covariates, where $n$ and $m$ are group sizes. 
We are interested in modeling the conditional distribution of $y_{ij}$ given $x_{ij}$, defined as 
\begin{equation}\label{model}
y_{ij}|x_{ij}\sim f(y_{ij};x_{ij}^{\top}\beta + a_i+b_j, \psi), 
\end{equation}
where $\beta$ is a vector of regression coefficients, $\psi$ is a dispersion or variance parameter, and $a_i$ and $b_j$ are unobserved additive effects. 
Note that the dimensions of $A=(a_1,\ldots,a_n)$ and $B=(b_1,\ldots,b_m)$ increase as the group sizes $n$ and $m$ grow, which is typically called incidental parameter problem \citep{neyman1948consistent}.
A typical example of the model is a Gaussian linear model with crossed-effects, $y_{ij}|x_{ij}\sim N(x_{ij}^{\top}\beta + a_i+b_j, \sigma^2)$, where $\psi=\sigma^2$. 
Without any restrictions on $a_i$ and $b_j$, the model is not necessarily identifiable.
Typically, it is assumed that $a_i$ and $b_j$ independently follow some distributions, denoted by $\pi_a(\cdot)$ and $\pi_b(\cdot)$, respectively, where $a_i$ and $b_j$ can be seen as random effects.
This setup leads to the following marginal likelihood: 
$$
\int \prod_{i=1}^n\prod_{j=1}^m f(y_{ij};x_{ij}^{\top}\beta + a_i+ b_j)\prod_{i=1}^n\pi_a(a_i)da_i\prod_{j=1}^m\pi_b(b_j)db_j,
$$
which requires high high-dimensional integral with respect to $a_i$ and $b_j$.
Hence, this integral is infeasible when the conditional distribution of $y_{ij}$ is not Gaussian.
A standard remedy is to use the Laplace approximation for the integral \citep[e.g.][]{bates2015fitting}, which may lead to poor approximation of the marginal likelihood.  
Even for the Gaussian case, the maximization of the marginal likelihood is computationally challenging \citep{ghosh2022backfitting}.
This issue becomes more serious when more than two-way structures are considered. 

To solve the computational issue, we propose a new modeling strategy of crossed effects, $a_i$ and $b_j$, by introducing discrete structures which does not require any integration. 
Suppose $a_1,\ldots,a_n$ have a discrete support
$\{ a(1), \dots, a(G) \}$ of size $G$.
We introduce an indicator value $g_i \in \{1, \dots, G\}$
that specifies which value in the support  the random effect $a_i$ takes: $a_i = a(g_i)$.
Similarly, $b_1, \dots, b_m$ have a discrete 
support $\{ b(1), \dots, b(H) \}$
of size $H$ with $h_j \in \{1, \dots, H\}$
as an indicator value so that $b_j =  b(h_j)$.
Hence, the proposed model is expressed as 
\begin{equation*}
y_{ij}|x_{ij}\sim f(y_{ij};x_{ij}^{\top}\beta + a_{i}+b_{j}),  \ \ \
a_i\in \{ a (1),\ldots, a(G)\}, \ \ \
b_j\in\{b (1),\ldots, b(H)\}.
\end{equation*}
or equivalently
\begin{equation}\label{model-proposal}
y_{ij}|x_{ij}\sim f(y_{ij};x_{ij}^{\top}\beta + a(g_i)+b (h_j)),  \ \ \
g_i\in \{1,\ldots, G\}, \ \ \
h_j\in\{1,\ldots, H\}.
\end{equation}
We refer to the model (\ref{model-proposal}) as a crossed discretized-effects (CDE) model due to its two-way discretized structure.
The unknown parameters in the CDE model are sample-specific two-way indicators $g_i$ and $h_j$, baseline regression coefficient $\beta$, and the supports of random effects $ a(1),\ldots,a(G)$ and $b (1),\ldots, b(H)$.
Under the model (\ref{model-proposal}), $g_i$ and $h_j$ are incidental parameters and $\{ a(1),\ldots,a(G)\}$ and $\{ b(1),\ldots, b(H)\}$ are finite-dimensional parameters and treated as fixed parameters rather than random effects.

Let $\Psi=\{\beta, a(1),\ldots,a(G), b(1),\ldots,b(H), g_1,\ldots,g_n, h_1,\ldots,h_m\}$ be a collection of unknown parameters. 
The estimation is performed by maximizing the following penalized log-likelihood function:
\begin{equation}\label{obj}
Q(\Psi)=\frac{1}{nm}\sum_{i=1}^n\sum_{j=1}^m \log f(y_{ij};x_{ij}^{\top}\beta + a(g_i)+ b(h_j)) - \frac{\lambda}{2}\left(\frac1n\sum_{i=1}^n a(g_i) -\frac1m\sum_{j=1}^m b(h_j) \right)^2,
\end{equation}
where $\lambda$ is a tuning parameter. 
We note that the first term of (\ref{obj}) corresponds to the log-likelihood function of the model (\ref{model-proposal}) and the second term, inspired by \cite{fernandez2016individual}, corresponds to a penalty for crossed effects, $a_i$ and $b_j$. 
Since adding the same constant to each element of $\{ a(1), \dots, a(G) \}$ and subtracting it from each element of $\{ b(1), \dots, b(H) \}$ would not change the value of the log-likelihood function, the penalty term in (\ref{obj}) normalizes the locations of $a_i$ and $b_j$. 
This is different from the role of the $L_2$ penalty in the equation (8) of \cite{ghosh2022backfitting} based on \cite{Robinson1991ss}, where the penalty provides an alternative representation of the estimator and is tied to BLUP-type estimation of the random effects.
Hence, $\lambda$ in \eqref{obj} can be set to an arbitrary large value (e.g. $\lambda=100$) to impose a constraint between two crossed discretized-effects, and the estimation result would not be sensitive to the value of $\lambda$.  
Since the penalty term controls only the difference of locations of $a(g_i)$ and $b(h_j)$, the intercept term can be recovered by the average of $a(g_i)$ and $b(h_j)$, namely, $n^{-1}\sum_{i=1}^n a(g_i) + m^{-1}\sum_{j=1}^m b(h_j)$ and an intercept term should not be included in the regression part, $x_{ij}^\top \beta$.  

It is worth noting that, although our formulation does not introduce an explicit variance component parameter for the random effects, their variability is still captured through the dispersion of the estimated grouped effects.
In particular, an analogue of the variance can be obtained as the empirical variance of the support points with respect to their associated weights. 
From this perspective, the proposed approach can be viewed as a nonparametric extension of conventional random effects models, where the distribution of the random effects is approximated by a discrete mixture without imposing parametric assumptions such as normality.

\subsection{Estimation algorithm }\label{sec:estimation}
We first consider maximizing the penalized likelihood function (\ref{obj}) to get estimates of $\Psi$, which can be done via the following iterative algorithm.

\begin{algo}\label{algo1}
%Let \[\Psi^{(s)} := \{ \beta^{(s)}, a^{(s)} (1), \dots, a^{(s)}(G), b^{(s)} (1), \dots, b^{(s)}(H), g_1^{(s)}, \dots, g_n^{(s)}, h_1^{(s)}, \dots, h_m^{(s)}\}\] denote the value of $\Psi$ after the $s$-th iteration is complete.
Starting with $s = 1$ and some initial values $a^{(0)} (1), \dots, a^{(0)}(G)$, $b^{(0)} (1), \dots, b^{(0)}(H)$, $g_1^{(0)}, \dots, g_n^{(0)}$, and $h_1^{(0)}, \dots, h_m^{(0)}$, repeat the following procedure until convergence:
\begin{enumerate}
\item 
Update the regression coefficient $\beta$ and dispersion parameter $\psi$ as 
\begin{align*}
(\beta^{(s)},\psi^{(s)})&={\rm argmax}_{\beta,\psi}\sum_{i=1}^n\sum_{j=1}^m \log f(y_{ij};x_{ij}^{\top}\beta + a^{(s-1)}(g_i^{(s-1)}) + b^{(s-1)} (h_j^{(s-1)}),\psi).
\end{align*}

\item 
For $g=1,\ldots,G$, update a point $a(g)$ in the support of $a_i$ 
 as $a^{(s)}(g)={\rm argmax}_{a(g)} \widetilde{Q}^{(s)}_{\Psi\setminus a(g)}(a(g))$, where $\widetilde{Q}^{(s)}_{\Psi\setminus a(g)}(a(g))$ is defined as 
\begin{align*}
\sum_{i=1}^n\sum_{j=1}^m I(g_i^{(s-1)}=g)\log f\big(y_{ij};x_{ij}^{\top}\beta^{(s)} + a (g) + b^{(s-1)} (h_j^{(s-1)}),\psi^{(s)}\big)  - \frac{\widetilde{\lambda}_g^{(s)}}{2}\big(a(g) -\widetilde{a}_g^{(s-1,s)}\big)^2
\end{align*}
and 
\begin{equation*}
\begin{split}
\widetilde{a}_g^{(s-1,s)}
= &\frac{n}{n_g^{(s)}}\left[
\frac1m\sum_{j=1}^m b^{(s-1)}(h_j^{(s-1)}) \right. \\
&\left. - \frac{1}{n}\sum_{i=1}^n 
\left\{
I(g_i^{(s-1)}<g)a^{(s)}(g_i^{(s-1)})+ I(g_i^{(s-1)}>g)a^{(s-1)} (g_i^{(s-1)})
\right\}
\right]
\end{split}
\end{equation*}
with $n_g^{(s)}=\sum_{i=1}^n I(g_i^{(s-1)}=g)$ and  $\widetilde{\lambda}_g^{(s)}=(n_g^{(s)})^2\lambda m/n$.

\item 
For $i=1,\ldots,n$, update the indicator $g_i$ as $g_i^{(s)}={\rm argmax}_{g=1,\ldots,G} \widetilde{Q}^{(s)}_{\Psi\setminus g_i}(g)$, where 
\begin{align*}
\widetilde{Q}^{(s)}_{\Psi\setminus g_i}(g)
&=\frac{1}{nm}\sum_{j=1}^m \log f(y_{ij};x_{ij}^{\top}\beta^{(s)} + a^{(s)} (g) + b^{(s-1)} (h_j^{(s-1)}),\psi^{(s)})\\
& - \frac{\lambda}{2}\left(\frac{a(g)}{n}
+ \frac1n\sum_{i'<i} a^{(s)} (g_{i'}^{(s)})
+\frac1n\sum_{i'>i} a^{(s)} (g_{i'}^{(s-1)})
-\frac1m\sum_{j=1}^m b^{(s-1)} (h_j^{(s-1)})\right)^2.
\end{align*}

\item 
For $h=1,\ldots,H$, update a point $b(h)$
in the support of $b_j$ as $b^{(s)} (h)={\rm argmax}_{b (h)}\widetilde Q_{\Psi\setminus b(h)}^{(s)}(b (h))$, where $\widetilde Q_{\Psi\setminus b(h)}^{(s)}(b(h))$ is defined as 
\begin{align*}
&\sum_{i=1}^n\sum_{j=1}^m I(h_j^{(s-1)}=h)\log f\big(y_{ij};x_{ij}^{\top}\beta^{(s)} + a^{(s)} (g_i^{(s)}) + b(h),\psi^{(s)}\big) - \frac{\widetilde{\lambda}_h^{(s)}}{2}\big(b(h)-\widetilde{b}_h^{(s-1,s)}\big)^2
\end{align*}
and 
\begin{equation*}
\begin{split}
\widetilde{b}_h^{(s-1,s)}
=&\frac{m}{m_h^{(s)}}\left[
\frac1n\sum_{i=1}^n a^{(s)} (g_i^{(s)}) \right. \\ 
&\left. - \frac{1}{m}\sum_{j=1}^m 
\left\{
I(h_j^{(s-1)}<h)b^{(s)} (h_j^{(s-1)})+ I(h_j^{(s-1)}>h)b^{(s-1)} (h_j^{(s-1)})
\right\}
\right]
\end{split}
\end{equation*}
with $m_h^{(s)}=\sum_{j=1}^m I(h_j^{(s-1)}=h)$ and  $\widetilde{\lambda}_h^{(s)}=(m_h^{(s)})^2\lambda n/m$.

\item
For $j=1,\ldots,m$, update the indicator $h_j$ as $h_j^{(s)}={\rm argmax}_{h=1,\ldots,H} \widetilde{Q}^{(s)}_{\Psi\setminus h_j}(h)$, where 
\begin{align*}
\widetilde{Q}^{(s)}_{\Psi\setminus h_j}(h)
&=\frac{1}{nm}\sum_{j=1}^m \log f(y_{ij};x_{ij}^{\top}\beta^{(s)} + a^{(s)} (g_i^{(s)}) + b^{(s)} (h),\psi^{(s)})\\
& - \frac{\lambda}{2}\left(\frac{b(h)}{m}
+ \frac1m\sum_{j'<j} b^{(s)} (h_{j'}^{(s)})
+\frac1m\sum_{j'>j} b^{(s)} (h_{j'}^{(s-1)})
-\frac1n\sum_{i=1}^n a^{(s)} (g_i^{(s)})\right)^2.
\end{align*}

\item 
Stop if the convergence criterion is met; otherwise, set $s \leftarrow s+1$ and return to Step~1.
\end{enumerate}
\end{algo}

The above algorithm is a conditional maximization algorithm, and each step guarantees a monotone increase in the objective function.
Each of the updating processes described above can be easily carried out.
Specifically, updating the grouping parameters, $g_i$ and $h_j$, can be done by evaluating the objective function at a finite number of inputs, and the updating process of $a(g)$ (similarly for $b(h)$) is the same as fitting the regression model with offset term $x_{ij}^{\top}\beta^{(s)}+ b^{(s-1)} (h_j^{(s-1)})$ based only on the data 
with $g_i = g$.
The updating process of $(\beta,\psi)$ is the same as fitting a regression model with offset $a^{(s-1)} (g_i^{(s-1)})+ b^{(s-1)}(h_j^{(s-1)})$. 
In particular, if the conditional distribution is Gaussian, namely, $y_{ij}\sim N(x_{ij}^\top\beta + a_i+b_j, \sigma^2)$, the updating processes for $(\beta,\sigma^2)$ can be obtained in the following closed form: 
\begin{align*}
&\beta^{(s)}
=\left(\sum_{i=1}^n\sum_{j=1}^m x_{ij}x_{ij}^\top \right)^{-1}\sum_{i=1}^n\sum_{j=1}^m x_{ij}\left(y_{ij}-a^{(s-1)} (g_i^{(s-1)}) - b^{(s-1)} (h_j^{(s-1)})\right)\\
&(\sigma^2)^{(s)}
=\frac{1}{nm}\sum_{i=1}^n\sum_{j=1}^m \left(y_{ij}-x_{ij}^\top \beta^{(s)} - a^{(s-1)} (g_i^{(s-1)})- b^{(s-1)} (h_j^{(s-1)})\right)^2,
\end{align*}
and that for a point $a(g)$ in the support, is 
\begin{align*}
a^{(s)} (g)
&=\left\{\sum_{i=1}^n\sum_{j=1}^m (\sigma^{-2})^{(s)} I(g_i^{(s-1)}=g) + \widetilde{\lambda}_g^{(s)}\right\}^{-1}\\
&\times 
\left\{\sum_{i=1}^n\sum_{j=1}^m 
(\sigma^{-2})^{(s)} I(g_i^{(s-1)}=g)\left(y_{ij}-x_{ij}^\top \beta^{(s)}- b^{(s-1)} (h_j^{(s-1)})\right) + \widetilde{\lambda}_g^{(s)}\widetilde a_g^{(s-1,s)}
\right\}.
\end{align*}
Note that a quite similar updating process can be obtained for $b(h)$. 
For non-normal response, it would also be possible to update the parameter values of $\beta$ and the random effects by the one-step updating of the Newton-Raphson method.

The above algorithm would be computationally much less intensive than maximizing the marginal likelihood, as it does not involve complicated integrals or large matrices.
In particular, each iteration of the proposed algorithm involves only low-dimensional optimization problems, and thus the overall computational cost scales linearly with the sample size, i.e., $O(N)$.
This favorable computational property is analogous to that of the backfitting algorithm by \cite{ghosh2022backfitting}.

In the above algorithm, it is assumed that the unobserved effects have discrete distributions with finite supports.
In practice, more smoothed estimates may yield more accurate results than the original estimates produced by Algorithm~1.
To address this issue, we propose a post-hoc algorithm to provide more smoothed (fuzzy) estimates of $a_i$, $b_j$, and $\beta$. 
The main idea is to define a pseudo-posterior probability of $g_i=g$ proportional to $\prod_{j=1}^m f(y_{ij};x_{ij}^{\top}\widehat{\beta} + \widehat{a} (g) + \widehat{b} (\widehat{h}_j))$, which gives a smoothed estimate of $a_i$ as a weighted average of $\widehat{a}(g)$.
Similarly, we can define the pseudo-posterior probability of $h_j=h$ to obtain a smoothed estimate of $b_j$.
The post-hoc procedures are summarized as follows: 

\begin{algo}\label{algo2}
Given the estimates of $\Psi$, compute smoothed estimates of $a_i$ and $b_j$ as follows: 
\begin{enumerate}
\item 
For $i=1,\ldots,n$ and $g=1,\ldots,G$, compute pseudo-posterior probability
\begin{equation*}
\widetilde{\pi}_{ig}=\frac{\prod_{j=1}^m f(y_{ij};x_{ij}^{\top}\widehat{\beta} + \widehat{a}(g) + \widehat{b} (\widehat{h}_j))}
{\sum_{g'=1}^G \prod_{j=1}^m f(y_{ij};x_{ij}^{\top}\widehat{\beta} + \widehat{a} (g') + \widehat{b} (\widehat{h}_j))},
\end{equation*}
and compute $\widehat{a}_i=\sum_{g=1}^G \widetilde{\pi}_{ig}\widehat{a}(g)$.

\item
For $j=1,\ldots,m$ and $h=1,\ldots,H$, compute pseudo-posterior probability
\begin{equation*}
\widetilde{\pi}_{jh}=\frac{\prod_{i=1}^n f(y_{ij};x_{ij}^{\top}\widehat{\beta} + \widehat{a} (\widehat{g}_i) + \widehat{b} (h))}
{\sum_{h'=1}^H \prod_{i=1}^n f(y_{ij};x_{ij}^{\top}\widehat{\beta} + \widehat{a} (\widehat{g}_i) + \widehat{b} (h'))},
\end{equation*}
and compute $\widehat{b}_j=\sum_{h=1}^H \widetilde{\pi}_{jh}\widehat{b}(h)$.

\item
Re-estimate the regression coefficient $\beta$ as 
\begin{align*}
\widehat{\beta}&={\rm argmax}_{\beta}\sum_{i=1}^n\sum_{j=1}^m \log f(y_{ij};x_{ij}^{\top}\beta + \widehat{a}_i + \widehat{b}_j).
\end{align*}
\end{enumerate}
\end{algo}

Note that $\widetilde{\pi}_{ig}$ can be regarded as a pseudo-probability or posterior probability that the $i$th effect equals the $g$th point in the support.

\subsection{Interpretation of discretization and specification of the cardinality of supports}
Regarding the cardinality of the supports, $G$ and $H$, Section \ref{sec:theory} develops the asymptotic theory under the assumption that these quantities are fixed, in line with the grouped fixed-effect literature \citep[e.g.][]{hahn2010panel, bonhomme2015grouped}.
While the model still includes incidental parameters, $g_i$ and $h_j$  of increasing dimension, these indicators are consistently estimable, unlike the incidental parameters in fixed effects models. Consequently, the estimator of the relevant parameters is free from so-called ``incidental parameter bias,'' as proven in Section \ref{sec:theory}.

In practical random-effect modeling, however, imposing discreteness on the underlying random effect $a_i$ and $b_j$ might be restrictive.
Instead, $G$ and $H$ may be heuristically viewed as controlling the granularity of a discrete approximation to the true generative distribution of $a_i$ and $b_j$, which may in fact be continuous.
From this perspective, the group-level parameters represent support points of this discrete approximation, and increasing $G$ and $H$ corresponds to refining the approximation.
In particular, when $G = n$, the model reduces to a fixed effects specification for $a_i$, which suffers from the incidental parameter problem and leads to poor estimation accuracy. 
In contrast, setting $G < n$ allows for partial pooling across units, resulting in more stable estimates by borrowing information across similar individuals. However, if $G$ is too small, the random effect distribution may be approximated too coarsely, potentially leading to model misspecification.

Regarding the specification of $G$ and $H$, one could select $G$ and $H$ using a data-dependent selection criterion. 
For example, we can consider an information criterion for $G$ and $H$ of the form, $IC(G, H)\equiv -2L_N(G, H)+ P_N(G,H)$, where $L_N(G, H)=\sum_{i=1}^n\sum_{j=1}^m \log f(y_{ij}; x_{ij}^\top \widehat{\beta}+\widehat{a}_i+\widehat{b}_j)$ is the maximum log-likelihood and $P_N(G,H)$ is a penalty function.  
A possible example of the penalty function is $P_N(G,H)=2(G+H)$, leading to the AIC-type criterion. 
However, we do not recommend this approach since it is computationally intensive for large-scale applications, and the effect of $G$ and $H$ is limited as long as $G$ and $H$ are relatively large. 
As a practical and effective alternative, we suggest setting $G = \lfloor \sqrt{n} \rfloor \quad \text{and} \quad H = \lfloor \sqrt{m} \rfloor$, which would strike a balance between approximation accuracy and estimation stability.
In the supplementary material, we examine the sensitivity of the performance to the choice of $G$ and $H$ through numerical simulations.
%In the subsequent section, we establish the asymptotic properties of the estimator of $\beta$ under fixed, finite $G$ and $H$.
%In the Supplementary Material, we further provide sensitivity analyses with respect to the choice of $G$ and $H$ (and the number of supports under higher-way classifications), demonstrating that the results are largely insensitive to these choices and that data-dependent selection offerslittle additional benefit.

%
\subsection{Inference on regression coefficients}\label{subsec:inference}
We consider the asymptotic variance-covariance matrix of the estimator for $\beta$. As implied by discussion in Section~\ref{sec:theory}, the estimation error of the random effects has limited effect on the estimation of $\beta$, particularly when the number of groups is not so small and the random effects are not skewed. This observation justifies a simple and practical approach in which the random effects are treated as known when computing the asymptotic variance. Specifically, under the model, $f(y_{ij}; x_{ij}^\top \beta +  \hat{a} (\hat{g}_i) + \hat{b} (\hat{h}_j))$ with $\hat{a} (\hat{g}_i) + \hat{b} (\hat{h}_j)$ treated as an offset, the variance-covariance matrix of the estimator can be obtained by the inverse of the negative Hessian matrix of the log-likelihood 
$$
-\sum_{i=1}^n\sum_{j=1}^m \frac{\partial^2}{\partial \beta \partial \beta^\top} \log f(y_{ij}; x_{ij}^\top \beta +  \widehat{a} (\widehat{g}_i)+ \widehat{b} (\widehat{h}_j))\bigg|_{\beta=\widehat{\beta}},
$$
where $\hat{\beta}$ is the maximizer of $\sum_{i=1}^n\sum_{j=1}^m  \log f(y_{ij}; x_{ij}^\top \beta +  \widehat{a} (\widehat{g}_i) + \widehat{b} (\widehat{h}_j))$, which we refer to ``plug-in'' method. 
On the other hand, as shown in Section~\ref{sec:theory}, the estimator of $\beta$ defined through the loss function (\ref{obj}) admits asymptotic normality and the sandwich-type asymptotic covariance matrix can be consistently estimated, which we refer to ``Sandwich'' method. 
As confirmed in our numerical experiments, these estimators perform well when the sample size is large. 
For more refined inference, we recommend using weighted bootstrap procedures, where the bootstrap sample is generated by minimizing $Q_w(\Psi)$ defined by replacing $\log f(y_{ij}; x_{ij}^\top \beta + a(g_i)+b(h_j))$ with $w_{ij}\log f(y_{ij}; x_{ij}^\top \beta + a(g_i)+b(h_j))$ in $Q(\Psi)$ and $w=(w_{11},\ldots, w_{nm})$ is generated from a Dirichlet distribution, ${\rm Dir}(1,\ldots, 1)$ (i.e. uniform distribution on a simplex). 
The weighted loss function $Q_w(\Psi)$ can be easily optimized by a slight modification of Algorithm~\ref{algo1}.
Since such bootstrap procedure requires repeatedly optimizing the loss function, the computation cost will be considerably larger than Plug-in and Sandwich methods while it can naturally take into account of the uncertainty of discretization and would be recommended under moderate sample size.

\subsection{General models}\label{sec:multi-way}

We here discuss general models of (\ref{model-proposal}), while focusing on models without interaction terms between the classification factors.
For $q=1,\ldots,N$, let $y_q$ be a response variable and $x_q$ be a vector of covariates. 
Moreover, let $\ell_{a,q}\in \{1,\ldots,n\}$ and $\ell_{b,q}\in \{1,\ldots,m\}$ be indicators to exhibit which classes the $q$th observation belongs to. 
Then, we can consider a general model as 
\begin{equation}\label{eq:general-model}
y_q|x_q\sim f(y_q; x_q^\top \beta + a(g_{\ell_{a,q}}) + b (h_{\ell_{b,q}}), \psi), \ \ \ \ q=1,\ldots, N,
\end{equation}
where $g_1,\ldots,g_n$ and $h_1,\ldots,h_m$ indicate which values in the supports $a_{\ell_{a,q}}$ and $b_{\ell_{b,q}}$ take. 
The general formulation in (\ref{eq:general-model}) allows the total sample size $N$ to be greater than, equal to, or smaller than $nm$. 
When $N = nm$ and $(\ell_{a,q}, \ell_{b,q}) \in \{1,\ldots,n\} \times \{1,\ldots,m\}$, the above model is equivalent to (\ref{model}). 
More generally, the formulation accommodates unbalanced designs: when $N < nm$, some combinations of classification factors are not observed, whereas when $N > nm$, multiple observations may arise within the same cell, resulting in imbalanced cell sizes. 
For the general model (\ref{eq:general-model}), we can apply a similar iterative algorithm as described in Algorithm~\ref{algo1}.

The proposed method can also be extended to handle more than two-way classified data. 
Let $i_k = 1, \ldots, n_k$ be an index of the $k$th classification, for $k = 1, \ldots, K$, where $K$ is the number of classification types and $n_k$ is the size of the $k$th classification. 
Each observational unit is associated with a tuple $(i_1, \ldots, i_K)$ indicating its membership in each classification group. 
We denote the corresponding random effects by $(a_{1,i_1}, \ldots, a_{K,i_K})$, where $a_{k,i_k}$ represents the effect associated with the $i_k$th group in the $k$th classification.
Let $y_{i_1\ldots i_K}$ and $x_{i_1\ldots i_K}$ be outcome and covariate vector in $K$-ways structures. 
Then, we can consider the following multi-way crossed effects model, given by 
$$
y_{i_1\ldots i_K}|x_{i_1\ldots i_K}
\sim f(y_{i_1\ldots i_K}; x_{i_1\ldots i_K}^\top \beta + a_{1}(g_{1,i_1})+\ldots+a_{K}(g_{K,i_K}), \psi),
$$
where, for $k = 1, \dots, K$, the effect $a_{k, i_k}$ has support $\{ a_k (1), \dots, a_k (G_k) \}$ and $g_{k,i_k}\in \{1,\ldots,G_k\}$ indicates which value in the support $a_{k, i_k}$ takes.
The above model can be estimated in the same way as Algorithm~\ref{algo1} by introducing penalty terms for the difference of random effect means, whose details are given in Section~\ref{sec:theory}, where we establish asymptotic properties of the estimator under the general model.

\section{Asymptotic Properties}\label{sec:theory}
This section establishes the asymptotic properties of the grouped crossed random effect estimator for general models induced by concave objective functions. We first formally define the estimator and then establish its consistency and asymptotic normality in subsequent subsections.
We illustrate our general results with some classes of generalized linear models.
\subsection{Setting and estimator}
We begin by introducing the general model setup. Henceforth, the superscript $``0"$ or the tilde accent $``\sim"$ signify the true parameter value that characterizes the distribution of the observed data.
Furthermore, we use $:=$ to indicate ``equal by definition.''

Let $\{ z_{i_1 \dots i_K} : 1 \leq i_k \leq n_k, \ 1 \leq k \leq K \}$ be an array of observed random vectors with $K$-ways crossed structure, which is defined on some underlying probability space $(\Omega, \mathcal F, \mathbb P)$. 
For each $k=1, \dots, K$, random effects $a_{k, i_k} \ (i_k = 1, \dots, n_k)$ are assumed to have a discrete support $\{ \tilde a_k(1), \dots, \tilde a_k(G_k) \}$ with known cardinality $G_k$. 
Let $g^0_{k, i_k} \in \{1, \dots, G_k \} \ (i_k = 1, \dots, n_k)$ denote a (random) indicator, defined on $(\Omega, \mathcal F, \mathbb P)$, that specifies which point in the support $a_{k, i_k}$ takes as its realized value. Then we may write $a_{k, i_k} = \tilde a_k(g^0_{k, i_k})$.
We collect the unobserved true indicator structure as $\gamma^0 := ( {\gamma_1^0}^{\intercal}, \dots, {\gamma_K^0}^{\intercal} )^{\intercal}$ with $\gamma^0_k := ( g^0_{k, i_k} )_{i_k = 1}^{n_k}$ for $k = 1, \dots, K$, where we suppress the dependence of $\gamma^0$ and $\gamma^0_k$ on $n_k$ for brevity. Similarly, we define the vector for the support of cross random effects as $\tilde \alpha := ( \tilde {\alpha_1}^{\intercal}, \dots, \tilde {\alpha_K}^{\intercal} )^{\intercal}$ with $\tilde \alpha_k := ( \tilde a_k (1), \dots, \tilde a_k (G_k) )^{\intercal}$ for $k = 1, \dots, K$.

In this setting, our purpose is to estimate some finite-dimensional structural parameter $\beta^0$, the true group membership $\gamma^0$ and some location transformation of the support of the cross random effects $\alpha^0 := ( \alpha_1^{0^\intercal}, \dots, \alpha_K^{0^\intercal})^{\intercal}$ with $\alpha_{k}^0 := ( a^0_{k} (1), \dots, a^0_{k} (G_k) )^{\intercal}$. Collecting $\Psi^0 := ({\theta^0}^{\intercal}, {\gamma^0}^{\intercal})$, where $\theta^0 := ({\beta^0}^{\intercal}, \alpha^{0^\intercal})^{\intercal}$, this $\Psi^0$ is assumed to be characterized as the maximizer of the conditional expectation of some objective function given the true indicator structure $\gamma^0$:
\begin{equation}
    \Psi^0 = \arg \max_{\Psi \in \Xi} \mathbb E[Q (\Psi) | \gamma^0], \label{population minimization}
\end{equation}
where $\Xi$ is the parameter space, whose detail is discussed later, and a variable $\Psi$ has the same coordinate structure as $\Psi^0$: $\Psi := (\theta^{\intercal}, \gamma^{\intercal}) := (\beta^{\intercal}, \alpha^{\intercal}, \gamma^{\intercal})$ where  $\alpha := ( \alpha_1^\intercal, \dots, \alpha_K^\intercal)^{\intercal}$ and $\gamma := (\gamma_1^{\intercal}, \dots, \gamma_K^{\intercal} )^{\intercal}$ with $\alpha_k := (a_{k} (1), \dots, a_{k} (G_k)) )^{\intercal}$ and $\gamma_k := ( g_{k, i_k} )_{i_k = 1}^{n_k}$ for $k = 1, \dots, K$. The concave objective function $Q (\Psi)$ is decomposed as 
\begin{equation}
   Q (\Psi) = Q^* (\Psi) - Pen(\alpha, \gamma), \notag
\end{equation}
where the leading term $Q^* (\Psi)$ and the penalty term $Pen(\alpha, \gamma)$ have the following expressions:
\begin{align} 
 Q^* (\Psi) := \ &\frac{1}{n_1 \dots n_K} \sum_{i_1 = 1}^{n_1} \dots \sum_{i_K = 1}^{n_K} M (z_{i_1\dots i_K} | \beta,
 a_{1} (g_{1, i_1}) + \dots + a_{K}(g_{K, i_K})), \label{Q} \\
 Pen(\alpha, \gamma) := \ & \frac{\lambda}{2} \sum_{k=1}^{K-1} \left( \frac{1}{n_k} \sum_{i_k = 1}^{n_k} a_{k} (g_{k, i_k}) - \frac{1}{n_{k+1}} \sum_{i_{k+1} = 1}^{n_{k+1}} a_{k+1} (g_{k+1, i_{k+1}})\right)^2, \notag
\end{align}
for a tuning parameter $\lambda > 0$. The exact form of $M(z | \beta, a)$ depends on specific applications.
The penalty term, $Pen(\alpha, \gamma)$, is for the sake of normalization as explained in Section \ref{sec:model}.
Due to this normalization, the solution to the maximization problem is the location transformation $\alpha^0$, not necessarily $\tilde \alpha$ itself.
A similar normalization approach is employed in \cite{fernandez2016individual}. Note that the location transformation $\alpha^0$ might depend on the sample size even though $\tilde  \alpha$ does not; however, we suppress this dependency for notational simplicity.

We illustrate the general setting with the following likelihood models, which will be used to develop the generalized linear models later in this section.
\begin{exm}[Likelihood] \label{exm likelihood} Let $z_{i_1 \dots i_K} = (y_{i_1 \dots i_K}, x_{i_1 \dots i_K}^{\intercal})^{\intercal}$, where $y_{i_1 \dots i_K}$ is a scalar outcome variable and $x_{i_1 \dots i_K}$ is a vector of explanatory variables. Assume the following conditional distribution of $y_{i_1 \dots i_K}$ given $x_{i_1 \dots i_K}$ and the crossed random effects $\tilde \alpha$:
\begin{equation}
    y_{i_1 \dots i_K} | x_{i_1 \dots i_K}, \tilde \alpha \sim f (y_{i_1 \dots i_K} | x_{i_1 \dots i_K}, \tilde a_{1} (g^0_{1, i_1}) + \dots + \tilde a_{K} (g^0_{K, i_K}); \beta^0, \psi^0), \label{likelihood}
\end{equation}
where $f$ is a conditional density of $y_{i_1 \dots i_K}$ with respect to some dominating measure, $\beta^0$ is unknown parameter of interest, and $\psi^0$ is unknown nuisance parameter if exists. In this example, the objective function $M(z| \beta, a)$, if possible, equals the log-likelihood function $\log f(y|x, a; \beta, \psi)$ up to the effect of $\psi$. Under certain regularity conditions, $(\beta, a) \mapsto \mathbb E [M(z_{i_1 \dots i_K} | \beta, a) | \gamma^0]$ is uniquely maximized at $(\tilde a_{1} (g^0_{1, i_1}) + \dots + \tilde a_{K} (g^0_{K, i_K}), \beta^0)$; hence, the unpenalized objective function $\mathbb E[Q^* (\Psi) | \gamma^0]$ is maximized at $(\beta^0, \tilde \alpha, \gamma^0)$. To take into account for the penalization, let $\eta_k := \frac{1}{n_k} \sum_{i_k = 1}^{n_k} \tilde a_{k} (g^0_{k, i_k}) - \frac{1}{_{k+1}} \sum_{i_{k + 1} = 1}^{n_{k + 1}} \tilde a_{k+1} (g^0_{k+1, i_{k+1}})$ for $k = 1, \dots, K-1$, and consider a matrix,
\begin{equation}
    W =
    \begin{pmatrix}
        1 & \cdots & \cdots & \cdots & \cdots & 1 \\
        1 & - 1 & 0 & \cdots & \cdots & 0 \\
        0 & 1 & -1 & 0 & \cdots & 0 \\
        \vdots & \vdots & \vdots & \vdots & \vdots & \vdots \\
        0 & \cdots & \cdots & 0 & 1 & -1
    \end{pmatrix}. \label{W}
\end{equation}
By Lemma 1 in online Appendix C, the system of equation $W \zeta = (0, -\eta_1, \dots, -\eta_{K-1})^{\intercal}$ is solvable. In particular, each element $\zeta_k$ of $\zeta = (\zeta_1, \dots, \zeta_K)^{\intercal}$ is a linear combination of $\frac{1}{n_k} \sum_{i_k = 1}^{n_k} \tilde a_{k} (g^0_{k, i_k}) \ (k = 1, \dots, K)$
with bounded coefficients. Now the normalized cross random effects can be obtained by setting $a^0_{k} (g_k) := \tilde a_{k} (g_k) + \zeta_k$ for $g_k = 1, \dots, G_k$ $(k = 1, \dots, K)$. By construction of $\zeta$, $a^0_{1}(g_1) + \dots + a^0_{K} (g_K) = \tilde a_{1}(g_1) + \dots + \tilde a_{K} (g_K)$ for any $g_1, \dots, g_K$ so that $(\beta^0, \tilde \alpha^0, \gamma^0)$ maximizes $\mathbb E[Q^* (\Psi)|\gamma^0]$. Furthermore, $Pen (\tilde \alpha^0, \gamma^0) = 0$; hence, $\mathbb E[Q(\Psi)|\gamma^0]$ is maximized at $(\beta^0, \tilde \alpha^0, \gamma^0)$.
\end{exm}

We estimate $\Psi^0$ via the sample analog of the maximization \eqref{population minimization}.
Namely, the grouped crossed random effect estimator $\hat \Psi := (\hat \theta^{\intercal}, \hat \gamma^{\intercal})^{\intercal}$ with $\hat \theta := (\hat \beta^{\intercal}, \hat \alpha^{\intercal})^{\intercal}$ is defined as follows.
\begin{equation}
	\hat \Psi := \arg \max_{\Psi \in \Xi } Q (\Psi). \label{sample minimization}
\end{equation} 
Note that this $\hat \Psi$ has the same coordinate structure as $\Psi^0$ such that $\hat \alpha = ( \hat \alpha_1^\intercal, \dots, \hat \alpha_K^\intercal )^{\intercal}$ and $\hat \gamma = ( \hat \gamma_1^{\intercal}, \dots, \hat \gamma_K^{\intercal} )^{\intercal}$, where $\hat \alpha_k := ( \hat a_{k}(1), \dots, \hat a_{k}(G_k) )^{\intercal}$ and $\hat \gamma_k := (\hat g_{k, i_k} )_{i_k = 1}^{n_k}$ for $k = 1, \dots, K$.

In the remainder of this section and the supplementary material, we use the following notation. For any real vector $a$ and real matrix $A$, $\| a \|$ denotes the Euclidean norm of $a$ while $\| A \|$ denotes the Frobenius norm of $A$. 
For a real-valued function $(\beta, a) \mapsto L(M(z_{i_1 \dots i_K} | \beta, a))$, with $L(\cdot)$ being some functional and $M(z_{i_1 \dots i_K} | \beta, a)$ as in \eqref{Q}, let $\nabla_{\beta} L(M(z_{i_1 \dots i_K} |  \beta^*, a^*))$ denote a vector of partial derivatives of $L(M(z_{i_1 \dots i_K}|\beta, a))$ with respect to $\beta$ evaluated at $(\beta^*, a^*)$, $\nabla_{a} L(M(z_{i_1 \dots i_K} | \beta^*, a^*))$ denote a partial derivative of the same function with respect to $a$ evaluated at $(\beta^*, a^*)$, and
\[\nabla_{\beta_{a}} L(M(z_{i_1 \dots i_K}|\beta^*, a^*)) = (\nabla_{\beta} L(M(z_{i_1 \dots i_K} | \beta^*,  a^*))^{\intercal}, \nabla_{a} L(M(z_{i_1 \dots i_K}|\beta^*, a^*)))^{\intercal}.\]
Similarly, we define scalars/vectors/matrices of the second partial derivatives, such as
$\nabla_{a a} L(M(z_{i_1 \dots i_K}| \beta^*, a^*))$, $\nabla_{a \beta_a^{\intercal}} L(M(z_{i_1 \dots i_K}|\beta^*, a^*))$ and $\nabla_{\beta_a \beta_a^{\intercal}} L(M(z_{i_1 \dots i_K}|\beta^*, a^*))$ in a usual manner.
We use the abbreviations $N := \prod_{k = 1}^K n_k$ and $\sum_{i_1 \dots i_K} := \sum_{i_1 = 1}^{n_1} \dots \sum_{i_K = 1}^{n_K}$. For any two real sequences $\{ b_{n_1 \dots n_K} \}$ and $\{ c_{n_1 \dots n_K} \}$, $b_{n_1 \dots n_K} \lesssim c_{n_1 \dots n_K}$ means that there exists a finite constant $\mathcal D$ independent of $n_1, \dots, n_K$ such that $b_{n_1 \dots n_K} \leq \mathcal D c_{n_1 \dots n_K}$ for all $n_1, \dots, n_K$.
For any two real numbers $b$ and $c$, $b \lor c$ denotes $\max(b, c)$ while $b \land c$ denotes $\min(b, c)$. All the limits are taken as $n_k \rightarrow \infty$ for all $k = 1, \dots, K$ unless stated otherwise.

\subsection{Consistency and the asymptotic normality}
We first show the consistency of $\hat \theta$ for $\theta^0$. Regarding $\hat \alpha$, note that the objective function is invariant to reordering of the support indices, $1, \dots, G_k \ (k = 1, \dots, K)$. Hence, following Appendix B of \cite{bonhomme2015grouped}, we introduce the following Hausdorff distance $d_{H, k}$ in $\mathbb R^{G_k} \ (k = 1, \dots, K)$ for the consistency of $\alpha$:
\begin{equation}
    d_{H, k} (b, c)^2 := \max_{g \in \{ 1, \dots, G_k \}} \left( \min_{g' \in \{ 1, \dots, G_k \}} (b_{g'} - c_g)^2 \right) \lor \max_{g'
     \in \{1, \dots, G_k \}} \left( \min_{g \in \{ 1, \dots, G_k \}} (b_{g'} - c_g)^2 \right), \label{harsdorff}
\end{equation}
for real vectors $b := (b_1, \dots, b_{G_k})^{\intercal}$ and $c := (c_1, \dots, c_{G_k})^{\intercal}$. The following theorem establishes the consistency of $\hat \theta$ for $\theta^0$.
\begin{thm} \label{consistency m-estimator}
    Assume that, conditionally on $\gamma^0$, $z_{i_1 \dots i_K}$ are independent across all $i_1, \dots, i_K$. Then, under Assumptions 1, 2 and 3 in online Appendix A, $\hat \beta \rightarrow_p \beta^0$, and $d_{H, k} (\hat \alpha_k, \alpha^0_k) \rightarrow_p 0$ for $k = 1, \dots, K$.
\end{thm}
This consistency result is the essential first step for deriving the asymptotic distribution of $\hat \theta$. The same type of results has been established for several fixed effect models with grouping structure \citep[e.g.][]{bonhomme2015grouped, liu2020identification}. It is noted that $\gamma^0$ in the statement of the theorem can be replaced by the corresponding random effects $\alpha^0$ because the randomness of the latter comes solely from its support indicator.
This assumption is the conditional independence assumption, and is standard in the literature on random effect models \citep[e.g.,][]{Jiang2025aos}. In contrast, a somewhat weaker condition is commonly found in the econometric literature on fixed effects models. For example, Assumption 4.1 of \cite{fernandez2016individual} assumes conditional independence along one way while allowing weak conditional dependency along the other for their two-way fixed effect models. 

To proceed to the asymptotic normality for $\hat \theta$, we consider relabelling of indices of $\hat \alpha_k$ to ensure $\| \hat \alpha_k - \alpha^0_{k} \| \rightarrow_p 0$, following \cite{bonhomme2015grouped}. For any $k = 1, \dots, K$, let $\tilde \tau:\{1, \dots, G_k \} \mapsto \{1, \dots, G_k \}$ be a permutation such that $\tilde \tau = \tau$ with $\tau$ defined in the proof of Theorem \ref{consistency m-estimator} if $\tau$ is one-to-one, and $\tilde \tau(g) = g$ otherwise.
It follows from the proof of Theorem \ref{consistency m-estimator} that $(\hat a_{k} (\tilde \tau(g)) - a^0_{k} (g))^2 \rightarrow_p 0$ for $g = 1, \dots, G_k \ (k = 1, \dots, K)$. By simply relabelling the indices, we may assume $\tilde \tau (g) = g$ so that $\|\hat \alpha_k - \alpha^0_k \| \rightarrow_p 0$. In the following, we adopt this convention of ordering. We now establish the asymptotic normality of $\hat \theta$ and the consistency of the grouping structure $\hat \gamma$.
\begin{thm}\label{prop asy dist}
    Assume the assumption of Theorem \ref{consistency m-estimator}. Furthermore, for any $k = 1, \dots, K$ and $l = 1, \dots, K$, assume that $(\log n_l \log n_k) / \left( \prod_{m=1, m \neq k}^K n_m \right) \rightarrow 0$. Then, under Assumptions 1, 2, 3 and 4 in online Appendix A, $\mathbb P (\hat \gamma = \gamma^0) \rightarrow 1$, and $\sqrt{N} (\hat \theta - \theta^0)$ converges in distribution to a mean zero normal distribution. In particular, $\sqrt{N} (\hat \beta - \beta^0) \rightarrow_d N(0, V)$ for some positive definite matrix $V$. 
\end{thm}
This theorem can be thought of as an extension of the asymptotic normality of the grouped-fixed estimator for one-way fixed effect \citep[e.g.,][]{bonhomme2015grouped, liu2020identification} to the multi-way cases. Furthermore, unlike usual two-way fixed effect estimators without the discrete support structure in \cite{fernandez2016individual}, our estimator does not suffer the asymptotic bias.
The presence of the logarithmic factors in the rate condition indicates that we allow the group size $n_k$ for the $k$-th way to diverge nearly exponentially fast relative to the other group sizes.
This rate condition is weaker than, for example, an assumption in Theorem 1 of \cite{Jiang2025aos} and Assumption 4.1$(i)$ of \cite{fernandez2016individual}, where different group sizes must be of the same order. In Appendix F of the supplementary material, we make a detailed comparison between the rate condition of ours and that of \cite{Jiang2025aos}.
Note that \cite{LyuSissonWelsh2024aos}
does not impose any condition on the relative speed of divergence, but only assume that the ratio of different group sizes have the limit in $[0, \infty]$ (see part 2 of Condition B in the supplementary B of \cite{LyuSissonWelsh2024aos}). 
Although the condition of \cite{LyuSissonWelsh2024aos} might appear weaker and this could be the case in practical situations, we can make an artificial example where the assumption of \cite{LyuSissonWelsh2024aos} does not hold while our assumption does. See Appendix F of the supplementary material for the detail.
Also note that the focus of \cite{LyuSissonWelsh2024aos} is on linear models while ours is on models with parameter defined by objective functions.

The expression of the asymptotic covariance matrix of $\sqrt{N} (\hat \theta - \theta^0)$ is given in the proof. The expression involves the tuning parameter $\lambda$; however, Lemma 4 in online Appendix C clarifies that the matrix, in fact, is independent of the value of $\lambda$ in reasonably general settings including Example \ref{exm likelihood}. 
Note that a simple expression for the asymptotic covariance matrix $V$ for $\hat \beta$ has not been derived. The construction of the sample analogue of $V$ as a submatrix of the asymptotic covariance matrix of $\sqrt{N} (\hat \theta - \theta^0)$ involves matrix calculations involving $p + \sum_{k = 1}^K G_k$ dimensional matrices, which can be computationally costly. Having said that, as implied in the proof, the lower-left and upper-right submatrices of the asymptotic covariance matrix of $\sqrt{N} (\hat \theta - \theta^0)$ should be of little magnitude when the numbers of the groups $G_k \ (k = 1, \dots, K)$ are moderately large and the random effects $\tilde a_{k} (1), \dots, \tilde a_{k} (G_k)$ are not skewed distributed. Based on this observation, a heuristic way to estimate $V$ is to treat the covariance matrix of $\sqrt{N}(\hat \theta - \theta^0)$ as being block diagonal, which leads to matrix calculation of only $p$-dimensional matrices as suggested in section \ref{subsec:inference}.

In the setting of Example \ref{exm likelihood}, of particular importance is the generalized linear model (GLM):
\begin{equation}
f(y|x, a; \beta, \psi) = \exp\left\{ \frac{y \mu (x^{\intercal} \beta + a) - \nu(x^{\intercal} \beta + a)}{\psi} + \chi(y, \psi) \right\}, \label{glm}
\end{equation}
for some known functions $\mu(\cdot), \nu (\cdot)$, and $\chi (\cdot)$. In this model, the objective function is 
\[M(y, x | \beta, a) = y \mu (x^{\intercal} \beta + a) - \nu (x^{\intercal} \beta + a).\]
In the following, we provide three examples of GLMs.
\begin{exm}[Normal linear regression]  \label{exm glm normal}
    Consider the following density with respect to Lebesgue measure: $f(y|x, a; \beta, \psi) = (2 \pi \psi)^{-1/2} \exp [-\{ y - (x^{\intercal} \beta + a)\}^2/(2\psi)]$.
    Within the framework of \eqref{glm}, this density can be described with $\mu (u) = u, \nu (u) = u^2/2$, and $\chi (y, \psi) = -(y^2/\psi + \log (2 \pi \psi))/2$. Accordingly, the objective function is $M(y, x|\beta, a) = y (x^{\intercal} \beta + a) - (x^{\intercal} \beta + a)^2/2$.
\end{exm}

\begin{exm}[Logistic regression] \label{exm glm logistic}
Consider the following density with respect to a counting measure on $\{0, 1 \}$: $f(y|x, a; \beta) = L(x^{\intercal}\beta + a)^y (1 - L(x^{\intercal} \beta + a))^{1-y}$ for $L(t) = \exp(t)/(1+\exp(t))$. Within the framework of \eqref{glm}, this density can be described with $\mu (u) = u$, $\nu (u) = \log (1 + \exp(u))$, $\chi(y, \psi) = 0$, and $\psi = 1$. Accordingly, the objective function is $M(y, x|\beta, a) = y(x^{\intercal} \beta + a) - \log(1 + \exp(x^{\intercal}\beta + a))$.
\end{exm}

\begin{exm}[Poisson regression] \label{exm glm poisson} Consider the following density with respect to a counting measure on $\mathbb N$: $f (y|x, a; \beta) = \exp\{ y(x^{\intercal} \beta + a)\} \exp\{-\exp(x^{\intercal} \beta + a)\}/y !$. Within the framework of \eqref{glm}, this density can be described with $\mu (u) = u$, $\nu (u) = \exp(u)$, $\chi(y, \psi) = - \log y!$, and $\psi = 1$. Accordingly, the objective function is $M(y, x|\beta, a) = y(x^{\intercal}\beta + a) -\exp(x^{\intercal} \beta + a)$. 
\end{exm}
In addition to some technical requirements, the boundedness of covariates $x_{i_1 \dots i_K}$ and some condition on the Fisher information matrices would be sufficient for Theorem \ref{prop asy dist} for the above three classes of GLMs. The detail is discussed in online Appendix D.

\section{Simulation Study}\label{sec:sim}

\subsection{Binary response with two-way effects}\label{sec:binary}

We evaluate the numerical performance of the proposed method by comparing it with existing random effects approaches in terms of estimation accuracy and computation time. Specifically, we consider five covariates, $x_q = (x_{q1}, \ldots, x_{q5})^\top$, where each element is independently generated as $x_{qk} \sim N(0, 1)$ for $k = 1, \ldots, 5$ and $i = 1, \ldots, N$. Cluster indicators $\ell_{a, q} \in \{1, \ldots, n\}$ and $\ell_{b, q} \in \{1, \ldots, m\}$ are also randomly assigned.

In the experiment, we set $m = n = \lfloor N^{1/2} \rfloor$ and consider five different sample sizes: $N \in \{5000, 10000, 20000, 40000, 80000\}$. 
For each setting, we generate data from the logistic model with crossed effects, given by the general form~(\ref{eq:general-model}):
$$
y_q \sim \mathrm{Ber}\left( \frac{e^{\psi_q}}{1 + e^{\psi_q}} \right), \quad \psi_q = \alpha + x_q^\top \beta + a_{\ell_{a,q}} + b_{\ell_{b,q}}, \quad q = 1, \ldots, N,
$$
where we set $\alpha = 1$ and $\beta = (-1, 0.5, 0, 0, 0)$. The random effects $a_{\ell_{a,q}}$ and $b_{\ell_{b,q}}$ follow one of the following two scenarios:
\begin{align*}
\text{(Scenario 1)} \quad & a_1, \ldots, a_n \sim N(0, 0.5^2), \quad b_1, \ldots, b_m \sim N(0, 1), \\
\text{(Scenario 2)} \quad & a_1 + 1, \ldots, a_n + 1 \sim \mathrm{Ga}(1, 1), \quad 1 - b_1, \ldots, 1 - b_m \sim \mathrm{Ga}(1, 1).
\end{align*}
Note that in both scenarios, the random effects have zero mean. Scenario~2 introduces skewed distributions, and thus the normality assumption typically made for random effects is misspecified.

For each dataset, we apply the proposed logistic model with crossed discretized effects (CDE), using $G = \lfloor \sqrt{n} \rfloor$, $H = \lfloor \sqrt{m} \rfloor$, and regularization parameter $\lambda = 10^2$, as described in Section~2.3. For comparison, we also apply standard logistic mixed models with crossed random effects using various existing estimation methods. Specifically, we used the R package \verb+lme4+ \citep{bates2015fitting} to compute estimators via penalized quasi-likelihood (MLE0) and Laplace approximation to the marginal likelihood (MLE1). 
We also employ the R package \verb+glmmTMB+ \citep{mcgillycuddy2025parsimoniously}, which implements the Laplace approximation, and the backfitting (BF) algorithm by \citet{ghosh2022backfitting} for fast penalized quasi-likelihood estimation.

For each $N$, we assess the actual computation time and the performance of the estimation and inference on the regression coefficient $\beta$.
Based on $R=100$ replications, we compute the mean squared error (MSE), $R^{-1}\sum_{r=1}^R (\hat{\beta}_k^{(r)}-\beta_k)^2$ and coverage probability (CP) of 95\% confidence intervals, $R^{-1}\sum_{r=1}^R I(\beta_k\in {\rm CI}_k^{(r)})$, where $\beta_k^{(r)}$ is the true value, $\hat{\beta}_k^{(r)}$ is a point estimate, and ${\rm CI}_k^{(r)}$ is a  95\% confidence interval at the $r$th replication.
The averaged values over all coefficients, $k=1,\ldots,p$, are reported in Table~\ref{tab:sim-bin}. 
Note that the BF method does not provide interval estimation results, and CP of BF is not reported.
Figure~\ref{fig:sim-bin} displays the average computation time across replications, along with error bars indicating standard errors. As shown in the figure, the BF and MLE0 methods are considerably faster than the others. However, as seen in Table~\ref{tab:sim-bin}, their estimation accuracy is notably worse. In particular, MLE0 fails to achieve the nominal coverage level, indicating poor performance in interval estimation.
In contrast, CDE, MLE1, and TMB exhibit comparable levels of estimation accuracy and interval coverage. Among these, the proposed CDE method demonstrates a clear advantage in terms of computation time.

\begin{figure}[htbp!]
\centering
\includegraphics[width=\linewidth]{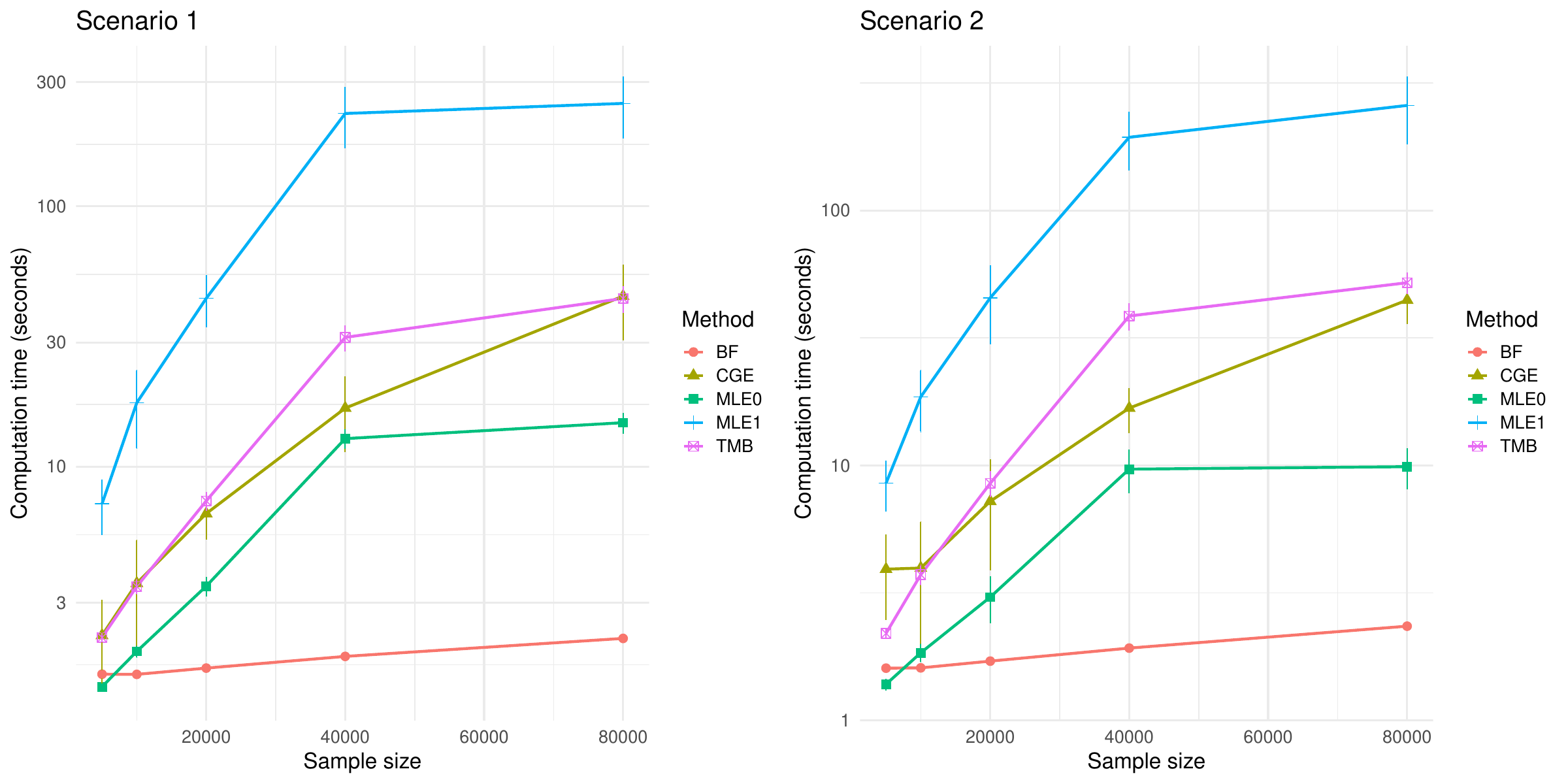}
\caption{The actual computation time (seconds) of the five methods under two-way logistic mixed models. } 
\label{fig:sim-bin}
\end{figure}

% Table
\begin{table}[htbp]
\centering
\begin{tabular}{cccccccccccccccc}
\hline
&&&  \multicolumn{5}{c}{Scenario 1} && \multicolumn{5}{c}{Scenario 2} \\
 \multicolumn{2}{r}{\footnotesize $N(\times 10^3)$} &  & {\footnotesize CDE} & {\footnotesize MLE1} & {\footnotesize MLE0} & {\footnotesize TMB} & {\footnotesize BF} &  & {\footnotesize CDE} & {\footnotesize MLE1} & {\footnotesize MLE0} & {\footnotesize TMB} & {\footnotesize BF} \\
\hline
 & 5 &  & 1.75 & 1.66 & 1.88 & 1.66 & 1.83 &  & 1.56 & 1.55 & 1.92 & 1.56 & 1.82 \\
 & 10 &  & 0.64 & 0.63 & 0.90 & 0.63 & 0.86 &  & 0.72 & 0.73 & 0.94 & 0.73 & 0.89 \\
MSE & 20 &  & 0.33 & 0.33 & 0.41 & 0.33 & 0.40 &  & 0.33 & 0.34 & 0.48 & 0.34 & 0.46 \\
 & 40 &  & 0.19 & 0.19 & 0.27 & 0.19 & 0.27 &  & 0.19 & 0.18 & 0.28 & 0.18 & 0.27 \\
 & 80 &  & 0.09 & 0.09 & 0.15 & 0.09 & 0.14 &  & 0.10 & 0.09 & 0.14 & 0.09 & 0.14 \\
 \hline
 & 5 &  & 93.0 & 94.4 & 92.8 & 94.4 & - &  & 93.0 & 94.8 & 91.8 & 95.0 & - \\
 & 10 &  & 95.2 & 96.0 & 92.8 & 96.2 & - &  & 94.8 & 95.2 & 93.0 & 95.2 & - \\
CP & 20 &  & 94.4 & 96.0 & 94.6 & 96.0 & - &  & 96.2 & 96.8 & 93.0 & 97.0 & - \\
 & 40 &  & 92.6 & 93.4 & 89.0 & 93.6 & - &  & 94.0 & 94.6 & 89.8 & 95.0 & - \\
 & 80 &  & 93.2 & 94.0 & 89.2 & 94.2 & - &  & 94.4 & 95.2 & 88.8 & 95.6 & - \\
\hline
\end{tabular}
\caption{Mean squared errors (MSE) of the point estimates and coverage probability (CP) of $95\%$ confidence intervals for $\beta$ under two-way logistic mixed models.}
\label{tab:sim-bin}
\end{table}

\subsection{Count response with three-way effects}
We further investigate the performance of the proposed method in the context of a three-way Poisson mixed model. Similar to the setup in the previous section, we generate covariates $x_q = (x_{q1}, \ldots, x_{q5})^\top$, where each element is independently sampled as $x_{qk} \sim N(0, 1)$ for $k = 1, \ldots, 5$ and $q = 1, \ldots, N$. In addition, we assign three independent clustering indicators: $\ell_{a,q} \in \{1, \ldots, n_a\}$, $\ell_{b,q} \in \{1, \ldots, n_b\}$, and $\ell_{c,q} \in \{1, \ldots, n_c\}$, drawn uniformly at random.
In this experiment, we set $n_a = n_b = n_c = 2\lfloor N^{1/2} \rfloor$ and considered five different sample sizes: $N \in \{2500, 5000, 10000, 20000, 40000\}$. For each setting, data are generated from the following three-way Poisson mixed model:
$$
y_q \sim \mathrm{Poisson}\left( \exp(\psi_q) \right), \quad \psi_q = \alpha + x_q^\top \beta + a_{\ell_{a,q}} + b_{\ell_{b,q}} + c_{\ell_{c,q}}, \quad q = 1, \ldots, N,
$$
where $\alpha = 1$ and $\beta = (-0.3, 0.3, 0, 0, 0)$. 
The three random effects are specified under the following two scenarios:
\begin{align*}
\text{(Scenario 1)} \quad & a_1, \ldots, a_{n_a} \sim N(0, (0.2)^2), \quad b_1, \ldots, b_{n_b} \sim N(0, (0.3)^2), \\
& c_1, \ldots, c_{n_c} \sim N(0, (0.3)^2), \\
\text{(Scenario 2)} \quad & a_1 + \frac{1}{5}, \ldots, a_{n_a} + \frac{1}{5} \sim \mathrm{Ga}(1, 5), \quad \frac{1}{5} - b_1, \ldots, \frac{1}{5} - b_{n_b} \sim \mathrm{Ga}(1, 1/5), \\
& c_1, \ldots, c_{n_c} \sim 0.5N(-0.3, (0.15)^2) + 0.5 N(0.3, (0.15)^2).
\end{align*}
As in the two-way setting, the expectations of the random effects are zero in both scenarios. Scenario 2 represents a case with misspecified, skewed, and bimodal random effects.

For each dataset, we apply the Poisson model with the proposed crossed discretized effects (CDE), using the group numbers, $\lfloor n_a^{1/2} \rfloor$, $\lfloor n_b^{1/2} \rfloor$ and $\lfloor n_c^{1/2} \rfloor$, and regularization parameter $\lambda = 10^2$, following the formulation in Section~2.4. For comparison, we also apply standard Poisson mixed models with three-way crossed random effects estimated using existing methods. These include penalized quasi-likelihood (MLE0) and Laplace approximation (MLE1) from the R package \verb+lme4+, and the Laplace-based estimator from \verb+glmmTMB+.
Each Monte Carlo experiment is repeated 100 times. As in Section~\ref{sec:binary}, we evaluate the mean squared error (MSE) of the point estimates of $\beta$ and the empirical coverage probability of the 95\% confidence intervals, averaged across all coefficients $\beta_k$, $k = 1, \ldots, p$. 
Figure~\ref{fig:sim-po} summarizes the average computation times and associated standard errors across the 100 replications. 
In this case, the computation time of the proposed CDE method is minimum among the four comparative methods in both scenarios. 
In Table~\ref{tab:sim-po}, we show MSE of the point estimates and CP of $95\%$ confidence intervals, which indicates that the estimation accuracy of the four methods is quite comparable, and CPs are around the nominal level. 
Hence, the proposed CDE method provides comparable point and interval estimation accuracy to existing methods, while offering advantages in terms of computational efficiency.

%  Figure 
\begin{figure}[htbp!]
\centering
\includegraphics[width=\linewidth]{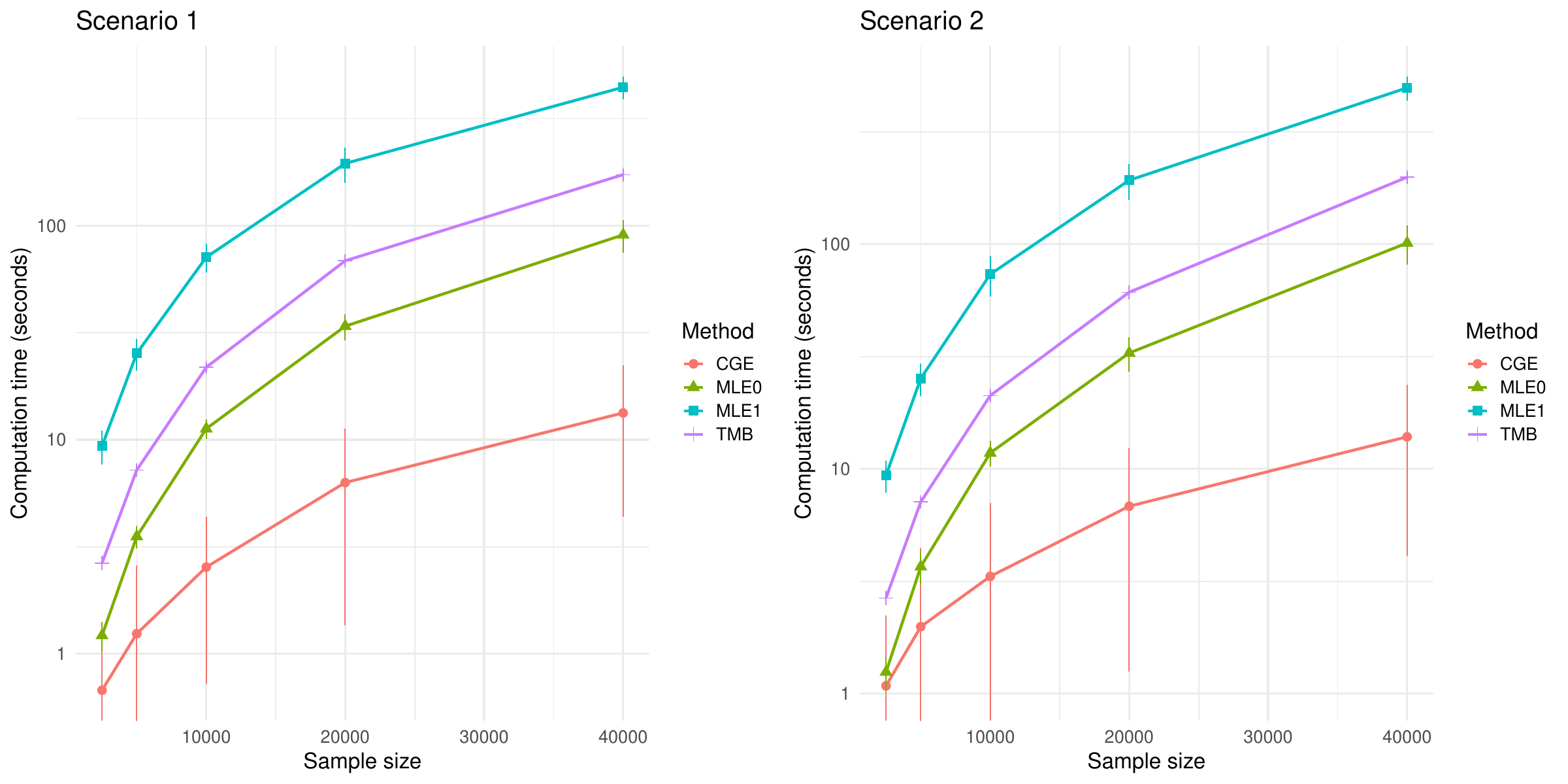}
\caption{The actual computation time (seconds) of the four methods under three-way Poisson mixed models. } 
\label{fig:sim-po}
\end{figure}

% Table
\begin{table}[htbp]
\centering
\begin{tabular}{cccccccccccccccc}
\hline
&&&  \multicolumn{4}{c}{Scenario 1} && \multicolumn{4}{c}{Scenario 2} \\
 \multicolumn{2}{r}{\footnotesize $N(\times 10^3)$} &  & {\footnotesize CDE} & {\footnotesize MLE1} & {\footnotesize MLE0} & {\footnotesize TMB} &  & {\footnotesize CDE} & {\footnotesize MLE1} & {\footnotesize MLE0} & {\footnotesize TMB} & \\
 \hline
 & 2.5 &  & 1.31 & 1.26 & 1.26 & 1.26 &  & 1.41 & 1.36 & 1.36 & 1.36 \\
 & 5 &  & 0.62 & 0.60 & 0.60 & 0.60 &  & 0.78 & 0.75 & 0.75 & 0.75 \\
MSE & 10 &  & 0.36 & 0.35 & 0.35 & 0.35 &  & 0.32 & 0.30 & 0.30 & 0.30 \\
 & 20 &  & 0.17 & 0.16 & 0.16 & 0.16 &  & 0.19 & 0.18 & 0.18 & 0.18 \\
 & 40 &  & 0.09 & 0.09 & 0.09 & 0.09 &  & 0.09 & 0.09 & 0.09 & 0.09 \\
  \hline
 & 2.5 &  & 92.0 & 96.0 & 96.2 & 96.2 &  & 91.6 & 94.4 & 94.6 & 94.6 \\
 & 5 &  & 93.4 & 95.4 & 95.6 & 95.4 &  & 91.8 & 93.4 & 93.4 & 93.4 \\
CP & 10 &  & 91.6 & 93.4 & 93.4 & 93.4 &  & 94.6 & 95.2 & 95.0 & 95.2 \\
 & 20 &  & 92.4 & 94.6 & 94.4 & 94.6 &  & 90.0 & 92.6 & 92.6 & 92.6 \\
 & 40 &  & 93.2 & 94.6 & 94.6 & 94.6 &  & 92.6 & 94.0 & 94.0 & 94.0 \\
\hline
\end{tabular}
\caption{Mean squared errors (MSE) of the point estimates and coverage probability (CP) of $95\%$ confidence intervals for $\beta$ under three-way Poisson mixed models.}
\label{tab:sim-po}
\end{table}

\subsection{Additional simulation studies}

We conducted additional simulation studies regarding configurations of the proposed method under both two-way binary and three-way Poisson settings.
Specifically, we examined the sensitivity to the choice of group numbers around the heuristic $(\sqrt{n}, \sqrt{m})$, and evaluated both estimation accuracy and computational cost across a range of settings. 
The results show that estimation accuracy is stable over a wide range of group numbers, while excessively large group numbers lead to increased computational cost with little gain in performance. 
We also explored adaptive selection strategies based on information criteria, but found that they offer limited practical benefit.
In addition, we investigated inference based on a smoothed bootstrap approach using weighted likelihood minimization. 
This approach can be implemented with only minor modifications of the proposed algorithm and does not require explicit derivation of asymptotic variance formulas. 
Simulation results indicate that the bootstrap method provides slightly improved coverage accuracy in finite samples, while yielding similar inference to asymptotic approximations in large samples.
Detailed results are reported in the Supplementary Material.

\section{Real Data Application}\label{sec:app}

We apply the proposed method together with existing methods to the MovieLens 100K dataset \citep{harper2015movielens}, which is publicly available from the GroupLens research group at the University of Minnesota (\url{https://grouplens.org/datasets/movielens/100k/}) and comprises 100,000 ratings provided by 943 users on 1,682 movies.
The outcome of the dataset is a rating on a scale of 1 to 5.
As covariates, we employ variables including 20 dummy variables for user occupation, 18 dummy variables for movie genres (with multiple genres possibly assigned to each movie), a dummy variable for gender, and age.
Age is modeled using a second-order P-spline with knots placed at the 10th, 20th, ..., and 90th percentiles of the age distribution.
In total, 50 covariates are used in the model.

To model the ratings by the available covariates and potential heterogeneous effects across users and movies, we consider the ordered probit model with crossed effects.
To our knowledge, no scalable algorithm for the ordered probit model with crossed random effects is available, and standard maximum likelihood estimation of the ordered probit model with crossed random effects would be infeasible under large sample sizes such as in this example.
Hence, we apply the proposed approach to fitting the ordered probit model with two-way effects (user and movie effects), and also apply the standard probit regression without such effects for comparison.
Let $y_q$ be the ordered outcome taking values in ${1,\ldots,K}$ ($K=5$ in this example), and $x_q$ be a $p$-dimensional vector of covariates ($p=50$ in this case).
Further, let $\ell_{a,q}\in {1,\ldots,n}$ and $\ell_{b,q}\in{1,\ldots,m}$ be indicators of users and movies, where $n=943$ and $m=1682$.
We then consider the following ordered probit model:
\begin{equation}\label{probit-model}
{\rm P}(y_q=k|x_q)=\Phi(c_k-\eta_q)-\Phi(c_{k-1}-\eta_q), \quad \eta_q=x_q^\top \beta + a_{\ell_{a,q}} + b_{\ell_{b,q}}
\end{equation}
for $k=1,\ldots,K$, where $\Phi(\cdot)$ is the standard normal distribution function, $a_{\ell_{ai}}$ and $b_{\ell_{bi}}$ are user and movie effects, respectively, and $\beta$ is a vector of regression coefficients.
Here, $-\infty=c_0<c_1<\cdots < c_{K-1}< c_K=\infty$ is a set of thresholds.
We first apply the standard ordered probit regression without user and movie effects (denoted by WoE), corresponding to a model without $a_{\ell_{ai}}$ and $b_{\ell_{bi}}$ in (\ref{probit-model}), by using the R package \verb+ordinal+ \citep{christense2023}.
We then apply the proposed crossed discretized effects (CDE) model, assuming $a_{\ell_{a,q}}=a(g_{\ell_{a,q}})$ and $b_{\ell_{b,q}}=b(h_{\ell_{b,q}})$, where $\ell_{a,q}\in \{1,\ldots,G\}$ and $\ell_{b,q}\in \{1,\ldots,H\}$ are unknown indicators.
We set the unknown thresholds $c_1,\ldots,c_{K-1}$ to the values obtained from the results of the standard ordered probit regression.
We provide details of the iterative algorithm for the CDE model in the Supplementary Material.
In this example, we set $(G, H)=(\lfloor \sqrt{n}\rfloor, \lfloor \sqrt{m}\rfloor)=(30, 41)$ (denoted by CDE1) and $(G, H)=(20, 20)$ (denoted by CDE2).
For comparison, we also fitted the crossed random effects (RE) model via the maximum likelihood method using the Laplace approximation, as implemented in the TMB framework \citep{kristensen2016tmb}, where the thresholds are fixed to the same values as CDE.

The computation times for CDE1 and CDE2 are 92 and 152 seconds, respectively, while that for WoE is 9 seconds and that for RE is 530 seconds.
This indicates that the proposed method is feasible within a reasonable amount of time even under such a large sample size, and the computation time is considerably smaller than RE. 
Since the results of CDE1 and CDE2 were quite similar, we only show the results of CDE1 in what follows. 
In Figure~\ref{fig:app-reg}, we present the estimated curves of the age effects (left panel) and scatter plots of the estimated dummy effects (right panel) with 95\% confidence intervals.
These results show the regression structures obtained from CDE and WoE.
In particular, the dummy effects detected by CDE tend to be more variable than those by WoE, and the nonlinear effects of age detected by WoE are more smoothed than those by CDE.
In Figure~\ref{fig:app-RE}, we provide histograms of the estimated user and movie effects by CDE, which clearly exhibit the existence of heterogeneity among users and movies that cannot be explained by the covariates.

To compare CDE, RE and WoE, we evaluate the performance of out-of-sample prediction. 
Specifically, we randomly hold out 10,000 observations as test data and train the models on the remaining 90,000 observations.
We then compute the mean absolute error (MAE) between the posterior predictive means and the observed outcomes in the test data.
In addition, we obtain predicted ordered categories by rounding the posterior predictive means to the nearest integers and compute the classification accuracy (AC0) as well as the 1-off classification accuracy (AC1), which allows for a tolerance of one level in the predicted outcome.
This procedure is repeated 20 times, and we report the average and standard error of each evaluation metric.
The results are summarized in Table~\ref{tab:validation}, which shows the superior performance of the CDE to WoE. 
It is also observed that the performance of CDE1 and CDE2 is quite comparable, showing the limited effects of the settings of $G$ and $H$.

%  Figure
\begin{figure}[htbp!]
\centering
\includegraphics[width=\linewidth]{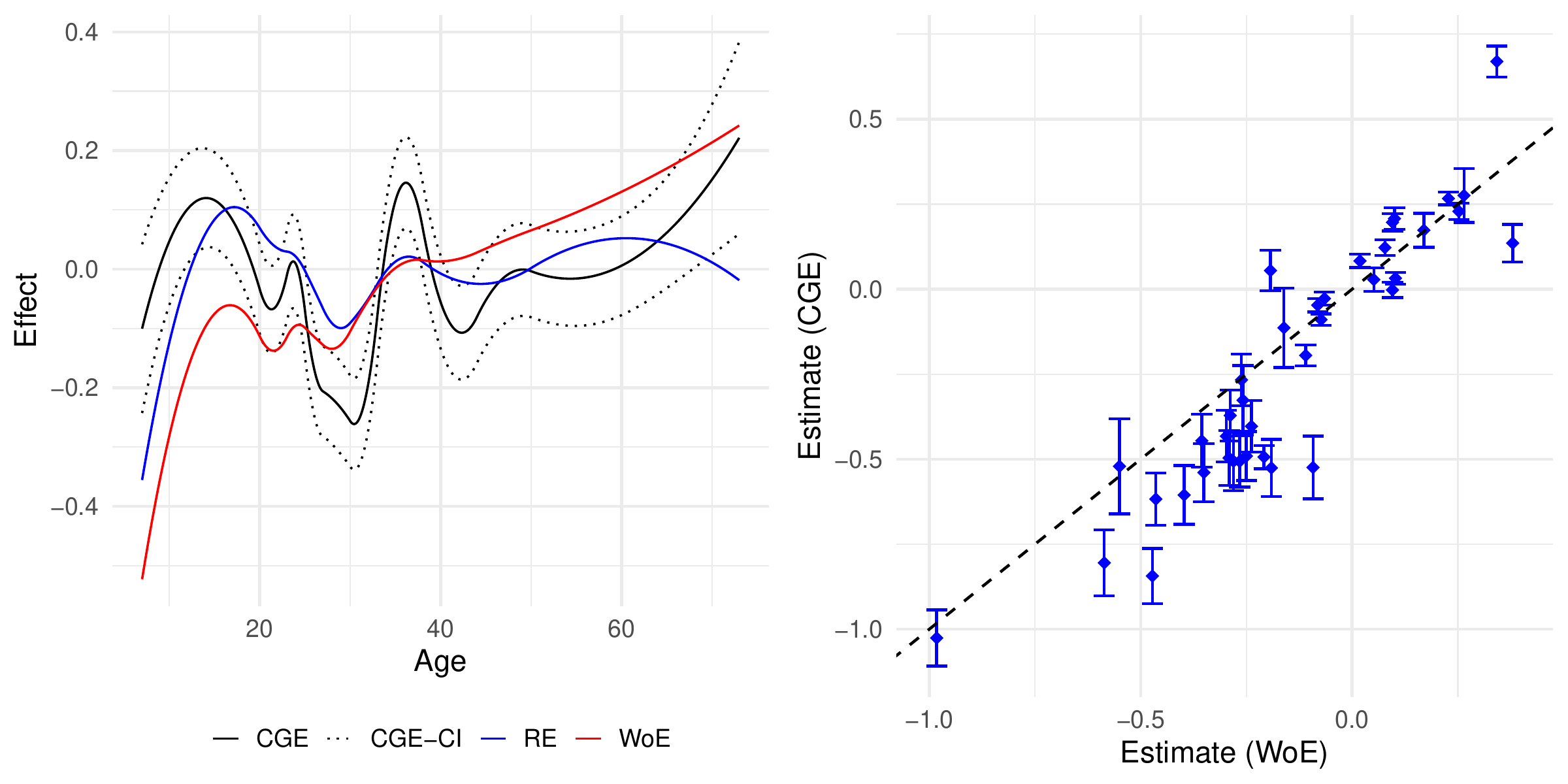}
\caption{Estimated regression function of age (left) and scatter plots of estimates of regression coefficients for dummy variables with $95\%$ confidence intervals for CDE (right). } 
\label{fig:app-reg}
\end{figure}

%  Figure
\begin{figure}[htbp!]
\centering
\includegraphics[width=\linewidth]{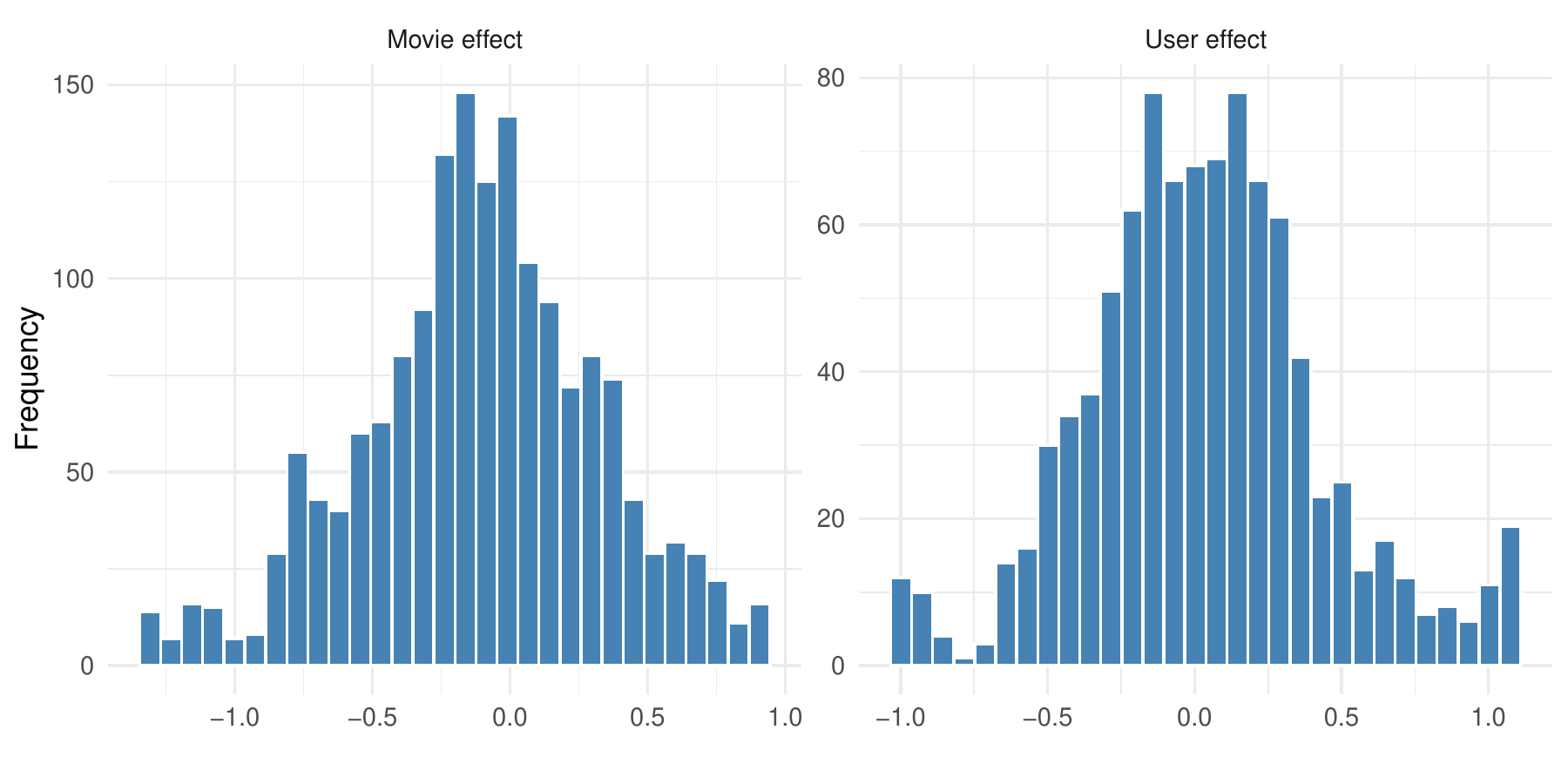}
\caption{Histograms of estimated user and movie effects. } 
\label{fig:app-RE}
\end{figure}

% Table

\begin{table}[htbp]
\centering
\begin{tabular}{cccccccccccc}
\hline
&& \multicolumn{2}{c}{MAE} && \multicolumn{2}{c}{AC0 (\%)} && \multicolumn{2}{c}{AC1 (\%)} \\
\hline
WoE &  & 0.846 & (0.013) &  & 34.8 & (0.76) &  & 83.6 & (0.53) \\
CDE1 &  & 0.690 & (0.009) &  & 43.2 & (0.72) &  & 89.6 & (0.33) \\
CDE2 &  & 0.690 & (0.010) &  & 43.1 & (0.80) &  & 89.6 & (0.32) \\
RE &  & 0.690 & (0.009) &  & 43.1 & (0.76) &  & 89.6 & (0.31) \\
\hline
\end{tabular}
\caption{Mean absolute errors (MAE), accuracy (AC0) and 1-off accuracy (AC1) of the ordered probit regression without individual and movie effects (WoE), with the proposed crossed discretized effects (CDE) and with the crossed random effect (RE), evaluated for the test data. 
The reported values are averaged over 20 replications of splitting the dataset. The standard deviations are given in parentheses.  }
\end{table}
\label{tab:validation}

Since the dataset contains ZIP code of users (795 categories in total), we can also incorporate the information into the model, leading to a three-way ordered probit model. 
We applied the proposed CDE1, CDE2 and RE models. 
The computation time was 115 (CDE1), 160 (CDE2) and 1378 (RE) seconds, indicating considerable computational reduction of the proposed CDE model.
It also shows the scalability of CDE since the increase of computation time for CDE is limited while that of RE is quite large. 
The detailed estimation results of CDE and RE are provided in the Supplementary Material.

\section{Discussion}\label{sec:disc}

While we focus on crossed random effects only for intercept terms, the proposed model (\ref{model-proposal}) can be extended to crossed random slope model, described as $y_{ij}|x_{ij}\sim f(y_{ij}; x_{ij}^\top a(g_i)+x_{ij}^\top b(h_j))$, where $a_{g}$ and $b_h$ are $p$-dimensional vectors of group-wise regression coefficients and $g_i$ and $h_j$ are unknown grouping parameters. 
While fitting the standard crossed random slope models may be computationally infeasible, the above grouped model can be fitted using an algorithm similar to Algorithm~\ref{algo1}, and would be scalable under large samples. 
The detailed investigation is left to a future study. 
As a potential extension of the proposed method, it would be possible to incorporate the auxiliary information (e.g. location information) into the clustering step, as considered in \cite{sugasawa2021grouped} and \cite{sugasawa2021spatially}, such that the observations at closer locations are more likely to be classified to the same group.

\section*{Acknowledgement}
This work is partially supported by JSPS KAKENHI Grant Numbers 24K21420 and 25H00546.

\section*{Data Availability Statement }
The data that support the findings of this study are openly available from the GroupLens research group at the University of Minnesota (\url{https://grouplens.org/datasets/movielens/100k/})

\vspace{1cm}
%   Reference
\bibliographystyle{chicago}
\bibliography{ref}

\begin{thebibliography}{}

\bibitem[\protect\citeauthoryear{Allenby and Rossi}{Allenby and Rossi}{1998}]{allenby1998marketing}
Allenby, G.~M. and P.~E. Rossi (1998).
\newblock {Marketing Models of Consumer Heterogeneity}.
\newblock {\em Journal of Econometrics\/}~{\em 89\/}(1-2), 57--78.
\newblock DOI: 10.1016/S0304-4076(98)00055-4.

\bibitem[\protect\citeauthoryear{Bates, M{\"a}chler, Bolker, and Walker}{Bates et~al.}{2015}]{bates2015fitting}
Bates, D., M.~M{\"a}chler, B.~Bolker, and S.~Walker (2015).
\newblock {Fitting Linear Mixed-Effects Models Using lme4}.
\newblock {\em Journal of Statistical Software\/}~{\em 67\/}(1), 1--48.
\newblock DOI: 10.18637/jss.v067.i01.

\bibitem[\protect\citeauthoryear{Bellio, Ghosh, Owen, and Varin}{Bellio et~al.}{2025}]{bellio2023scalable}
Bellio, R., S.~Ghosh, A.~B. Owen, and C.~Varin (2025).
\newblock {Consistent and Scalable Composite Likelihood Estimation of Probit Models with Crossed Random Effects}.
\newblock {\em Biometrika\/}.
\newblock DOI: 10.1093/biomet/asaf037.

\bibitem[\protect\citeauthoryear{Bonhomme and Manresa}{Bonhomme and Manresa}{2015}]{bonhomme2015grouped}
Bonhomme, S. and E.~Manresa (2015).
\newblock {Grouped Patterns of Heterogeneity in Panel Data}.
\newblock {\em Econometrica\/}~{\em 83\/}(3), 1147--1184.
\newblock DOI: 10.3982/ECTA11319.

\bibitem[\protect\citeauthoryear{Browne, Subramanian, Jones, and Goldstein}{Browne et~al.}{2005}]{browne2005variance}
Browne, W.~J., S.~V. Subramanian, K.~Jones, and H.~Goldstein (2005).
\newblock {Variance Partitioning in Multilevel Logistic Models that Exhibit Overdispersion}.
\newblock {\em Journal of the Royal Statistical Society Series A: Statistics in Society\/}~{\em 168\/}(3), 599--613.
\newblock DOI: 10.1111/j.1467-985X.2004.00365.x.

\bibitem[\protect\citeauthoryear{Chernozhukov, Chetverikov, and Kato}{Chernozhukov et~al.}{2015}]{chernozhukov2015comparison}
Chernozhukov, V., D.~Chetverikov, and K.~Kato (2015).
\newblock {Comparison and Anti-Concentration Bounds for Maxima of Gaussian Random Vectors}.
\newblock {\em Probability Theory and Related Fields\/}~{\em 162\/}(1), 47--70.
\newblock DOI: 10.1007/s00440-014-0565-9.

\bibitem[\protect\citeauthoryear{Christensen}{Christensen}{2023}]{christense2023}
Christensen, R. H.~B. (2023).
\newblock {\em ordinal---Regression Models for Ordinal Data}.
\newblock R package version 2023.12-4.1.

\bibitem[\protect\citeauthoryear{Fern{\'a}ndez-Val and Weidner}{Fern{\'a}ndez-Val and Weidner}{2016}]{fernandez2016individual}
Fern{\'a}ndez-Val, I. and M.~Weidner (2016).
\newblock {Individual and Time Effects in Nonlinear Panel Models with Large N, T}.
\newblock {\em Journal of Econometrics\/}~{\em 192\/}(1), 291--312.
\newblock DOI: 10.1016/j.jeconom.2015.12.014.

\bibitem[\protect\citeauthoryear{Ghandwani, Ghosh, Hastie, and Owen}{Ghandwani et~al.}{2023}]{ghandwani2023scalable}
Ghandwani, D., S.~Ghosh, T.~Hastie, and A.~B. Owen (2023).
\newblock {Scalable Solution to Crossed Random Effects Model with Random Slopes}.
\newblock {\em arXiv:2307.12378\/}.

\bibitem[\protect\citeauthoryear{Ghosh, Hastie, and Owen}{Ghosh et~al.}{2022a}]{ghosh2022backfitting}
Ghosh, S., T.~Hastie, and A.~B. Owen (2022a).
\newblock {Backfitting for Large Scale Crossed Random Effects Regressions}.
\newblock {\em Annals of Statistics\/}~{\em 50\/}(1), 560--583.
\newblock DOI: 10.1214/21-AOS2121.

\bibitem[\protect\citeauthoryear{Ghosh, Hastie, and Owen}{Ghosh et~al.}{2022b}]{ghosh2022scalable}
Ghosh, S., T.~Hastie, and A.~B. Owen (2022b).
\newblock {Scalable Logistic Regression with Crossed Random Effects}.
\newblock {\em Electronic Journal of Statistics\/}~{\em 16\/}(2), 4604--4635.
\newblock DOI: 10.1214/22-EJS2047.

\bibitem[\protect\citeauthoryear{Ghosh and Zhong}{Ghosh and Zhong}{2021}]{ghosh2021convergence}
Ghosh, S. and C.~Zhong (2021).
\newblock Convergence rate of a collapsed gibbs sampler for crossed random effects models.
\newblock {\em arXiv preprint arXiv:2109.02849\/}.

\bibitem[\protect\citeauthoryear{Gu and Volgushev}{Gu and Volgushev}{2019}]{GuVolgushev2019joe}
Gu, J. and S.~Volgushev (2019).
\newblock {Panel Data Quantile Regression with Grouped Fixed Effects}.
\newblock {\em Journal of Econometrics\/}~{\em 213\/}(1), 68--91.
\newblock DOI: 10.1016/j.jeconom.2019.04.006.

\bibitem[\protect\citeauthoryear{Hahn and Moon}{Hahn and Moon}{2010}]{hahn2010panel}
Hahn, J. and H.~R. Moon (2010).
\newblock {Panel Data Models with Finite Number of Multiple Equilibria}.
\newblock {\em Econometric Theory\/}~{\em 26\/}(3), 863--881.
\newblock DOI: 10.1017/S0266466609990132.

\bibitem[\protect\citeauthoryear{Harper and Konstan}{Harper and Konstan}{2015}]{harper2015movielens}
Harper, F.~M. and J.~A. Konstan (2015).
\newblock {The MovieLens Datasets: History and Context}.
\newblock {\em ACM Transactions on Interactive Intelligent Systems (TiiS)\/}~{\em 5\/}(4), 1--19.
\newblock DOI: http://dx.doi.org/10.1145/2827872.

\bibitem[\protect\citeauthoryear{Ito and Sugasawa}{Ito and Sugasawa}{2023}]{ito2023grouped}
Ito, T. and S.~Sugasawa (2023).
\newblock {Grouped Generalized Estimating Equations for Longitudinal Data Analysis}.
\newblock {\em Biometrics\/}~{\em 79\/}(3), 1868--1879.
\newblock DOI: 10.1111/biom.13718.

\bibitem[\protect\citeauthoryear{Jeon, Rijmen, and Rabe-Hesketh}{Jeon et~al.}{2017}]{jeon2017variational}
Jeon, M., F.~Rijmen, and S.~Rabe-Hesketh (2017).
\newblock {A Variational Maximization--Maximization Algorithm for Generalized Linear Mixed Models with Crossed Random Effects}.
\newblock {\em Psychometrika\/}~{\em 82\/}(3), 693--716.
\newblock DOI: 10.1007/s11336-017-9555-z.

\bibitem[\protect\citeauthoryear{Jiang}{Jiang}{1998}]{Jiang1998jasa}
Jiang, J. (1998).
\newblock Consistent estimators in generalized linear mixed models.
\newblock {\em Journal of the American Statistical Association\/}~{\em 93\/}(442), 720--729.

\bibitem[\protect\citeauthoryear{Jiang}{Jiang}{2013}]{Jiang2013aos}
Jiang, J. (2013).
\newblock {The Subset Argument and Consistency of {MLE} in {GLMM}: Answer to an Open Problem and Beyond}.
\newblock {\em Annals of Statistics\/}~{\em 41\/}(1), 177--195.
\newblock DOI: 10.1214/13-AOS1084.

\bibitem[\protect\citeauthoryear{Jiang}{Jiang}{2025}]{Jiang2025aos}
Jiang, J. (2025).
\newblock {Asymptotic Distribution of Maximum Likelihood Estimator in Generalized Linear Mixed Models with Crossed Random Effects}.
\newblock {\em Annals of Statistics\/}~{\em 53\/}(3), 1298--1318.
\newblock DOI: 10.1214/25-AOS2504.

\bibitem[\protect\citeauthoryear{Jiang and Zhang}{Jiang and Zhang}{2001}]{JiangZhang2001robust}
Jiang, J. and W.~Zhang (2001).
\newblock Robust estimation in generalised linear mixed models.
\newblock {\em Biometrika\/}~{\em 88\/}(3), 753--765.

\bibitem[\protect\citeauthoryear{Kato}{Kato}{2019}]{katoep}
Kato, K. (2019).
\newblock {\em Lecture Notes on Empirical Process Theory}.
\newblock Available at https://sites.google.com/site/kkatostat/home/research.

\bibitem[\protect\citeauthoryear{Kristensen, Nielsen, Berg, Skaug, and Bell}{Kristensen et~al.}{2016}]{kristensen2016tmb}
Kristensen, K., A.~Nielsen, C.~W. Berg, H.~Skaug, and B.~M. Bell (2016).
\newblock Tmb: automatic differentiation and laplace approximation.
\newblock {\em Journal of statistical software\/}~{\em 70}, 1--21.

\bibitem[\protect\citeauthoryear{Liu, Shang, Zhang, and Zhou}{Liu et~al.}{2020}]{liu2020identification}
Liu, R., Z.~Shang, Y.~Zhang, and Q.~Zhou (2020).
\newblock {Identification and Estimation in Panel Models with Overspecified Number of Groups}.
\newblock {\em Journal of Econometrics\/}~{\em 215\/}(2), 574--590.
\newblock DOI:10.1016/j.jeconom.2019.09.008.

\bibitem[\protect\citeauthoryear{Lyu, Sisson, and Welsh}{Lyu et~al.}{2024}]{LyuSissonWelsh2024aos}
Lyu, Z., S.~A. Sisson, and A.~H. Welsh (2024).
\newblock {Increasing Dimension Asymptotics for Two-way Crossed Mixed Effect Models}.
\newblock {\em Annals of Statistics\/}~{\em 52\/}(6), 2956--2978.
\newblock DOI: 10.1214/24-aos2469.

\bibitem[\protect\citeauthoryear{McGillycuddy, Popovic, Bolker, and Warton}{McGillycuddy et~al.}{2025}]{mcgillycuddy2025parsimoniously}
McGillycuddy, M., G.~Popovic, B.~M. Bolker, and D.~I. Warton (2025).
\newblock {Parsimoniously Fitting Large Multivariate Random Effects in glmmTMB}.
\newblock {\em Journal of Statistical Software\/}~{\em 112\/}(1), 1--19.
\newblock DOI: 10.18637/jss.v112.i01.

\bibitem[\protect\citeauthoryear{Neyman and Scott}{Neyman and Scott}{1948}]{neyman1948consistent}
Neyman, J. and E.~L. Scott (1948).
\newblock {Consistent Estimates Based on Partially Consistent Observations}.
\newblock {\em Econometrica\/}~{\em 16\/}(1), 1--32.
\newblock DOI: 10.2307/1914288.

\bibitem[\protect\citeauthoryear{Papaspiliopoulos, Roberts, and Zanella}{Papaspiliopoulos et~al.}{2020}]{papaspiliopoulos2020scalable}
Papaspiliopoulos, O., G.~O. Roberts, and G.~Zanella (2020).
\newblock Scalable inference for crossed random effects models.
\newblock {\em Biometrika\/}~{\em 107\/}(1), 25--40.

\bibitem[\protect\citeauthoryear{Robinson}{Robinson}{1991}]{Robinson1991ss}
Robinson, G.~K. (1991).
\newblock {That BLUP is a good thing: the estimation of random effects}.
\newblock {\em Statistical Science\/}~(1), 15--51.

\bibitem[\protect\citeauthoryear{Sugasawa}{Sugasawa}{2021}]{sugasawa2021grouped}
Sugasawa, S. (2021).
\newblock {Grouped Heterogeneous Mixture Modeling for Clustered Data}.
\newblock {\em Journal of the American Statistical Association\/}~{\em 116\/}(534), 999--1010.
\newblock DOI: 10.1080/01621459.2020.1777136.

\bibitem[\protect\citeauthoryear{Sugasawa and Murakami}{Sugasawa and Murakami}{2021}]{sugasawa2021spatially}
Sugasawa, S. and D.~Murakami (2021).
\newblock {Spatially Clustered Regression}.
\newblock {\em Spatial Statistics\/}~{\em 44}, 100525.
\newblock DOI: 10.1016/j.spasta.2021.100525.

\bibitem[\protect\citeauthoryear{Torabi}{Torabi}{2012}]{torabi2012likelihood}
Torabi, M. (2012).
\newblock {Likelihood Inference in Generalized Linear Mixed Models with Two Components of Dispersion using Data Cloning}.
\newblock {\em Computational Statistics \& Data Analysis\/}~{\em 56\/}(12), 4259--4265.
\newblock DOI: 10.1016/j.csda.2012.04.008.

\bibitem[\protect\citeauthoryear{van~der Vaart and Wellner}{van~der Vaart and Wellner}{2023}]{van1996weak}
van~der Vaart, A. and J.~A. Wellner (2023).
\newblock {\em Weak Convergence and Empirical Processes}.
\newblock Springer Series in Statistics. Springer, Cham.
\newblock DOI: 10.1007/978-3-031-29040-4.

\bibitem[\protect\citeauthoryear{Williams}{Williams}{1991}]{Williams1991}
Williams, D. (1991).
\newblock {\em Probability with Martingales}.
\newblock Cambridge Mathematical Textbooks. Cambridge University Press, Cambridge.
\newblock DOI: 10.1017/CBO9780511813658.

\bibitem[\protect\citeauthoryear{Xu, Reid, and Kong}{Xu et~al.}{2023}]{xu2023gaussian}
Xu, L., N.~Reid, and D.~Kong (2023).
\newblock {Gaussian Variational Approximation with Composite Likelihood for Crossed Random Effect Models}.
\newblock {\em arXiv:2310.12485\/}.

\end{thebibliography}

%------------------------------------------%
%        Supplementary Material            %
%------------------------------------------%
\clearpage
%   section number in supplement 
\setcounter{equation}{0}
\setcounter{section}{0}
\setcounter{table}{0}
\setcounter{page}{1}
\renewcommand{\thesection}{S\arabic{section}}
\renewcommand{\theequation}{S\arabic{equation}}
\renewcommand{\thetable}{S\arabic{table}}

\vspace{1cm}
\begin{center}
{\LARGE
{\bf Supplementary Material for  ``Scalable Estimation of Crossed Random Effects Models via Multi-way Discretization''}
}
\end{center}

This Supplementary Material provides assumptions for the theorems in the main text (Appendix \ref{app:assumption}), the proofs of these theorems (Appendix \ref{app:proofs}), auxiliary results required for the proof (Appendix \ref{app:lemma}), sufficient conditions for the theorems in the case of GLMs (Appendix \ref{app:glm}), the details of the maximization algorithms for the ordered probit with crossed effects (Appendix \ref{app:algorithm}), the discussion on the rate condition in Theorem 2 in the main text (Appendix \ref{app:rate}), additional simulation results (Appendix \ref{app:addition}), and the detail of additional application results (Appendix \ref{app:three-way}).

%\appendix

\begin{appendix}
    \section{Assumptions}\label{app:assumption}
    The parameter space $\Xi$ for $\Psi$ is defined to be the product, $\Xi = \mathcal B \times \mathcal A^{\prod}  \times \Gamma$ with $\mathcal A^{\prod} := \prod_{l=1}^{\sum_{k=1}^{K} G_k} \mathcal A$.
Here, $\mathcal B$ and $\mathcal A$ denote parameter spaces for $\beta$ and a component of $\alpha$, respectively, while $\Gamma$ collects all possible groupings for $\gamma$. We set the following assumptions with regard to the parameter space.
\begin{as} \label{as parameter space}
    $(a)$ $\mathcal B$ is a compact, convex subset of $\mathbb R^{p}$, where $p$ is a dimension of $\beta$ while $\mathcal A$ is a compact, convex subset of $\mathbb R$, $(b)$ $\beta^0$ lies in an interior of $\mathcal B$, and $(c)$ there exists $\ep_{a} > 0$ such that, for all $g_k = 1, \dots, G_k \ (k = 1, \dots, K)$, an $\ep_{a}$-neighborhood of $a^0_{k}(g_k)$ is included in $\mathcal A$.
\end{as}
Assumptions \ref{as parameter space}$(a)$ and $(b)$ are a standard assumption for establishing consistency in the literature. For example, see Assumption 1 of \cite{bonhomme2015grouped} and Section 2 of \cite{liu2020identification} for fixed effect estimations with grouped structure. Because the normalized parameter $a^0_{k}(g_k)$ may depend on the sample size, Assumption \ref{as parameter space}$(c)$ strengthens the usual interior condition such that $ a^0_{k}(g_k)$ is distant away from the boundary of $\mathcal A$ uniformly in the sample size.
In Example 1 in the main text, $a^0_{k} (g_k)$ is within a bounded distance of the true random effect $\tilde a_{k}(g_k)$ so that taking sufficiently large $\mathcal A$ would ensure Assumption \ref{as parameter space}$(c)$.

The following condition is about $M(z_{i_1 \dots i_K} | \beta, a)$ for Theorem 1 in the main text.
\begin{as}\label{as m}
    Let $\mathcal A^{+} := \{ a^*_1 + \dots + a^*_K : (a^*_1, \dots,  a^*_K) \in \mathcal A \times \dots \times \mathcal A \}$, then the following holds:
    $(a)$ the function $(\beta, a) \mapsto \mathbb E[M(z_{i_1 \dots i_K} | \beta, a)|\gamma^0]$ is twice continuously differentiable  with its unique maximum attained at $(\beta^0, a^0_{1}(g^0_{1, i_1}) + \dots + a^0_{K} (g^0_{K, i_K}))$ over $\mathcal B \times \mathcal A^{+}$,
    $(b)$ for the matrix $\nabla_{\beta_a \beta_a^{\intercal}} \mathbb E[M(z_{i_1 \dots i_K} | \beta, a) | \gamma^0]$, there exists $\sigma > 0$, independent of $i_1, \dots, i_K$, such that the maximum eigenvalue of the matrix is smaller than $- \sigma < 0$ for all $\beta \in \mathcal B$ and $a \in \mathcal A^{+}$, $(c)$ there exists an envelope function $M_{env}(z)$ such that, for any $\beta \in \mathcal B$ and $a \in \mathcal A^{+}$, $|M(z|\beta, a)| \leq M_{env}(z)$ holds for all $z$, and $\mathbb E [M_{env}(z_{i_1 \dots i_K})^2|\gamma^0]$ is uniformly bounded over all $i_1, \dots, i_K$, and $(d)$ for a class of functions $\mathcal M := \{ M(\cdot | \beta, a) : \beta \in \mathcal B, a \in \mathcal A^{+} \}$ and $L_2$ norm $\| \cdot \|_{R, 2}$ associated with a probability measure $R$, the covering number $\mathcal N(\ep \| M_{env} \|_{R, 2}, \mathcal M, L_2 (R))$ satisfies the uniform entropy condition, $\sup_{R} \int^{1}_0 \sqrt{1 + \log \mathcal N(\ep \| M_{env} \|_{R, 2}, \mathcal M, L_2(R))} d\ep < \infty$,
    where the supremum is take over all discrete probability measure.
\end{as}
Assumption \ref{as m}$(a)$ and $(b)$ are concerned with the concavity of $M(z_{i_1 \dots i_K}|\beta, a)$, which is easy to verify for our intended applications. Similar concavity assumptions can be found, for example, in Assumption 4.1 of \cite{fernandez2016individual}.
Assumptions \ref{as m}$(c)$ and $(d)$ are for the convergence of the sample objective function to its population counterpart uniformly over the parameter space, which is a key step for proving the consistency. For the definition of the covering number, refer to Definition 2.1.5 of \cite{van1996weak}.
In GLM examples in the main text, it is not hard to verify Assumption \ref{as m}$(d)$; see Examples \ref{exm glm normal as}, \ref{exm glm logistic as}, and \ref{exm glm poisson as} in Appendix \ref{app:glm}.

We also need the following assumptions on the distribution of the grouping structure $\gamma^0$ for Theorem 1 in the main text.
\begin{as}\label{as group N}
    For all $k = 1, \dots, K$, the following holds: $(a)$ for all $g \in \{1, \dots, G_k \}$, $\frac{1}{N_k} \sum_{i_k = 1}^{N_k} 1 \{ g^0_{k, i_k} = g \} \rightarrow \pi_{k, g}$ almost surely for some $\pi_{k, g} > 0$, and $(b)$ for all $g, g' \in \{1, \dots, G_k \}$ with $g \neq g'$, there exists $c^{k}_{g, g'} > 0$, independent of $n_1, \dots, n_K$, such that $|a^0_{k}(g) - a^0_{k}(g')| > c^k_{g, g'}$
\end{as}
These assumptions ensure that the numbers of the groups are correctly specified, and the grouped random effects are well-separated. Similar assumptions can be found in Assumption 2$(a)$ and $(b)$ of \cite{bonhomme2015grouped}. 
In Example 1 in the main text, $|a^0_{k}(g) - a^0_{k}(g')|$ equals $|\tilde a_{k}(g) - \tilde a_{k}(g')|$, which does not depend on $n_1, \dots, n_K$. Hence, in this case, $|a^0_{k}(g) - a^0_{k}(g')| > c^k_{g, g'}$ in Assumption \ref{as group N}$(b)$ can be simplified to $a^0_{k}(g) \neq a^0_{k}(g')$.
Note that \cite{liu2020identification} relax their version of Assumption \ref{as group N}$(b)$, allowing the overspecified number of groups in the context of their grouped coefficients models.

For Theorem 2 in the main text, we additionally need the following assumptions on the objective function $M(z_{i_1 \dots i_K}| \beta, a)$. We introduce the notation for the derivative matrices:
\begin{align}
    D_{i_1 \dots i_K}^{\beta_a} (\beta^*, a^*) &:= \mathbb E [\nabla_{\beta_a} M(z_{i_1 \dots i_K} | \beta^*, a^*) \nabla_{\beta_a} M(z_{i_1 \dots i_K} | \beta^*, a^*)^{\intercal}| \gamma^0], \notag \\
    J_{i_1 \dots i_K}^{\beta_a} (\beta^*, a^*) &:= \mathbb E [\nabla_{\beta_a \beta_a^{\intercal}} M(z_{i_1 \dots i_K}|\beta^*, a^*) | \gamma^0]. \notag
\end{align}
We then set the following assumptions.
\begin{as} \label{as m added}
    For all $i_1, \dots, i_K$ with $i_k = 1, \dots, n_k$ $(k = 1, \dots, K)$, the followings hold: $(a)$ the function $(\beta, a) \mapsto M(z_{i_1 \dots i_K} | \beta, a)$ is twice continuously differentiable and concave over $\mathcal B \times \mathcal A^{+}$ with all the derivatives being dominated by some function $\bar M(z_{i_1 \dots i_K})$, uniformly over $\mathcal B \times \mathcal A^{+}$, such that $\mathbb E[\bar M(z_{i_1 \dots i_K})|\gamma^0] < \infty$,
    $(b)$ for any $(\beta, a) \in \mathcal B \times \mathcal A^{+}$, there exist finite constants $\mathcal C_1$ and $\mathcal C_2$, independent of $i_1, \dots, i_K$ and $\gamma^0$, such that, for all $x > 0$, $\mathbb P(|\nabla_{a} M (z_{i_1 \dots i_K} | \beta, a)| > x | \gamma^0) \leq \mathcal C_1 \exp(-\mathcal C_2 x)$, $(c)$  $\nabla_{aa} M(z_{i_1 \dots i_K}|\beta, a)$ and $\| \nabla_{a \beta^{\intercal}} M(z_{i_1 \dots i_K}|\beta, a) \|$ are uniformly bounded over $\mathcal B \times \mathcal A^{+}$, $(d)$ there exist a finite constant $\mathcal C_3$ such that, for any $\beta_{a}^* := ( {\beta^*}^{\intercal}, a^*)^{\intercal}$ and $\bar \beta_{a} := (\bar \beta^{\intercal}, \bar a)^{\intercal}$ in $B \times A^{+}$, it holds that $\| J^{\beta_a}_{i_1 \dots i_K} (\beta^*, a^*) - J^{\beta_a}_{i_1 \dots i_K} (\bar \beta, \bar a) \| \leq \mathcal C_3 \| \beta_{a}^* - \bar \beta_{a}  \|$ uniformly over $i_1, \dots, i_K$,  $(e)$ $\mathbb E [\| \nabla_{\beta_a} M(z_{i_1 \dots i_K} | \beta^0, a^0_{1}(g^0_{1, i_1}) + \dots + a^0_{K} (g^0_{K, i_K})) \|^l | \gamma^0] $ is uniformly bounded over $i_1, \dots i_K$, for $l = 2, 3$, $(f)$ the matrices 
    \begin{align}
        &\frac{1}{N} \sum_{i_1 \dots i_K} 1 \{ g^0_{1, i_1} = g_1, \dots, g^0_{K, i_K} = g_K \} D^{\beta_a}_{i_1 \dots i_K} (\beta^0, a^0_{1}(g^0_{1, i_1}) + \dots + a^0_{K}(g^0_{K, i_k})), \notag \\
        &\frac{1}{N} \sum_{i_1 \dots i_K} 1 \{ g^0_{1, i_1} = g_1, \dots, g^0_{K, i_K} = g_K \} J^{\beta_a}_{i_1 \dots i_K} (\beta^0, a^0_{1} (g^0_{1, i_1}) + \dots + a^0_{K} (g^0_{K, i_K})), \notag
    \end{align}
    converge to some non-random positive definite $D_{g_1 \dots g_K}$, and negative definite $J_{g_1 \dots g_K}$, respectively, almost surely for all $g_k = 1, \dots, G_k$ $(k = 1, \dots, K)$, and $(g)$ for any element $M''(\cdot | \beta, a)$ of a matrix $\nabla_{\beta_a \beta_a^{\intercal}} M(\cdot | \beta, a)$, a class of functions $ \mathcal M''  := \{ M''(\cdot| \beta, a) : \beta \in \mathcal B, a \in \mathcal A^{+} \}$ satisfies $\sup_{R} \int^1_0 \sqrt{1 + \log \mathcal N(\ep \| M_{env}'' \|_{R, 2}, \mathcal M'', L_2 (R))} d \ep < \infty$, where the supremum is taken over all discrete probability measure, and $M_{env}'' (z)$ is an envelope function for $\mathcal M''$ such that $\mathbb E[M_{env}'' (z_{i_1 \dots i_K})^2 | \gamma^0] $ is  uniformly bounded over $i_1, \dots, i_K$.
\end{as}
Assumption \ref{as m added}$(a)$ allows us to apply the dominated convergence theorem, which ensures the exchangeability of the conditional expectation and the derivatives, and the continuity of the expectation of the derivatives in $\beta$ and $a$. The sub-exponentiality in Assumption \ref{as m added}$(b)$ helps bounding the expectations of maxima in the proof. The uniform entropy condition in Assumption \ref{as m added}$(f)$ can be checked similarly to Assumption \ref{as m}$(d)$. Assumption \ref{as m added}$(c)$ and $(d)$ are other technical conditions for the proofs.
\end{appendix}

\section{Proofs}\label{app:proofs}
Throughout all the proofs in Appendix \ref{app:proofs} and Appendix \ref{app:lemma}, we use the following notations. For any random vectors $w_1$ and $w_2$, we write the conditional expectation $\mathbb E[w_1 | w_2]$ as $\mathbb E_{w_2} [w_1]$.
We employ operators $P$ and $\mathbb P_N$ such that, for an array of random vectors $\{ w_{i_1 \dots i_K} : 1 \leq i_k \leq N_k, 1 \leq k \leq K \}$ on $(\Omega, \mathcal F, \mathbb P)$, we write $P w_{i_1 \dots i_K} = \frac{1}{N} \sum_{i_1 \dots i_K} \mathbb E_{\gamma^0} [w_{i_1 \dots i_K}]$, and $ \mathbb P_{N} w_{i_1 \dots i_K} = \frac{1}{N} \sum_{i_1 \dots i_K} w_{i_1 \dots i_K}$.
We also abbreviate $a_{1}(g_{1, i_1}) + \dots + a_{K}(g_{K, i_K})$ to $a_{i_1 \dots i_K}$, $a^0_{1}(g^0_{1, i_1}) + \dots + a^0_{K}(g^0_{K, i_K})$ to $ a^0_{i_1 \dots i_K}$, and $M(z_{i_1 \dots i_K} | \beta, a)$ to $M_{i_1 \dots i_K} (\beta, a)$. Furthermore, for all $i_1, \dots, i_K$, we write $g^0_{i_1 \dots i_K} = (g^0_{1, i_1}, \dots, g^0_{K, i_K})^{\intercal}$, and $\hat g_{i_1 \dots i_K} = (\hat g_{1, i_1}, \dots, \hat g_{K, i_K})^{\intercal}$. For an event $E$, let $I(E)$ denote an indicator function for $E$.
Let $I_{d}$ denote $d$-dimensional identity matrix while $0_{d_1, d_2}$ denotes $d_1 \times d_2$ zero matrix.

\begin{proof} [Proof of Theorem 1]
    We first prove an intermediate result: $d (\hat \Psi, \Psi^0) \rightarrow_p 0$, where
\begin{equation} 
	d(\Psi, \Psi^0)^2:= \| \beta - \beta^0\|^2 + \sum_{k = 1}^{K} \frac{1}{n_k} \sum_{i_k = 1}^{n_k} (a_{k}(g_{k, i_k}) - a^0_{k}(g^0_{k, i_k}))^2. \label{metric}
\end{equation}
	Our proof strategy follows that of Theorem 1 of \cite{liu2020identification}.
	We divide the proofs into four steps.

	\underline{Step 1.}
	Define $d^*(\Psi, \Psi^0)^2 := \| \beta - \beta^0 \|^2 + \mathbb P_{N} ( a_{i_1 \dots i_K} - a^0_{i_1 \dots i_K} )^2
		+  Pen(\alpha, \gamma)$.
	In this step, we show that there exists $c_0 > 0$ such that, uniformly over $\theta \in \Theta$,  
	\begin{equation} 
		\mathbb E_{\gamma^0} [Q(\Psi^0)] - \mathbb E_{\gamma^0} [Q(\Psi)] \geq c_0 d^*(\Psi, \Psi^0)^2, \label{well separation}
	\end{equation}
	holds.
    Lemma \ref{normalization} implies that $\mathbb E_{\gamma^0} [Q(\Psi^0)] - \mathbb E_{\gamma^0} [Q(\Psi)] = \mathbb E_{\gamma^0} [Q^*(\Psi^0)] - \mathbb E_{\gamma^0} [Q^*(\Psi)] + Pen(\alpha, \gamma)$. By Taylor's theorem in conjunction with Assumptions \ref{as parameter space}$(b)$ and $(c)$, and \ref{as m}$(a)$ and $(b)$, the latter term is no smaller than $\sigma/2 (\| \beta - \beta^0 \|^2 + \mathbb P_N(a_{i_1 \dots i_K} - a^0_{i_1 \dots i_K})^2) + Pen(\alpha, \gamma)$. Choosing $c_0 = (\sigma/2) \land 1$ completes the step 1.

	\underline{Step 2.}
	In this step, we show that $\sup_{\Psi \in \Xi} |Q(\Psi)  - \mathbb E_{\gamma^0}[Q(\Psi)]| \rightarrow_p 0$.
	By Markov's inequality and the law of iterated expectations, it suffices to prove that
	\begin{equation} 
    \mathbb E \left[ \sup_{\Psi \in \Xi} \left| Q(\Psi) - \mathbb E_{\gamma^0} [ Q(\Psi)]\right|  \right] =
		\mathbb E \left[ \mathbb E_{\gamma^0} \left[ \sup_{\Psi \in \Xi} \left| Q(\Psi) - \mathbb E_{\gamma^0} [ Q(\Psi)]\right|  \right] \right] \rightarrow 0. \label{expectation uniform convergence}
	\end{equation}
	We examine the conditional expectation in \eqref{expectation uniform convergence}.
    Define an $i.i.d.$ sequence of Rademacher variables $\{ \xi_{i_1 \dots i_K} : 1 \leq i_k \leq n_k, 1 \leq k \leq K \}$ that are independent of $\{z_{i_1 \dots i_K} : 1 \leq i_k \leq n_k, 1 \leq k \leq K \}$. It is straightforward to adapt the proof of Lemma 2.3.1 of \cite{van1996weak} to the case of independent and non-identically distributed variables. Thus, we obtain the symmetrization inequality for the conditional expectation by conditional independence assumption on $z_{i_1 \dots i_K}$:
    \begin{align}
        \mathbb E_{\gamma^0} \left[ \sup_{\Psi \in \Xi} \left| Q(\Psi) - \mathbb E_{\gamma^0} [ Q(\Psi)]\right|  \right] = \ &\mathbb E_{\gamma^0} \left[\sup_{\Psi \in \Xi}| (\mathbb P_N - P) M_{i_1 \dots i_K} (\beta, a_{i_1 \dots i_K})|\right] \notag \\
        \leq \ &2\mathbb E_{\gamma^0} \left[\sup_{\Psi \in \Xi} | \mathbb P_N \xi_{i_1 \dots i_K} M_{i_1 \dots i_K} (\beta, a_{i_1 \dots i_K}) \ | \right].\label{consistency symmetrization}
    \end{align}
    For the right side, observe that 
    \begin{align}
        &\mathbb E_{\gamma^0} \left[ \sup_{\Psi \in \Xi}| \mathbb P_N \xi_{i_1 \dots i_K} M_{i_1 \dots i_K} (\beta, a_{i_1 \dots i_K}) \ | \right] \notag \\
        = \ &\mathbb E_{\gamma^0} \left[ \max_{\gamma \in \Gamma} \sup_{\beta \in \mathcal B, \alpha \in \mathcal A^{\prod}} \left| \frac{1}{N} \sum_{i_1 \dots i_K} \sum_{g_1 \dots g_K} \xi_{i_1 \dots i_K} 1 \{ (g_{k, i_k})_{k=1}^K = (g_k)_{k=1}^K \} M_{i_1 \dots i_K} (\beta, a_{g_1 \dots g_K}) \right|\right], \label{consistency grouping}
    \end{align}
    where $\sum_{g_1 \dots g_K}$ denotes $\sum_{g_1 = 1}^{G_1} \dots \sum_{g_K = 1}^{G_K}$ and $a_{g_1 \dots g_K}$ denotes $a_{1}(g_1) + \dots + a_{K}(g_K)$, respectively. The latter term is bounded by 
    \begin{align}
        &\sum_{g_1 \dots g_K} \mathbb E_{\gamma^0} \left[ \max_{\gamma \in \Gamma} \sup_{\beta \in \mathcal B, \alpha \in \mathcal A^{\prod}} \left| \frac{1}{N} \sum_{i_1 \dots i_K} \xi_{i_1 \dots i_K} 1 \{ (g_{k, i_k})_{k=1}^K = (g_k)_{k=1}^K \} M_{i_1 \dots i_K} (\beta, a_{g_1 \dots g_K}) \right|\right] \notag \\
        \leq \ &\left( \prod_{k = 1}^K G_k \right) \mathbb E_{\gamma^0} \left[ \max_{ \{ \delta_{k, i_k} \}_{i_k = 1}^{N_k} (k = 1, \dots, K) } \sup_{\beta \in \mathcal B, a \in \mathcal A^{+}} \left| \frac{1}{N} \sum_{i_1 \dots i_K} \delta_{1, i_1} \dots \delta_{K, i_K} \xi_{i_1 \dots i_K} M_{i_1 \dots i_K} (\beta, a) \right| \right], \label{consistency absolute bound}
    \end{align}
    where the maximum of $\{ \delta_{k, i_k} \}_{i_k = 1}^{n_k}$ is taken over $\{0, 1\}^{n_k}$, the set of all sequences of length $n_k$ consisting only of $0$ and $1$ $(k = 1, \dots, K)$.
    We now bound the conditional expectation in the last term. Let $\mathbb E_{\gamma^0, Z}$ be the conditional expectation given the array $\{z_{i_1 \dots i_K} : 1 \leq i_k \leq n_k, 1 \leq k \leq K \}$, in addition to $\gamma^0$. Let $\| \cdot \|_{\psi_2}$ be the Orlicz norm with $\psi_2 (x) = \exp(x^2) - 1$ with respect to $\mathbb E_{\gamma^0, Z}$ (see page 144 of \cite{van1996weak}) for the definition). Letting $\max_{\delta_{k, i_k} }$ be a shorthand notation for $\max_{\{ \delta_{k, i_k} \}_{i_k = 1}^{n_k} (k = 1, \dots, K) }$, we observe that 
    \begin{align}
        &\mathbb E_{\gamma^0, Z} \left[ \max_{\delta_{k, i_k}} \sup_{\beta \in \mathcal B, a \in \mathcal A^{+}} |\mathbb P_N \delta_{1, i_1} \dots \delta_{K, i_K} \xi_{i_1 \dots i_K} M_{i_1 \dots i_K} (\beta, a)| \right] \notag \\
        \leq \ & \left\|\max_{\delta_{k, i_k}} \sup_{\beta \in \mathcal B, a \in \mathcal A^{+}} |\mathbb P_N \delta_{1, i_1} \dots \delta_{K, i_K} \xi_{i_1 \dots i_K} M_{i_1 \dots i_K} (\beta, a)|  \right\|_{\psi_2}, \label{consistency orlicz bound}
    \end{align}
    by the page 145 of \cite{van1996weak} combined with the Cauchy-Schwarz inequality.
	As the cardinality of $\{0, 1 \}^{n_k}$ is $2^{n_k}$ $(k = 1, \dots, K)$, Lemma 2.2.2 of \cite{van1996weak} gives that 
    \begin{align}
        &\left\|\max_{\delta_{k, i_k}} \sup_{\beta \in \mathcal B, a \in 
        \mathcal A^{+}} |\mathbb P_N \delta_{1, i_1} \dots \delta_{K, i_K} \xi_{i_1 \dots i_K} M_{i_1 \dots i_K} (\beta, a)|  \right\|_{\psi_2} \notag \\
        \lesssim \ & \sqrt{\log (1 + 2^{\sum_{k = 1}^K n_k})} \max_{\delta_{k, i_k}} \left\| \sup_{\beta \in \mathcal B, a \in \mathcal A^{+}} |\mathbb P_N \delta_{1, i_1} \dots \delta_{K, i_K} \xi_{i_1 \dots i_K} M_{i_1 \dots i_K} (\beta, a)| \right\|_{\psi_2}. \label{consistency maximal bound}
    \end{align}
    By a straightforward calculation, the latter term is further bounded, up to a constant, by
    \begin{align}
        &\sqrt{\sum_{k = 1}^K n_k} \max_{\delta_{k, i_k}} \left\| \sup_{\beta \in \mathcal B, a \in \mathcal A^{+}} |\mathbb P_N \delta_{1, i_1} \dots \delta_{K, i_K} \xi_{i_1 \dots i_K} M_{i_1 \dots i_K} (\beta, a)| \right\|_{\psi_2} \notag \\
        \lesssim \ &\sqrt{\frac{\sum_{k=1}^K n_k}{N}} \max_{\delta_{k, i_k}} \sqrt{\frac{\Delta}{N}} \left\| \sup_{\beta \in \mathcal B, a \in \mathcal A^{+}} \left| \frac{1}{\sqrt{\Delta}} \sum_{i_1 \dots i_K} \delta_{1, i_1} \dots \delta_{K, i_K} \xi_{i_1 \dots i_K} M_{i_1 \dots i_K} (\beta, a)\right| \right\|_{\psi_2}, \label{consistency Delta bound}
    \end{align}
    where $\Delta := \prod_{k = 1}^K \Delta_k$ with $\Delta_k := \sum_{i_k = 1}^{n_k} \delta_{k, i_k}$ $(k = 1, \dots, K)$. We may ignore the case where $\Delta_k = 0$ for some $k$ as the summation is zero in this case. For the last term in the above display, by the proof of Theorem 2.14.1 of \cite{van1996weak} in conjunction with Assumption \ref{as m}$(c)$ and $(d)$, we have the following bound.
    \begin{align}
        &\sqrt{\frac{\Delta}{N}} \left\| \sup_{\beta \in \mathcal B, a \in \mathcal A^{+}} \left| \frac{1}{\sqrt{\Delta}} \sum_{i_1 \dots i_K} \delta_{1, i_1} \dots \delta_{K, i_K} \xi_{i_1 \dots i_K} M_{i_1 \dots i_K} (\beta, a)\right| \right\|_{\psi_2} \notag \\
        \lesssim \ &\sqrt{\frac{\Delta}{N}} \sqrt{ \frac{1}{\Delta} \sum_{i_1 \dots i_K} \delta_{1, i_1} \dots \delta_{K, i_K} M_{env} (z_{i_1 \dots i_K})^2 } \notag \\ 
        &\times \sup_{R} \int^{1}_{0} \sqrt{1 + \log \mathcal N (\ep \| M_{env} \|_{R, 2}, \mathcal M, L_2 (R))} d\ep \notag \\
        \lesssim\ & \sqrt{\mathbb P_N M_{env}(z_{i_1 \dots i_K})^2}. \notag
    \end{align}
    As the last term does not depend on $\delta_{k, i_k}$, \eqref{consistency orlicz bound}, \eqref{consistency maximal bound} and \eqref{consistency Delta bound} lead to
    \begin{align}
        &\mathbb E_{\gamma^0, Z} \left[ \max_{\delta_{k, i_k}} \sup_{\beta \in  \mathcal B, a \in \mathcal A^{+}} |\mathbb P_N \delta_{1, i_1} \dots \delta_{K, i_K} \xi_{i_1 \dots i_K} M_{i_1 \dots i_K} (\beta, a)| \right] \notag \\
        \lesssim \ & \sqrt{\frac{\sum_{k = 1}^K n_k}{N}} \sqrt{\mathbb P_N M_{env}(z_{i_1 \dots i_K})^2}. \notag
    \end{align}
    Taking the conditional expectation of both sides with respect to $\{z_{i_1 \dots i_K} : 1 \leq i_k \leq n_k, 1 \leq k \leq K \}$ given $\gamma^0$, we have
    \begin{align}
        &\mathbb E_{\gamma^0} \left[ \max_{\delta_{k, i_k}} \sup_{\beta \in \mathcal B, a \in \mathcal A^{+}} |\mathbb P_N \delta_{1, i_1} \dots \delta_{K, i_K} \xi_{i_1 \dots i_K} M_{i_1 \dots i_K} (\beta, a)| \right] \notag \\
        \lesssim \ &\sqrt{\frac{\sum_{k = 1}^K n_k}{N}} \mathbb E_{\gamma^0} \left[  \sqrt{\mathbb P_N M_{env}(z_{i_1 \dots i_K})^2}\right]. \notag 
    \end{align}
    Because $\mathbb E_{\gamma^0} \left[  \sqrt{\mathbb P_N M_{env}(z_{i_1 \dots i_K})^2}\right] \leq \sqrt{P M_{env} (z_{i_1 \dots i_K})^2}$ by Jensen's inequality, Assumption \ref{as m}$(c)$ yields that the above display is uniformly bounded, up to a universal constant, by $\sqrt{\sum_{k=1}^{K} n_k / N}$.
    Combining this with \eqref{consistency symmetrization}, \eqref{consistency grouping} and \eqref{consistency absolute bound}, we conclude that
    \[
        \mathbb E_{\gamma^0} \left[ \sup_{\Psi \in \Xi} |Q(\Psi) - \mathbb E_{\gamma^0} [Q(\Psi) ]| \right] \lesssim \sqrt{\frac{\sum_{k = 1}^K n_k}{N}}.
    \]
    Because the term on the right side does not depend on $\gamma^0$, the desired convergence follows.
	
	\underline{Step 3.}
	This step shows that $d^*(\hat \Psi, \Psi^0) \rightarrow_p 0$, where $d^*(\Psi, \Psi^0)$ is defined in Step 1.
	Define $S_{Q} := \sup_{\Psi \in \Xi} |Q(\Psi) - \mathbb E_{\gamma^0}[Q(\Psi)]|$. Then, by the definition of $\hat \Psi$,
	we have 
	\begin{equation} 
		\mathbb E_{\gamma^0} [Q(\Psi)]|_{\Psi = \hat \Psi} + S_{Q} \geq Q (\hat \Psi)  \geq Q (\Psi^0) \geq \mathbb E_{\gamma^0}[Q(\Psi^0)] - S_{Q}. \notag
	\end{equation}
    Subtracting $\mathbb E_{\gamma^0} [Q(\Psi)]|_{\Psi = \hat \Psi} - S_{Q}$ from the far left and the far right sides of the inequality, we have $\mathbb E_{\gamma^0} [Q(\Psi^0) ] - \mathbb E_{\gamma^0} [Q(\Psi) ]|_{\Psi = \hat \Psi} \leq 2 S_{Q}$.
    It follows from \eqref{well separation} in Step 1 that $c_0 d^* (\hat \Psi, \Psi^0)^2 \leq 2 S_{Q}$.
	As $S_{Q} \rightarrow_p 0$ from Step 2, we now conclude this step.

	\underline{Step 4.}
    This step completes the proof for $d (\hat \Psi, \Psi^0) \rightarrow_p 0$. By a straightforward calculation,
    \begin{equation} 
        d^*(\hat \Psi, \Psi^0)^2 = d (\hat \Psi, \Psi^0)^2 + Rem(\hat \Psi, \Psi^0), \notag
    \end{equation}
    where $Rem(\hat \Psi, \Psi^0)$ is defined to be
    \[
        2 \sum_{l = 1}^{K} \sum_{m = l+1}^{K} \left( \frac{1}{n_l} \sum_{i_l = 1}^{n_l} \left(\hat a_{l}(\hat g_{l, i_l}) - a^0_{l} (g^0_{l, i_l}) \right) \right) 
        \left( \frac{1}{n_m} \sum_{i_m = 1}^{n_m} \left(\hat a_{m} (\hat g_{m, i_m}) - a^0_{m}(g^0_{m, i_m})\right) \right) + Pen(\hat \alpha, \hat \gamma). 
    \]
    Let $\hat S_{k} := \frac{1}{n_k} \sum_{i_k = 1}^{n_k} \hat a_{k}(\hat g_{k, i_k})$ and $S^0_{k} := \frac{1}{n_k} \sum_{i_k = 1}^{n_k} a^0_{k}(g^0_{k, i_k})$.
    Then for each $l, m$ with $l < m$, we have that 
    \begin{align}
        (\hat S_l - S^0_{l}) (\hat S_m - S^0_{m}) = \ &\{ (\hat S_m - S^0_{m}) + (\hat S_l - \hat S_m) + (S^0_m - S^0_l)\} (\hat S_m - S^0_m) \notag \\
        = \ &\{ (\hat S_m - S^0_{m}) + (\hat S_l - \hat S_m)\} (\hat S_m - S^0_m) \notag \\
        \geq \ & (\hat S_l - \hat S_m)(\hat S_m - S^0_m), \label{consistency S inequality}
    \end{align}
    where the second equality follows from Lemma \ref{normalization}. Here, $\hat S_l - \hat S_m$ in the last term converges to zero in probability because of the convergence of $Pen(\hat \alpha, \hat \gamma)$, which is implied by the convergence of $d^*(\hat \Psi, \Psi^0)$. Additionally, $\hat S_m - S_m^0$ is bounded by Assumption \ref{as parameter space}$(a)$ so that $(\hat S_l - \hat S_m)(\hat S_m - S^0_m)$ converges to zero in probability. Combining this fact with \eqref{consistency S inequality} and the form of $Rem(\hat \Psi, \Psi^0)$, it follows that $Rem(\hat \Psi, \Psi^0) \geq \varrho_N$ for some random $\varrho_N$ converging to zero in probability. Therefore, $d^*(\hat \Psi, \Psi^0)^2 \geq d (\hat \Psi, \Psi^0)^2 + \varrho_N$ and the proof for $d (\hat \Psi, \Psi^0) \rightarrow_p 0$ is complete. Then $\hat \beta \rightarrow_p \beta^0$ immediately follows.

    To prove the latter part of the theorem, let $k = 1, \dots, K$. As in the proof of Lemma B.3 of \cite{bonhomme2015grouped}, we evaluate the two terms in the max operator of $d_{H, k}$ in turn.

We first show that, for all $g \in \{1, \dots, G_k\}$, it holds that
\begin{equation}\label{alpha mim convertence}
\min_{g' \in \{1, \dots, G_k\}} (\hat a_{k}(g') - a^0_{k}(g))^2 \rightarrow_p 0.
\end{equation}
Let $g \in \{1, \dots, G_k \}$. We observe that 
\begin{align} 
&\frac{1}{n_k} \sum_{i_k = 1}^{n_k} \min_{g' \in \{1, \dots, G_k \}} 1 \{g^0_{k, i_k} = g \} (\hat a_{k}(g') - a^0_{k}(g))^2 \notag \\
= \ & \min_{g' \in \{1, \dots, G_k \}} (\hat a_{k}(g') - \tilde a^0_{k}(g))^2 \frac{1}{n_k}\sum_{i_k = 1}^{n_k} 1 \{g^0_{k, i_k} = g \}. \notag
\end{align}
Assumption \ref{as group N}$(a)$ implies that $\frac{1}{n_k} \sum_{i_k = 1}^{n_k} 1 \{ g^0_{k, i_k} = g \} \rightarrow_p \pi_{k, g} > 0$. Thus, for \eqref{alpha mim convertence} to hold, it suffices to show that 
\[
    \frac{1}{n_k} \sum_{i_k = 1}^{n_k} \min_{g' \in \{1, \dots, G_k \}} 1 \{g^0_{k, i_k} = g \} (\hat a_{k}(g') - a^0_{k}(g))^2 \rightarrow_p 0.
\]
However, this follows because the left term is bounded by
\begin{align}
    \frac{1}{n_k} \sum_{i_k = 1}^{n_k} 1 \{ g^0_{k, i_k} = g \} \min_{g' \in \{1, \dots, G_k \}} (\hat a_{k}(g') - a^0_{k}(g))^2 \leq \ &\frac{1}{n_k} \sum_{i_k = 1}^{n_k} 1 \{ g^0_{k, i_k} = g \} (\hat a_{k}(\hat g_{k, i_k}) - a^0_{k}( g^0_{k, i_k}))^2 \notag \\
    \leq \ &\frac{1}{n_k} \sum_{i_k = 1}^{n_k} 1 (\hat a_{k}(\hat g_{k, i_k}) - a^0_{k}(g^0_{k, i_k}))^2, \notag
\end{align}
and the last term converges to zero in probability by an earlier part of this proof.

We next show that, for all $g' \in \{1, \dots, G_k\}$, it holds that
\begin{equation}\label{alpha min convergence 2}
    \min_{g \in \{1, \dots, G_k \} } (\hat a_{k}(g') - a^0_{k}(g))^2 \rightarrow_p 0,
\end{equation}
Define $\tau (g) := {\rm argmin}_{g' \in \{ 1, \dots, G_k\}} (\hat a_{k}(g') - a^0_{k}(g))^2$
for all $g \in \{1, \dots, G_k \}.$ We start by proving that the map $\tau: \{1, \dots, G_k \} \mapsto \{1, \dots, G_k \}$ is one-to-one with probability approaching one. Let $g \neq h$. Using the triangle inequality,
\begin{align}
    |\hat a_{k}(\tau(g)) - \hat a_{k}(\tau(h))| \geq \ &| a^0_{k}(g) - a^0_{k}(h)| - |a^0_{k}(g) - \hat a_{k}(\tau(g))| - |\hat a_{k}(\tau(h)) - a^0_{k}(h)| \notag \\
    = \ &|a^0_{k}(g) - a^0_{k}(h)|
    - \min_{g' \in \{1, \dots, G_k \}} | \hat a_{k}(g') - a^0_{k}(g) | -  \min_{h' \in \{1, \dots, G_k \}} |\hat a_{k}(h') - a^0_{k}(h)|, \notag
\end{align}
where the far right side converges to $|a^0_{k}(g) - a^0_{k}(h)|$ in probability by \eqref{alpha mim convertence}. It now follows from Assumption \ref{as group N}$(b)$ that there exists a positive constant $c_{g, h}$ such that $|\hat a_{k}(\tau(g)) - \hat a_{k}(\tau(h))| > c_{g, h}$ with probability approaching one. Because the probability that $\tau$ is one-to-one is no smaller than 
\[
\mathbb P \left( \bigcap_{g \neq h} \{ |\hat a_{k}(\tau(g)) - \hat a_{k}(\tau(h))| > c_{g, h}\} \right) \geq 1 - \sum_{g \neq h} \mathbb P (|\hat a_{k}(\tau(g)) - \hat a_{k}(\tau(h))| \leq c_{g, h}),
\]
and the last sum in the above display converges to zero, the map $\tau$ is one-to-one with probability approaching one. This implies that $\tau$ admits a well-defined inverse $\tau^{-1}$ with probability approaching one. Therefore, for all $g' \in \{1, \dots, G_k\}$,
\begin{align}
    \min_{g \in \{1, \dots, G_k \}} (\hat a_{k}(g') - a^0_{k}(g))^2 \leq \ & 1 \{ \tau \ \text{is one-to-one.}\} (\hat a_{k}( g') - a^0_{k}(\tau^{-1}(g')))^2 + \kappa_N \notag \\
    = \ & 1 \{ \tau \ \text{is one-to-one.}\} 
    \min_{g \in \{1, \dots, G_k \}} (\hat a_{k}(g) - \tilde a^0_{k}(\tau^{-1} (g')))^2 + \kappa_N, \notag
\end{align}
where $\kappa_N$ is a random variable converging to zero in probability. Now the far right side on the above display converges to zero in probability by \eqref{alpha mim convertence}.
Hence, \eqref{alpha min convergence 2} has been shown.
This completes the proof.
\end{proof}

\begin{proof}[Proof of Theorem 2]
We first prove $\mathbb P(\hat \gamma = \gamma^0) \rightarrow 1$. By Assumption \ref{as parameter space}$(c)$, an $\ep_{a}$-neighborhood of $a^0_{k}(g_k)$ is included in $\mathcal A$ for all $k = 1, \dots, K$ and $g_k = 1, \dots, G_k$ for some $\ep_a > 0$. Define $a_{diff}^0 := \min_{1 \leq k \leq K, g_k \neq g'_k} |a^0_{k}(g_k) -  a^0_{k}(g'_k)|$.
We may take $\ep_{a}$ sufficiently small such that $\ep_{a} <  a_{diff}^0/2$ by Assumption \ref{as group N}$(b)$.
Furthermore, define an event 
\[E := \left\{ |\hat a_{k}(g_k) - a^0_{k}(g_k)| < \frac{\ep_{a}}{2} \ (g_k = 1, \dots, G_k, k = 1, \dots, K) \right\}.\]
By Theorem 1, $\mathbb P (E) \rightarrow 1$. 

Because $\mathbb P (\hat \gamma \neq \gamma^0) \leq \sum_{k = 1}^K \mathbb P (\hat \gamma_k \neq \gamma^0_k)$, it suffices to show that $\mathbb P ( \{ \hat \gamma_k \neq \gamma^0_k \} \cap E) \rightarrow 0$ $(k = 1, \dots, K)$ for the desired result.
Fix $l = 1, \dots, K$. Observe that $\mathbb P ( \{ \hat \gamma_l \neq \gamma^0_l \} \cap E) = \mathbb P \left( \cup_{m = 1}^{n_l} \{ \hat g_{l, m} \neq g^0_{l, m} \} \cap E \right).$

We now consider an event $\{ \hat g_{l, m} \neq g^0_{l, m} \} \cap E$ for any $m = 1, \dots, n_l$.
In the following, we use the following notations for any $m = 1, \dots, n_l$, and $i_k = 1, \dots, n_k$ $(k = 1, \dots, l-1, l+1, \dots, K)$:
\begin{align}
    z_{i_1 \dots i_K}^{m} &:= z_{i_1 \dots i_{l-1} m i_{l+1} \dots i_K}, \notag \\
    \hat a_{i_1 \dots i_K}^m &:= \hat a_{1}(\hat g_{1, i_1}) + \dots + \hat a_{l-1}(\hat g_{l-1, i_{l - 1}}) + \hat a_{l}(\hat g_{l, m}) + \hat a_{l+1}(\hat g_{l+1, i_{l+1}}) + \dots + \hat a_{K}(\hat g_{k, i_k}), \notag \\
    \hat a^{m, 0}_{i_1 \dots i_K} &:= \hat a_{1}(\hat g_{1, i_1}) + \dots + \hat a_{l-1}(\hat g_{l-1, i_{l - 1}}) + \hat a_{l}(g^0_{l, m}) + \hat a_{l+1}(\hat g_{l+1, i_{l+1}}) + \dots + \hat a_{K}(\hat g_{k, i_k}). \notag 
\end{align}
For any $m = 1, \dots, n_l$, observe that $\hat g_{l, m} \neq g^0_{l, m}$ implies $Q (\hat \Psi) \geq Q (\hat \Psi_{l, m})$, or equivalently $Q^* (\hat \Psi) - Pen (\hat \alpha, \hat \gamma) \geq Q^* (\hat \Psi_{l, m}) - Pen(\hat \alpha, \hat \gamma_{l, m})$,
where $\hat \Psi_{l, m}$ and $\hat \gamma_{l, m}$ share the same elements with $\hat \Psi$ and $\hat \gamma$ except for $\hat g_{l, m}$ replaced by $g_{l, m}^0$. On $E$, $\hat \alpha$ is an interior point of its parameter space so that $Pen (\hat \alpha, \hat \gamma) = 0$ by Lemma \ref{normalization estimator}. It follows from this fact in conjunction with Assumption \ref{as parameter space}$(a)$ that $Pen(\hat \alpha, \hat \gamma_{l, m}) \leq \mathcal C_{Pen}/n_l^2$ holds on $E$ for some universal finite constant $\mathcal C_{Pen}$.
Consequently, $Q^*(\hat \Psi) \geq Q^*(\hat \Psi_{l, m}) - \mathcal C_{Pen}/n_l^2$ holds on $E$.
By a straightforward calculation, this implies that, on $E$,
\begin{equation}
\frac{1}{N_{-l}} \sum_{i_l = m} M(z_{i_1 \dots i_K}^{m} | \hat \beta, \hat a_{i_1 \dots i_K}^m)
\geq \frac{1}{N_{-l}} \sum_{i_l = m} M(z_{i_1 \dots i_K}^{m} | \hat \beta, \hat a_{i_1 \dots i_K}^{m, 0}) - \mathcal C_{Pen}/n_l, \label{inequality on event}
\end{equation}
holds where $N_{-l} := n_1 \dots n_{l-1} n_{l+1} \dots n_{K}$, and 
$\sum_{i_l = m}$ is the summation over all $i_1, \dots, i_K$ with $i_l = m$.

By Assumption \ref{as m added}$(a)$, applying Taylor's theorem to \eqref{inequality on event} yields that
\begin{align}
    0 \leq & \ \frac{1}{N_{-l}} \sum_{i_l = m} \nabla_{a} M(z_{i_1\dots i_K}^m | \hat \beta, a^{m, 0}_{i_1 \dots i_K}) (\hat a_{l}(\hat g_{l, m}) - \hat a_{l}(g^0_{l, m}))
    \notag \\
    & +  
    \frac{1}{N_{-l}} \sum_{i_l = m}
    \int_{\hat a^{m, 0}_{i_1 \dots i_K}}^{\hat a^{m}_{i_1 \dots i_K}} \nabla_{a a} M(z_{i_1 \dots i_K}^m | \hat \beta, t)  (\hat a^{m}_{i_1 \dots i_K} - t) dt + \mathcal C_{Pen}/n_l, \label{taylor wrt est}
\end{align}
holds on $E$.
For the first term on the right side of \eqref{taylor wrt est}, we apply Taylor's theorem again and obtain
\begin{align}
    &\frac{1}{N_{-l}} \sum_{i_l = m} \nabla_{a} M(z_{i_1\dots i_K}^m | \hat \beta, \hat a^{m, 0}_{i_1 \dots i_K}) \notag \\
    = \ &\frac{1}{N_{-l}} \sum_{i_l = m} \nabla_{a} M(z_{i_1\dots i_K}^m | \beta^0, a^0_{i_1 \dots i_K})
    + \frac{1}{N_{-l}} \sum_{i_l = m} \nabla_{a \beta^{\intercal}} M(z_{i_1\dots i_K}^m|\bar \beta, \bar a_{i_1 \dots i_K}) (\hat \beta - \beta^0) \notag \\
    &+ \frac{1}{N_{-l}} \sum_{i_l = m} \left\{ \nabla_{a a} M(z_{i_1 \dots i_K}^m | \bar \beta, \bar a_{i_1 \dots i_K}) \left( \sum_{k \neq l}^K (\hat a_{k}(\hat g_{k, i_k}) - a^0_{k}(g^0_{k, i_k}))  + (\hat a_{l}(g^0_{l, m}) -  a^0_{l}(g^0_{l, m}))\right)  \right\}, \label{taylor wrt true}
\end{align}
where $(\bar \beta, \bar a_{i_1 \dots i_K})$ lies on the line between $(\hat \beta, \hat a^{m, 0}_{i_1 \dots i_K})$ and $(\beta^0, a^0_{i_1 \dots i_K})$.
Meanwhile, for the second term on the right side of \eqref{taylor wrt est}, without loss of generality, assume $\hat a_{l}(g_{l, m}^0) \leq \hat a_{l}( \hat g_{l, m})$.
Then, by the choice of $\ep_{a}$ and $a_{diff}^0$, it holds that, on $E \cap \{ \hat g_{l, m} \neq g^0_{l, m} \}$, this term is no larger than
\begin{equation}
\frac{1}{N_{-l}} \sum_{i_l = m} 
    \int_{\hat a^{m, 0}_{i_1 \dots i_K} + \ep_{a}/4}^{\hat a^{m, 0}_{i_1 \dots i_K} + \ep_{a}/2} \nabla_{a a} M(z_{i_1 \dots i_K}^m | \hat \beta, t) \frac{\ep_{a}}{2} dt. \label{lower bound int}
\end{equation}
Note that the above display equals 
\begin{align}
    &\frac{\ep_{a}}{2 N_{-l}} \sum_{i_l = m} \int_{\hat a^{m, 0}_{i_1 \dots i_K} + \ep_{a}/4}^{\hat a^{m, 0}_{i_1 \dots i_K} + \ep_{a}/2} \mathbb E_{\gamma^0} [\nabla_{a a} M(z_{i_1 \dots i_K}^m | \beta, t)]|_{\beta = \hat \beta} dt. \notag \\
    + \ &\frac{\ep_{a}}{2 N_{-l}} \sum_{i_l = m} \int_{\hat a^{m, 0}_{i_1 \dots i_K} + \ep_{a}/4}^{\hat a^{m, 0}_{i_1 \dots i_K} + \ep_{a}/2} (\nabla_{a a} M(z_{i_1 \dots i_K}^m | \hat \beta, t) -  \mathbb E_{\gamma^0} [\nabla_{a a} M(z_{i_1 \dots i_K}^m | \beta, t)]|_{\beta = \hat \beta}) dt, \label{addition subtraction}
\end{align}
where $\mathbb E_{\gamma^0} [\nabla_{a a} M(z_{i_1 \dots i_K}^m | \hat \beta, t)]$ is integrable as a function of $t$ by Assumption \ref{as m added}$(a)$. Here Assumption \ref{as m added}$(a)$ ensures the interchangeability of the conditional expectation and the differentiation, by which we obtain
$\mathbb E_{\gamma^0} [\nabla_{a a} M(z_{i_1 \dots i_K}^m | \hat \beta, t)] = \nabla_{a a}\mathbb E_{\gamma^0} [M(z_{i_1 \dots i_K}^m | \hat \beta, t)]$. Combining this interchangeability with Assumption \ref{as m}$(b)$, \eqref{addition subtraction} is no larger than
\begin{equation}
    -\frac{\sigma \ep_{a}^2}{8} + \frac{\ep_{a}}{2 N_{-l}} \sum_{i_l = m} \int_{\hat a^{m, 0}_{i_1 \dots i_K} + \ep_{a}/4}^{\hat a^{m, 0}_{i_1 \dots i_K} + \ep_{a}/2} (\nabla_{a a} M(z_{i_1 \dots i_K}^m | \hat \beta, t) -  \mathbb E_{\gamma^0} \left[ \nabla_{a a} M(z_{i_1\dots i_K}^m | \beta, t) \right]|_{\beta = \hat \beta}) dt. \label{lower bound lambda int}
\end{equation}
By \eqref{taylor wrt est}, \eqref{taylor wrt true}, \eqref{lower bound lambda int} and Assumption \ref{as parameter space}$(a)$, there exists a universal finite constant $\mathcal C_{a}$, such that, on $\{ \hat g_{l, m} \neq g^0_{l, m} \} \cap E$, 
\[
    0 \leq \mathcal C_{a} (T_1 + T_2 + T_3 + T_4) - \sigma \ep^2_{a}/8 + \mathcal C_{Pen}/n_l,
\]
holds, where
\begin{align}
    T_1 := &\max_{m = 1, \dots, n_l} \left| \frac{1}{N_{-l}} \sum_{i_l = m} \nabla_{a} M (z_{i_1\dots i_K}^m | \beta^0, a^0_{i_1 \dots i_k}) \right|, \notag \\
    T_2 := &\| \hat \beta - \beta^0 \| \max_{m = 1, \dots, n_l} \left\| \frac{1}{N_{-l}} \sum_{i_l = m} \nabla_{a \beta^{\intercal}} M(z_{i_1 \dots i_K}^m | \bar \beta, \bar a_{i_1 \dots  i_k})  \right\|, \notag \\
    T_3 := &\max_{m = 1, \dots, n_l} \left| \frac{1}{N_{-l}} \sum_{i_l = m} \left\{ \nabla_{a a} M(z_{i_1 \dots i_K}^m | \bar \beta, \bar a_{i_1 \dots i_K}) \right. \right. \notag \\
    &\left. \left. \times \left( \sum_{k \neq l}^K (\hat a_{k}(\hat g_{k, i_k}) -  a^0_{k}(g^0_{k, i_k}))  + (\hat a_{l}(g^0_{l, m}) - a^0_{l}(g^0_{l, m}))\right)  \right\} \right|, \notag \\
    T_4 := &\frac{\ep_{a}}{2} \max_{m = 1, \dots, n_l} \left| \frac{1}{N_{-l}} \sum_{i_l = m} \int^{\hat a^{m, 0}_{i_1 \dots i_K} + \ep_{a}/2}_{\hat a^{m, 0}_{i_1 \dots i_K} + \ep_{a}/4}  \left ( \nabla_{aa} M(z_{i_1\dots i_K}^m | \hat \beta, t) \right. \right. \notag \\
    & \left. \left. - \mathbb E_{\gamma^0} \left[ \nabla_{a a} M(z_{i_1\dots i_K}^m | \beta, t) \right]|_{\beta = \hat \beta} \right) dt \right|. \notag
\end{align}
As $T_1, T_2, T_3,$ and $T_4$ do not depend on $m = 1, \dots, N_l$, we have that
\begin{equation}
	\mathbb P \left( \cup_{i_l = 1}^{n_l} \{ \hat g_{l, i_l} \neq g^0_{l, i_l}\} \cap E \right) \leq \mathbb P (0 \leq \mathcal C_{\alpha} (T_1 + T_2 + T_3 + T_4) - \sigma \ep^{2}_{a} / 8 + \mathcal C_{Pen} / n_l ). \notag
\end{equation}
If $T_1, T_2, T_3,$ and $T_4$ all converge to zero in probability, the right side goes to zero and the proof for $\mathbb P(\hat \gamma = \gamma^0) \rightarrow 1$ is complete. Hence, it remains to show that they actually converge. Because $\mathbb P(E) \rightarrow 1$, we may assume that $E$ holds in the following.

For $T_1$, by Markov's inequality and the law of iterated expectations, it suffices to show that
\begin{equation}
    \mathbb E \left[ \mathbb E_{\gamma^0} \left[ \max_{m = 1, \dots, n_l} \left| \frac{1}{N_{-l}} \sum_{i_l = m} \nabla_{a} M (z_{i_1\dots i_K}^m | \beta^0, a^0_{i_1 \dots i_k}) \right|  \right] \right]
    \rightarrow 0. \label{T_1 convergence}
\end{equation}
By Assumption \ref{as m}$(a)$, $\nabla_a \mathbb E_{\gamma^0}[M(z_{i_1 \dots i_K} | \beta^0, a^0_{i_1 \dots i_K} )] = 0$. Here, we can interchange differentiation and expectation by Assumption \ref{as m added}$(a)$ so that $\mathbb E_{\gamma^0} [\nabla_a M(z_{i_1 \dots i_K} | \beta^0, a^0_{i_1 \dots i_K} )] = 0$.
By Lemma 8 of \cite{chernozhukov2015comparison} in view of the conditional independent assumption on $z_{i_1 \dots i_K}$, we can bound the inner conditional expectation in \eqref{T_1 convergence} as
\begin{equation}
    \mathbb E_{\gamma^0} \left[ \max_{m = 1, \dots, n_l} \left| \frac{1}{N_{-l}} \sum_{i_l = m} \nabla_{a} M (z_{i_1\dots i_K}^m | \beta^0, a^0_{i_1 \dots i_k}) \right|  \right] \lesssim \frac{\sigma_M \sqrt{\log n_l} + \sqrt{\mathbb E_{\gamma^0} [\mathcal C_{M}^2]} \log n_l}{N_{-l}}, \label{A_1 c.e. bound}
\end{equation}
where
\begin{align}
    \mathcal C_M &:= \max_{i_1, \dots, i_K} |\nabla_{a} M(z^m_{i_1 \dots i_K} | \beta^0, a^0_{i_1 \dots i_K})|, \notag \\
    \sigma^2_M &:= \max_{m = 1, \dots, n_l} \sum_{i_l = m} \mathbb E_{\gamma^0} [\{ \nabla_{a} M(z^m_{i_1 \dots i_K} | \beta^0, a^0_{i_1 \dots i_k}) \}^2 ], \notag
\end{align}
with the maximum $\max_{i_1, \dots, i_K}$ being taken over all $i_k = 1, \dots, n_k$ $(k = 1, \dots, K)$.
Let $\| \cdot \|_{\psi_1}$ be the Orlicz norm for $\psi_2(x) = \exp(x) - 1$ with respect to the conditional expectation $\mathbb E_{\gamma^0}$.
By the page 145 of \cite{van1996weak}, we have that, for some universal finite constant $\mathcal C_{\psi_1}$, 
\begin{equation} 
	\sqrt{\mathbb E_{\gamma^0} [\mathcal C_M^2]} \leq \mathcal C_{\psi_1} \left\| \max_{i_1, \dots, i_K} |\nabla_{a} M(z^m_{i_1 \dots i_K} | \beta^0, a^0_{i_1 \dots i_K})| \right\|_{\psi_1}. \label{M bound orlicz}
\end{equation}
By Lemma 2.2.2 of \cite{van1996weak}, the Orlicz norm on the right side can be bounded as follows.
\begin{align} 
	&\left\| \max_{i_1, \dots, i_K} |\nabla_{a} M(z^m_{i_1 \dots i_K} | \beta^0,  a^0_{i_1 \dots i_K})| \right\|_{\psi_1} \notag \\
	\lesssim \ &\log (1 + N) \max_{i_1, \dots, i_K} \| \nabla_{a} M(z_{i_1 \dots i_K}^m | \beta^0, a^0_{i_1 \dots i_K}) \|_{\psi_1} \notag \\
	\leq & ((1 + \mathcal C_1)/\mathcal C_2) \log(1 + N), \label{M finite bound}
\end{align}
where the last inequality follows from Lemma 2.2.1 of \cite{van1996weak} in conjunction with Assumption \ref{as m added}$(b)$. 
For $\sigma^2_M$, similarly to the bound on $\sqrt{\mathbb E_{\gamma^0} [\mathcal C_M^2]}$, we have 
\begin{equation} 
	\sigma^2_{M} \leq \mathcal C_{\psi_1}^2 \max_{m = 1, \dots, n_l} \sum_{i_l = m} \| \nabla_a M(z_{i_1 \dots i_K}^m | \beta^0, a^0_{i_1 \dots i_K}) \|_{\psi_1}^2 \leq N_{-l} \mathcal C_{\psi_1}^2 ((1 + \mathcal C_1)/\mathcal C_2)^2.\label{sigma bound}
\end{equation}
Combining \eqref{M bound orlicz}, \eqref{M finite bound}, and \eqref{sigma bound}, \eqref{A_1 c.e. bound} leads to 
\begin{equation} 
	\mathbb E_{\gamma^0} \left[ \max_{m = 1, \dots, n_l} \left| \frac{1}{N_{-l}} \sum_{i_l = m} \nabla_{a} M (z_{i_1\dots i_K}^m | \beta^0, a^0_{i_1\dots i_k}) \right|  \right] \lesssim \sqrt{\frac{\log n_l}{N_{-l}}} + \frac{\log n_l \log (1 + N)}{N_{-l}}. \notag
\end{equation}
As the right side does not depend on the value of $\gamma^0$, the convergence of $T_1$ follows from the assumption of the theorem.

For the convergence of $T_2$, note that the max term in $T_2$ is bounded by Assumption \ref{as m added}$(c)$.
Because $\hat \beta \rightarrow_p \beta^0$ by Theorem 1, the convergence of $T_2$ holds.

For $T_3$, we first note that, by the triangle inequality,
\begin{align} 
	T_3 \leq \ &\sum_{k \neq l} \max_{m = 1, \dots, n_l} \frac{1}{N_{-l}} \sum_{i_l = m} \left| \nabla_{a a} M(z_{i_1 \dots i_K}^m | \bar \beta, \bar a_{i_1 \dots i_K}) (\hat a_{k}(\hat g_{k, i_k}) - a^0_{k}(g^0_{k, i_k}))\right| \notag \\
	&+ \max_{m = 1, \dots, n_l} \frac{1}{N_{-l}} \sum_{i_l = m} \left| \nabla_{a a} M(z_{i_1 \dots i_K}^m | \bar \beta, \bar a_{i_1 \dots i_K}) (\hat a_{l}(g^0_{l, m}) - a^0_{l}(g^0_{l, m}))\right|. \label{A_3 bound}
\end{align}
For the first term on the right side, the Cauchy-Schwarz inequality gives that, for each $k \neq l$, 
\begin{align} 
	&\max_{m = 1, \dots, n_l} \frac{1}{N_{-l}} \sum_{i_l = m} \left| \nabla_{a a} M(z_{i_1 \dots i_K}^m | \bar \beta, \bar a_{i_1 \dots i_K}) (\hat a_{k}(\hat g_{k, i_k}) - a^0_{k}(g^0_{k, i_k}))\right| \notag \\
	\leq \ & \max_{m = 1, \dots, n_l}  \sqrt{\frac{1}{N_{-l}} \sum_{i_l = m} \nabla_{aa} M(z_{i_1 \dots i_K}^m | \bar \beta, \bar a_{i_1 \dots i_K})^2}  \sqrt{ \frac{1}{N_{-l}} \sum_{i_l = m}  (\hat a_{k}(\hat g_{k, i_k}) - a^0_{k} (g^0_{k, i_k}))^2 } \notag \\
	= \ &  \sqrt{ \frac{1}{n_{k}} \sum_{i_k = 1}^{n_k}  (\hat a_{k} (\hat g_{k, i_k}) - a^0_{k}(g^0_{k, i_k}))^2 } \max_{m = 1, \dots, n_l} \sqrt{ \frac{1}{N_{-l}} \sum_{i_l = m} \nabla_{aa} M(z_{i_1 \dots i_K}^m | \bar \beta, \bar a_{i_1 \dots i_K})^2 }. \notag
\end{align}
By the proof of Theorem 1 and Assumption \ref{as m added}$(c)$, the last term converges to zero in probability.
Meanwhile, the second term on the right side of \eqref{A_3 bound} can be bounded by 
\[
	 \| \hat \alpha_l - \alpha^0_l \| \max_{m = 1, \dots, n_l} \frac{1}{N_{-l}} \sum_{i_l = m} |\nabla_{aa} M(z_{i_1 \dots i_K}^m | \bar \beta, \bar a_{i_1 \dots i_K})|. \notag 
\]
By Theorem 1 and Assumption \ref{as m added}$(c)$, this term converges to zero in probability. In view of \eqref{A_3 bound}, the convergence of $T_3$ follows.

Lastly, for $T_4$, $\mathbb E_{\gamma^0} [\nabla_{aa} M(z_{i_1 \dots i_K}^m | \beta, t)] = \nabla_{a} \mathbb E_{\gamma^0} [\nabla_{a} M(z_{i_1 \dots i_K}^m | \beta, t)]$
by Assumption \ref{as m added}$(a)$. Hence, the fundamental theorem of calculus and the triangle inequality implies that the max term in $T_4$ is bounded by
\begin{align}
   T_4 
	\leq & \max_{m = 1, \dots, n_l} \left|  \frac{1}{N_{-l}} \sum_{i_l = m} \left( \nabla_{a} M (z^m_{i_1 \dots i_K} | \hat \beta, \hat a_{i_1 \dots i_K}^{m, 0} + \ep_{a}/2)  \right. \right. \notag \\
	&\left. \left.  - \mathbb E_{\gamma^0} [\nabla_{a} M (z^m_{i_1 \dots i_K} | \beta, t)]|_{\beta = \hat \beta, t = \hat a^{m, 0}_{i_1 \dots i_K} + \ep_{a}/2}  \right) \right| \notag \\
	&+ \max_{m = 1, \dots, n_l} \left|  \frac{1}{N_{-l}} \sum_{i_l = m} \left( \nabla_{a} M (z^m_{i_1 \dots i_K} | \hat \beta, \hat a_{i_1 \dots i_K}^{m, 0} + \ep_{a}/4)  \right. \right. \notag \\
	&\left. \left.  - \mathbb E_{\gamma^0} [\nabla_{a} M (z^m_{i_1 \dots i_K} | \beta, t)]|_{\beta = \hat \beta, t = \hat a^{m, 0}_{i_1 \dots i_K} + \ep_{a}/4}  \right) \right|. \notag 
\end{align}
We only show the convergence of the fist maximum on the right side as the convergence for the other can be shown in the same manner. Let $\tilde \ep_{a} := \ep_{a}/2$ for brevity. By Assumption \ref{as m added}$(a)$, $(\beta, t) \mapsto  \mathbb E_{\gamma^0} [\nabla_{a} M (z^m_{i_1 \dots i_K} | \beta, t)]$ is continuously differentiable, and its derivative and expectation are interchangeable.
Hence, Taylor's theorem gives that
\begin{align} 
	&\max_{m = 1, \dots, n_l} \left|  \frac{1}{N_{-l}} \sum_{i_l = m} \left( \nabla_{a} M (z^m_{i_1 \dots i_K} | \hat \beta, \hat a_{i_1 \dots i_K}^{m, 0} + \tilde \ep_{a})  \right. \right. \notag \\
	&\left. \left.  - \mathbb E_{\gamma^0} [\nabla_{a} M (z^m_{i_1 \dots i_K} | \beta, t)]|_{\beta = \hat \beta, t = \hat a^{m, 0}_{i_1 \dots i_K} + \tilde \ep_{a}}  \right) \right| \notag \\
	= &\max_{m = 1, \dots, n_l} \left| \frac{1}{N_{-l}} \sum_{i_l = m} \left\{ \nabla_{a} M (z_{i_1 \dots i_K}^m | \beta^0, a^0_{i_1\dots i_K} + \tilde \ep_{a}) \right. \right. \notag \\
    &- \mathbb E_{\gamma^0} [\nabla_{a} M(z_{i_1 \dots i_K}^m | \beta^0,  a^0_{i_1\dots i_K} + \tilde \ep_{a})]  \notag \\
	& + (\nabla_{a \beta^{\intercal}} M(z_{i_1 \dots i_K}^m | \bar \beta, \bar a) - \mathbb E_{\gamma^0} [\nabla_{a \beta^{\intercal}}M(z_{i_1 \dots i_K}^m | \beta, t) ]|_{\beta = \bar \beta, t = \bar a}) (\hat \beta - \beta^0)  \notag \\
	&\left. \left. + (\nabla_{aa} M(z_{i_1 \dots i_K}^m | \bar \beta, \bar a) - \mathbb E_{\gamma^0} [\nabla_{aa}M(z_{i_1 \dots i_K}^m | \beta, t) ]|_{\beta = \bar \beta, t = \bar a}) (\hat a_{i_1 \dots i_K}^{m, 0} -  a^0_{i_1 \dots i_K}) \right\} \right|, \notag
\end{align}
where $\bar \beta$ lies between $\hat \beta$ and $\beta^0$, and $\bar a$ lies between $\hat a^{m, 0}_{i_1 \dots i_K} + \tilde \ep_{a}$ and $a^0_{i_1 \dots i_K} + \tilde \ep_{a}$ so that the value of $\bar a$ may depend on $i_1, \dots, i_K$.
Then, by the triangle inequality and a similar argument to that for the convergence of $T_3$, the last term in the above display is further bounded by
\begin{align} 
	&\max_{m= 1, \dots, n_l} \left|  \frac{1}{N_{-l}} \sum_{i_l = m} \nabla_{a} M (z_{i_1 \dots i_K}^m | \beta^0, a^0_{i_1\dots i_K} + \tilde \ep_{a}) - \mathbb E_{\gamma^0} [\nabla_{a} M(z_{i_1 \dots i_K}^m | \beta^0,  a^0_{i_1\dots i_K} + \tilde \ep_{a})] \right| \notag \\
	&+ \mathcal C \left( \| \hat \beta - \beta^0 \| + 
    \sum_{k \neq l} \left( \frac{1}{N_k} \sum_{i_k = 1}^{N_k} (\hat a_{k}(\hat g_{k, i_k}) - a^0_{k}(g^0_{k, i_k}))^2 \right)^{1/2} + \| \hat \alpha_l - \alpha^0_l \|\right). \notag
\end{align}
for some universal finite constant $\mathcal C > 0$.
The convergence of the maximum term follows from a similar argument to that for the convergence of $T_1$, while the convergence for the other terms follows from Theorem 1 and its proof. 
Hence, the convergence of $T_4$ follows, and the proof for $\mathbb P(\hat \gamma = \gamma^0) \rightarrow 1$ is complete. 

We proceed to deriving the asymptotic distribution of $\sqrt{N}(\hat \theta - \theta^0)$. To this end, we first define
\[
    Q^0 (\theta) := 
    \mathbb P_N M_{i_1 \dots i_K} (\beta, a_{1}(g^0_{1, i_1}) + \dots + a_{K}(g^0_{K, i_K}))
    - Pen(\alpha, \gamma^0).
\]
This is an infeasible loss function with known $\gamma^0$. 
It follows from an earlier part of this proof that $\hat \theta = {\rm argmax}_{\beta \in \mathcal B, \alpha \in \mathcal A^{\prod}} Q^0 (\theta)$ holds with probability approaching one.
By Assumption \ref{as parameter space}$(b)$ and $(c)$, and Theorem 1, $\hat \theta$ lies in an interior of its parameter space with probability approaching one. Thus, the first order condition gives that, with probability approaching one, 
\begin{equation}
    \sqrt{N} (\mathcal S(\hat \theta) - \mathcal S^{Pen} (\hat \theta) ) = 0, \label{foc}
\end{equation}
where 
\begin{align} 
&\mathcal S(\hat \theta) := \left( s_{\beta} (\hat \theta)^{\intercal}, s_{a_{1, 1}} (\hat \theta), \dots, s_{a_{1, G_1}}(\hat \theta), \dots, s_{a_{K, 1}} (\hat \theta), \dots, s_{a_{k, g_k}} (\hat \theta)\right)^{\intercal}, \notag \\
&\mathcal S^{Pen} (\hat \theta) :=
\left(0^{\intercal}_{p, 1}, s_{a_{1, 1}}^{Pen} (\hat \theta), \dots, s_{a_{1, G_1} }^{Pen}(\hat \theta), \dots, s_{a_{K, 1}}^{Pen} (\hat \theta), \dots, s_{a_{k, g_k}}^{Pen} (\hat \theta) \right)^{\intercal}, \notag
\end{align}
are score vectors corresponding to the leading term and the penalty term, respectively, with 
\begin{align}
    s_{\beta} (\hat \theta) := \ &\mathbb P_N \nabla_{\beta} M_{i_1 \dots i_K}( \hat \beta, \hat a_{1}(g^0_{1, i_1}) + \dots + \hat a_{K}(g^0_{K, i_K})), \notag \\
    s_{a_{k, g_k}} (\hat \theta) := \ &\mathbb P_N  1 \{ g^0_{k, i_k} = g_k \} \nabla_{a}M_{i_1 \dots i_K}( \hat \beta, \hat a_{1}(g^0_{1, i_1}) + \dots + \hat a_{K} (g^0_{K, i_K})), \notag \\
    s_{a_{k, g_k}}^{Pen} (\hat \theta) := \ &\lambda \left(\frac{1}{n_k} \sum_{i_k = 1}^{n_k} 1 \{ g^{0}_{k, i_k} = g_k\} \right)  \notag \\
    \ &\times \left( \sum_{l = (k-1)\lor 1}^{k \land (K-1)}(-1)^{k-l} \left( \frac{1}{n_{l}} \sum_{i_{l}=1}^{n_{l}} \hat a_{l}(g^0_{l, i_{l}}) - \frac{1}{n_{l+1}} \sum_{i_{l+1}=1}^{n_{l+1}} \hat a_{l+1}(g^0_{l+1, i_{l+1}}) \right) \right), \notag
\end{align}
for $k = 1, \dots, K$ and $g_k = 1, \dots, G_k$. 
By Lemma \ref{normalization}, $\mathcal S^{Pen}(\theta^0) = 0$. Hence, for \eqref{foc}, expanding $\mathcal S(\hat \theta) + \mathcal S^{Pen} (\hat \theta)$ around $\theta^0$ by Taylor's theorem yields that
\begin{equation}
    -\sqrt{N} \mathcal S (\theta^0) = (\mathcal J_{N} - \lambda \mathcal J_{Pen, N}) \sqrt{N} (\hat \theta - \theta^0), \label{foc score}
\end{equation}
holds with probability approaching one. Here $\mathcal J_N$ and $\lambda \mathcal J_{Pen, N}$ are the Hessian matrices collecting derivatives of $\mathcal S (\cdot)$ and $\mathcal S^{Pen} (\cdot)$, respectively, whose definitions are detailed later. To derive the asymptotic distribution of $\sqrt{N} (\hat \theta - \theta^0)$, we now divide the proof into the three steps: $(i)$ weak convergence of $\sqrt{N} \mathcal S(\theta^0)$, $(ii)$ probability convergence of $\mathcal J_N$ and $\mathcal J_N^{Pen}$, and $(iii)$ weak convergence of $\sqrt{N} (\hat \theta - \theta^0)$.

\underline{$(i)$ weak convergence of $\sqrt{N} \mathcal S(\theta^0)$.}
For $g_k = 1, \dots, G_k$ $(k = 1, \dots, K)$, define
\begin{equation}
    s_{g_1 \dots g_K} =
        \sqrt{N} \mathbb P_N 1 \{ g^0_{i_1 \dots i_K} = (g_k)_{k = 1}^{K} \} \nabla_{\beta_a} M_{i_1 \dots i_K}(\beta^0,  a^0_{i_1 \dots i_K}). \notag
\end{equation}
Accordingly, let $s_N$ denote $(p + 1) \prod_{k = 1}^K G_k$-dimensional vector concatenating $s_{g_1 \dots g_K}$ for all $g_k = 1, \dots, G_k \ (k = 1, \dots, K)$. Let $\mathcal L = (\mathcal L_{\beta}^{\intercal}, \mathcal L_{\alpha}^{\intercal})^{\intercal}$ be a non-random matrix of the following form:
\begin{align}
    \mathcal L_{\beta} := \ & (1, \dots, 1) \otimes (I_{p}, 0_{p, 1}), \label{lambda beta} \\ \mathcal L_{\alpha}  := \ & \sum_{k = 1}^{K} G_k \times \left\{ (p + 1) \prod_{k = 1}^K G_k \right\} \text{ matrix whose } g^*(p + 1) \text{-th column can be} \notag \\
    & \text{non-zero for } g^* = 1, \dots, \prod_{k = 1}^K G_k \text{ while the rest elements are zero}, \label{lambda alpha}
\end{align}
where $(1, \dots, 1)$ is a $1 \times \prod_{k=1}^K G_k$ matrix.
Then, we can write
$\sqrt{N} \mathcal S(\theta^0) = \mathcal L s_{N}$ for some $\mathcal L$ satisfying \eqref{lambda beta} and \eqref{lambda alpha}.

We derive the asymptotic distribution of $s_{N}$. To this end, let $t_{N}$ be a $(p + 1) \prod_{k=1}^{K} G_k$-dimensional real vector. We regard this $t_{N}$ as a concatenation of $p + 1$-dimensional vector $t_{g_1 \dots g_K}$ over all $g_k = 1, \dots, G_K \ (k =1, \dots, K)$ in the same way as $s_N$ is a concatenation of $s_{g_1 \dots g_K}$.
We can evaluate the characteristic function of $s_N$ as
\begin{align}
    \mathbb E[\exp(\sqrt{-1} t_N^{\intercal} s_N)] &= \mathbb E\left[\mathbb E_{\gamma^0} \left[\prod_{g_1 = 1}^{G_1} \dots \prod_{g_K = 1}^{G_K} \exp(\sqrt{-1} t_{g_1 \dots g_K}^{\intercal} s_{g_1 \dots g_K})\right]\right] \notag \\
    &= \mathbb E\left[ \prod_{g_1 = 1}^{G_1} \dots \prod_{g_K = 1}^{G_K}\mathbb E_{\gamma^0} [ \exp(\sqrt{-1} t_{g_1 \dots g_K}^{\intercal} s_{g_1 \dots g_K})]\right], \label{ch.f}
\end{align}
where the second equality follows from the conditional independence assumption on $z_{i_1 \dots i_K}$. For the inner conditional expectation, observe that
\begin{align}
    &\mathbb E_{\gamma^0} [ \exp(\sqrt{-1} t_{g_1 \dots g_K}^{\intercal} s_{g_1 \dots g_K})] \notag \\
    = \ &\prod_{g^0_{i_1 \dots i_K} = (g_k)_{k = 1}^K} \mathbb E_{\gamma^0} \left[ \exp \left\{
    \frac{\sqrt{-1}t_{g_1 \dots g_K}^{\intercal} \nabla_{\beta_a} M_{i_1 \dots i_K}(\beta^0, a^0_{i_1 \dots i_K})}{\sqrt{N}}\right\} \right], \label{ch.f conditional}
\end{align}
where the product $\prod_{g^0_{i_1 \dots i_K} = (g_k)_{k = 1}^K}$ is taken over all indices $i_1, \dots, i_K$ such that $g^0_{i_1 \dots i_K} = (g_k)_{k = 1}^K$.
By page 188 of \cite{Williams1991}, it holds that, for any real number $u$,
\[
    \left| \exp(\sqrt{-1} u) - \left(1 + \sqrt{-1} u - \frac{u^2}{2}\right) \right| \leq |u|^3/6.
\]
Applying this fact to \eqref{ch.f conditional}, in conjunction with Assumptions \ref{as m}$(a)$ and \ref{as m added}$(a)$, we have that
\begin{align}
    &\mathbb E_{\gamma^0} [ \exp(\sqrt{-1} t_{g_1 \dots g_K}^{\intercal} s_{g_1 \dots g_K})] \notag \\
    = \ &\prod_{g^0_{i_1 \dots i_K} = (g_k)_{k = 1}^K} \left( 1 - \frac{\mathbb E_{\gamma^0} \left[ \left\{ t_{g_1 \dots g_K}^{\intercal} \nabla_{\beta_a} M_{i_1 \dots i_K} (\beta^0, a^0_{i_1 \dots i_K}) \right\}^2\right]
    }{2N} + \mathbb E_{\gamma^0} [\rho_{i_1 \dots i_K}] \right), \label{ch.f. taylor}
\end{align}
where $\rho_{i_1 \dots i_K}$ is a random variable such that 
\begin{equation}
|\rho_{i_1 \dots i_K}| \leq \frac{ \left| t_{g_1 \dots g_K}^{\intercal} \nabla_{\beta_a} M_{i_1 \dots i_K} (\beta^0, a^0_{i_1 \dots i_K}) \right|^3}{6N^{3/2}}. \label{rho bound}
\end{equation}
Taking the logarithm of the right side of \eqref{ch.f. taylor} and applying Taylor's theorem yields that
\begin{align}
    &\log \mathbb E_{\gamma^0} [ \exp(\sqrt{-1} t_{g_1 \dots g_K}^{\intercal} s_{g_1 \dots g_K})] \notag \\
    = \ & \sum_{g^0_{i_1 \dots i_K} = (g_k)_{k = 1}^K} \left( - \frac{\mathbb E_{\gamma^0} \left[ \left\{ t_{g_1 \dots g_K}^{\intercal} \nabla_{\beta_a} M_{i_1 \dots i_K} (\beta^0, a^0_{i_1 \dots i_K}) \right\}^2\right]
    }{2N} + \mathbb E_{\gamma^0} [\rho_{i_1 \dots i_K}]\right) + o(1), \label{log taylor}
\end{align}
where the sum $\sum_{g^0_{i_1 \dots i_K} = (g_k)_{k = 1}^K}$ is taken over all indices $i_1, \dots, i_K$ such that $g^0_{i_1 \dots i_K} = (g_k)_{k = 1}^K$, and $o(1)$ term is due to \eqref{rho bound} and Assumption \ref{as m added}$(e)$. Using \eqref{rho bound} and Assumption \ref{as m added}$(e)$ again, $\sum_{g^0_{1, i_1} = g_1, \dots, g^0_{K, i_K} = g_K} \mathbb E_{\gamma^0}[\rho_{i_1 \dots i_K}] = o(1)$. Therefore, in view of \eqref{log taylor} and Assumption \ref{as m added}$(f)$, $\mathbb E_{\gamma^0} [\exp(\sqrt{-1} t^{\intercal}_{g_1 \dots g_K} s_{g_1 \dots g_K})] \rightarrow \exp(- t_{g_1 \dots g_K}^{\intercal} D_{g_1 \dots g_K} t_{g_1 \dots g_K}/2)$ almost surely. It follows from the dominated convergence theorem and \eqref{ch.f} that 
$\mathbb E[\exp(\sqrt{-1} t_N^{\intercal} s_N)]$ converges to $\prod_{g_1 = 1}^{G_1} \dots \prod_{g_K = 1}^{G_K} \exp(- t^{\intercal}_{g_1 \dots g_K} D_{g_1 \dots g_K} t_{g_1 \dots g_K}/2)$. By L\'{e}vy's continuity theorem and Cram\'{e}r-Wold device, $s_N \rightarrow_d N(0, \mathcal D)$, where $\mathcal D$ is a block-diagonal matrix whose diagonal submatrix corresponding to ``$s_{g_1 \dots g_K}$'' part is $D_{g_1 \dots g_K}$.

Given this convergence result and $\sqrt{N} \mathcal S(\theta^0) = \mathcal L s_N$, $\sqrt{N} \mathcal S(\theta^0) \rightarrow_d N(0, \mathcal L \mathcal D \mathcal L^{\intercal})$ follows, which completes this step.

\underline{$(ii)$ convergence of $\mathcal J_N$ and $\lambda \mathcal J_{Pen, N}$.}
We first show the convergence of $\mathcal J_N$. Observe that $\mathcal J_N$ has the form,
\[
    \mathcal J_N :=
    \begin{pmatrix}
        J^{\beta \beta}_N & J^{\beta \alpha_1}_N & \cdots & J^{\beta \alpha_K}_N \\
        J^{\alpha_1 \beta}_N & J^{\alpha_1 \alpha_1}_N & \cdots & J^{\alpha_1 \alpha_K}_N \\
        \vdots & \vdots & \ddots & \vdots \\
        J^{\alpha_K \beta}_N & J^{\alpha_K \alpha_1}_N & \cdots & J^{\alpha_K \alpha_K}_N
    \end{pmatrix}, \notag
\]
with the submatrices defined as
\begin{align}
    J^{\beta \beta}_N := \ & \mathbb P_N \nabla_{\beta \beta^{\intercal}} M_{i_1 \dots i_K} (\bar \beta,  \bar a_{i_1 \dots i_K}), \notag \\
    J^{\beta \alpha_k}_N := \ & p \times G_k \ \text{matrix whose } g \text{ th column is } \mathbb P_N 1 \{ g^0_{k, i_k} = g \} \nabla_{\beta a} M_{i_1 \dots i_K}(\bar \beta, \bar a_{i_1 \dots i_K}), \notag \\
    J^{\alpha_k \beta}_N := \ & G_k \times p \ \text{matrix whose } g \text{ th row is } \mathbb P_N 1 \{ g^0_{k, i_k} = g \} \nabla_{a \beta^\intercal} M_{i_1 \dots i_K}(\bar \beta, \bar a_{i_1 \dots i_K}), \notag \\
    J^{\alpha_k \alpha_{l}}_N := \ & G_k \times G_{l} \ \text{matrix whose } (g, h)\text{-th element} \notag \\ &\text{is } \mathbb P_N 1 \{g^0_{k, i_k} = g \} 1\{ g^0_{l, i_{l}} = h \} \nabla_{a a} M_{i_1 \dots i_K}(\bar \beta, \bar a_{i_1 \dots i_K}), \notag
\end{align}
where $\bar a_{i_1 \dots i_K} = \bar a_{1}(g^0_{1, i_1}) + \dots +\bar a_{K} (g^0_{K, i_K})$, and $\bar \beta$ and $\bar a_{k}(g_k)$ lie on the paths between $\hat \beta$ and $\beta^0$, and between $\hat a_{k}(g_k)$ and $a^0_{k}(g_k)$, respectively for $k = 1, \dots, K$ and $g_k = 1, \dots, G_k$. We allow the values of $\bar \beta$ and $\bar a_{k}(g_k)$ to vary across rows of $\mathcal J_N$.

We first make the decomposition $\mathcal J_N = \sum_{g_1 = 1}^{G_1} \dots \sum_{g_K = 1}^{G_K} \mathcal J_{N, g_1 \dots g_K},$ with
\begin{equation}
    \mathcal J_{N, g_1 \dots g_K} = 
    \begin{pmatrix}
        J^{\beta \beta}_{N, g_1 \dots g_K} & J^{\beta \alpha_1}_{N, g_1 \dots g_K} & \cdots & J^{\beta \alpha_K}_{N, g_1 \dots g_K} \\
        J^{\alpha_1 \beta}_{N, g_1 \dots g_K} & J^{\alpha_1 \alpha_1}_{N, g_1 \dots g_K} & \cdots & J^{\alpha_1 \alpha_K}_{N, g_1 \dots g_K} \\
        \vdots & \vdots & \ddots & \vdots \\
        J^{\alpha_K \beta}_{N, g_1 \dots g_K} & J^{\alpha_K \alpha_1}_{N, g_1 \dots g_K} & \cdots & J^{\alpha_K \alpha_K}_{N, g_1 \dots g_K}
    \end{pmatrix}, \notag
\end{equation}
where the submatrices are defined as
\begin{align}
    J^{\beta \beta}_{N, g_1 \dots g_K} := \ & \mathbb P_{N} 1 \{ g^0_{i_1 \dots i_K} = (g_k)_{k=1}^K\}\nabla_{\beta \beta^{\intercal}} M_{i_1 \dots i_K} (\bar \beta, \bar a_{i_1 \dots i_K}), \notag \\
    J^{\beta \alpha_k}_{N, g_1 \dots g_K} := \ & p \times G_k \ \text{matrix whose } g_k \text{-th column is} \notag \\
    & \mathbb P_N 1 \{ g^0_{i_1 \dots i_K} = (g_k)_{k=1}^K\} \nabla_{\beta a} M_{i_1 \dots i_K}(\bar \beta, \bar a_{i_1 \dots i_K}), \notag \\ &\text{while the rest elements are } 0, \notag \\
    J^{\alpha_k \beta}_{N, g_1 \dots g_K} := \ & G_k \times p  \ \text{matrix whose } g_k \text{-th row is} \notag \\
    & \mathbb P_N 1 \{ g^0_{i_1 \dots i_K} = (g_k)_{k=1}^K\} \nabla_{a \beta^{\intercal}} M_{i_1 \dots i_K}(\bar \beta, \bar a_{i_1 \dots i_K}), \notag \\ &\text{while the rest elements are } 0, \notag \\
   J_{N, g_1 \dots g_K}^{\alpha_k \alpha_l} := \ & G_k \times G_{l} \ \text{matrix whose } (g_k, g_l)\text{-th element is} \notag \\ &\mathbb P_N 1 \{ g^0_{i_1 \dots i_K} = (g_k)_{k=1}^K\} \nabla_{a a} M_{i_1 \dots i_K}(\bar \beta, \bar a_{i_1 \dots i_K}), \notag \\
    &\text{while the rest elements are } 0. \notag
\end{align}
Here, $\bar \beta$ and $\bar a_{i_1 \dots i_K}$ are defined similarly as in $\mathcal J_N$.

We consider the convergence of $\mathcal J_{N, g_1 \dots g_K}$ for each $g_k = 1, \dots, G_k \ (k = 1, \dots, K)$.

Define $J^{\beta \beta}_{g_1 \dots g_K}$ to be an upper-left $p \times p$ submatrix of $J_{g_1 \dots g_K}$ in Assumption \ref{as m added}$(f)$.
We first show that the submatrix $J^{\beta \beta}_{N, g_1 \dots g_K}$ of $\mathcal J_{N, g_1 \dots g_K}$ converges in probability to $J^{\beta \beta}_{g_1 \dots g_K}$. Collect $g = (g_1, \dots g_K)^{\intercal}$. The triangle inequality gives that 
\begin{align}
    &\|J^{\beta \beta}_{N, g_1 \dots g_K} - J^{\beta \beta}_{g_1 \dots g_K}\| \notag \\ 
    \leq \ &\left\|\mathbb P_N 1 \{ g^0_{i_1 \dots i_K} = (g_k)_{k=1}^K\} \nabla_{\beta \beta} M_{i_1 \dots i_K} (\bar \beta, \bar a_{i_1 \dots i_K}) \right. \notag \\
     & - \left. \frac{1}{N} \sum_{i_1, \dots i_K} 1 \{ g^0_{i_1 \dots i_K} = (g_k)_{k=1}^K\} \mathbb E_{\gamma_0} [\nabla_{\beta \beta}M_{i_1 \dots i_K} (\beta, a)]|_{\beta = \bar \beta, a = \bar a_{i_1 \dots i_K}}  \right\| \notag \\
     & + \left\| \frac{1}{N} \sum_{i_1 \dots i_K} 1 \{ g^0_{i_1 \dots i_K} = (g_k)_{k=1}^K\} (\mathbb E_{\gamma_0} [\nabla_{\beta \beta}M_{i_1 \dots i_K} (\beta, a)]|_{\beta = \bar \beta, a = \bar a_{i_1 \dots i_K}} \right. \notag \\
     & \left. - \mathbb E_{\gamma_0} [\nabla_{\beta \beta} M_{i_1 \dots i_K} (\beta^0, \tilde a^0_{i_1 \dots i_K})])  \right\| \notag \\
     & + \left\| \frac{1}{N} \sum_{i_1 \dots i_K} 1 \{ g^0_{i_1 \dots i_K} = (g_k)_{k=1}^K \} \mathbb E_{\gamma_0} [\nabla_{\beta \beta}M_{i_1 \dots i_K} (\beta^0, \tilde a^0_{i_1 \dots i_K})] - J^{\beta \beta}_{g_1 \dots g_K} \right\|. \label{hessian triangle}
\end{align}
We show the convergence of each of the three terms on the right side.
The first term is bounded, up to a finite constant, by
\begin{equation}
    \sup_{\beta \in \mathcal B, \alpha \in  \mathcal A^{\prod}} \| (\mathbb P_N - P) 1 \{ g^0_{i_1 \dots i_K} = g \} \nabla_{\beta \beta} M_{i_1 \dots i_K} (\beta, a_{1}(g^0_{1, i_1}) + \dots + a_{K}(g^0_{K, i_K}))\|. \notag
\end{equation}
This term converges to zero in probability by a similar argument to step 2 of the proof of Theorem 1, in conjunction with Assumption \ref{as m added}$(g)$. 
Subsequently, it follows from Assumption \ref{as m added}$(d)$ that the second term on the right side of \eqref{hessian triangle} is bounded by 
\[\mathcal C_3 \left( \| \bar \beta - \beta^0 \| + \sum_{k = 1}^K \| \bar a_{k}(g_k) - a^0_{k}(g_k)\| \right),\] which converges to zero in probability by the definition of $\bar \beta$ and $\bar a_{k}(g_k)$. Finally, the convergence of the third term on the right side of \eqref{hessian triangle} follows from Assumption \ref{as m added}$(f)$. Now the convergence of $J^{\beta \beta}_{N, g_1 \dots g_N}$ to $J^{\beta \beta}_{g_1 \dots g_K}$ has been shown.

Similarly, let $J^{\beta a}_{g_1 \dots g_K}$ and $J^{a a}_{g_1 \dots g_K}$ be an upper-right $p \times 1$  and a lower-right $1 \times 1$ submatrices of $J_{g_1 \dots g_k}$ in Assumption \ref{as m added}$(f)$, respectively. Then by a similar argument to the above, submatrices $J^{\beta \alpha_k}_{N, g_1 \dots g_K}$, $J^{\alpha_k \beta}_{N, g_1 \dots, g_K}$, and $J_{N, g_1 \dots g_K}^{\alpha_k \alpha_l}$ of $\mathcal J_{N, g_1 \dots g_K}$ converge in probability to $J^{\beta \alpha_k}_{g_1 \dots g_K}$, $J^{\alpha_k \beta}_{g_1 \dots g_K}$, and $J_{g_1 \dots g_K}^{\alpha_k \alpha_l}$, defined below, respectively:
\begin{align}
    J^{\beta \alpha_k}_{g_1 \dots g_K} := \ & p \times G_k \ \text{matrix whose } g_k \text{-th column is } J_{g_1 \dots g_K}^{\beta a}, \text{while the rest elements are } 0, \notag \\
    J^{\alpha_k \beta}_{g_1 \dots g_K} := \ & G_k \times p  \ \text{matrix whose } g_k \text{-th row is } (J_{g_1 \dots g_K}^{\beta a})^{\intercal},  \text{while the rest elements are } 0, \notag \\
    J_{g_1 \dots g_K}^{\alpha_k \alpha_l} := \ & G_k \times G_{l} \ \text{matrix whose } (g_k, g_l)\text{-th element is } J^{a a}_{g_1 \dots g_K}, \notag \\
    &\text{while the rest elements are } 0. \notag
\end{align}

Combining these convergence results for $J_{N, g_1 \dots g_N}^{\beta \beta}$, $J^{\beta \alpha_k}_{N, g_1 \dots g_K}$, $J^{\alpha_k \beta}_{N, g_1 \dots g_K}$, and $J^{\alpha_k \alpha_l}_{N, g_1 \dots g_K}$, we conclude that $\mathcal J_{N, g_1 \dots g_K}$ converges in probability to the following matrix:
\begin{equation}
    \mathcal J_{g_1 \dots g_K} =
    \begin{pmatrix}
        J^{\beta \beta}_{g_1 \dots g_K} & J^{\beta \alpha_1}_{g_1 \dots g_K} & \cdots & J^{\beta \alpha_K}_{g_1 \dots g_K} \\
        J^{\alpha_1 \beta}_{g_1 \dots g_K} & J^{\alpha_1 \alpha_1}_{g_1 \dots g_K} & \cdots & J^{\alpha_1 \alpha_K}_{g_1 \dots g_K} \\
        \vdots & \vdots & \ddots & \vdots \\
        J^{\alpha_K \beta}_{g_1 \dots g_K} & J^{\alpha_K \alpha_1}_{g_1 \dots g_K} & \cdots & J^{\alpha_K \alpha_K}_{g_1 \dots g_K}
    \end{pmatrix}. \label{I g_1, dots g_K}
\end{equation}
Accordingly, $\mathcal J_N$ converges in probability to $\mathcal J := \sum_{g_1 = 1}^{G_1} \dots \sum_{g_K = 1}^{G_K} \mathcal J_{g_1 \dots g_K}$.

We move on to the convergence of $\lambda \mathcal J_{Pen,  N}$. Here $\mathcal J_{Pen, N}$ is a symmetric matrix whose upper-right part has the following form:
\begin{equation}
\begin{pmatrix}
        0 & 0 &0 & 0 & \cdots & \cdots & 0 \\
        \ddots & J^{\alpha_1 \alpha_1}_{ Pen, N} & J^{\alpha_1 \alpha_2}_{Pen, N} & 0 & \cdots & \cdots & 0  \\ 
        \ddots & \ddots & J^{\alpha_2 \alpha_2}_{Pen, N} & J^{\alpha_2 \alpha_3}_{Pen, N} & 0 & \cdots & 0  \\
        \ddots & \ddots & \ddots & \ddots & \ddots & \ddots & \vdots  \\
        \ddots & \ddots & \ddots & \ddots & \ddots & \ddots & 0   \\
        \ddots & \ddots & \ddots & \ddots & \ddots & J^{\alpha_{K-1} \alpha_{K-1}}_{Pen, N} & J^{\alpha_{K-1} \alpha_K}_{Pen, N} \\
        \ddots & \ddots & \ddots & \ddots & \ddots & \ddots & J^{\alpha_K \alpha_K}_{Pen, N}
    \end{pmatrix}, \label{H N penalty}
\end{equation}
where, for $k = 1, \dots, K$, $J_{Pen, N}^{\alpha_k \alpha_k}$ is a $G_k \times G_k$ matrix whose $(g, h)$-th entry is given by
\[
2^{1 \{ 2 \leq k \leq K-1 \}} \left( \frac{1}{n_k} \sum_{i_k = 1}^{n_k} 1 \{ g^0_{k, i_k} = g \} \right) \left(  \frac{1}{n_k} \sum_{i_k = 1}^{n_k} 1 \{ g^0_{k, i_k} = h \} \right),
\]
while, for $l = 1, \dots, K-1$, $J_{ Pen, N}^{\alpha_l \alpha_{l+1}}$ is a $G_{l} \times G_{l+1}$ matrix whose $(g, h)$-th entry is given by
\[
    -  \left( \frac{1}{n_l} \sum_{i_l = 1}^{n_l}  1 \{ g^0_{l, i_l} = g \} \right) \left( \frac{1}{n_{l+1}} \sum_{i_{l+1}}^{n_{l+1}} 1 \{ g^0_{l+1, i_{l+1}} = h \} \right).
\]
It then follows from Assumption \ref{as group N}$(a)$ that $\lambda \mathcal J_{Pen, N}$ converges in probability to a symmetric matrix $\lambda \mathcal J_{Pen}$ with
\begin{equation}
\mathcal J_{Pen} :=
\begin{pmatrix}
        0 & 0 &0 & 0 & \cdots & \cdots & 0 \\
        \ddots & J^{\alpha_1 \alpha_1}_{Pen} & J^{\alpha_1 \alpha_2}_{Pen} & 0 & \cdots & \cdots & 0  \\ 
        \ddots & \ddots & J^{\alpha_2 \alpha_2}_{Pen} & J^{\alpha_2 \alpha_3}_{Pen} & 0 & \cdots & 0  \\
        \ddots & \ddots & \ddots & \ddots & \ddots & \ddots & \vdots  \\
        \ddots & \ddots & \ddots & \ddots & \ddots & \ddots & 0   \\
        \ddots & \ddots & \ddots & \ddots & \ddots & J^{\alpha_{K-1} \alpha_{K-1}}_{Pen} & J^{\alpha_{K-1} \alpha_K}_{Pen} \\
        \ddots & \ddots & \ddots & \ddots & \ddots & \ddots & J^{\alpha_K \alpha_K}_{Pen}
    \end{pmatrix}, \label{H pen}
\end{equation}
where, for $k=1, \dots, K$, $J^{\alpha_k \alpha_k}_{Pen}$ is a $G_k \times G_k$ matrix whose $(g, h)$-th element is given by $2^{1 \{ 2 \leq k \leq K-1 \}}  \pi_{k, g} \pi_{k, h}$, while, for $l = 1, \dots, K-1$, $H^{\alpha_l \alpha_{l+1}}_{Pen}$ is a $G_l \times G_{l+1}$ matrix whose $(g, h)$-th entry is given by $- \pi_{l, g} \pi_{l+1, h}$.

\underline{$(iii)$ weak convergence of $\sqrt{N} (\hat \theta - \theta^0)$.}
We first show the negative definiteness of $\mathcal J - \lambda \mathcal J_{Pen}$. Let $b_c := (b^{\intercal}, c_1^{\intercal}, \dots, c_K^{\intercal})^{\intercal}$ be a real vector such that $b \in \mathbb R^{p}$ and $c_k := (c_{k, 1}, \dots c_{k, G_k})^{\intercal} \in \mathbb R^{G_k}$ $(k = 1, \dots, K)$. Collecting $\pi_k := (\pi_{k, 1}, \dots, \pi_{k, G_k})^{\intercal}$ for $k = 1, \dots, K$, a straightforward calculation shows that
\begin{equation}
    b_c^{\intercal} \lambda \mathcal J_{Pen} b_c = \lambda \sum_{k=1}^{K-1} \{ (c_k^{\intercal} \pi_k)^2 - 2 (c_k^{\intercal} \pi_k) (c_{k+1}^{\intercal} \pi_{k+1}) + (c_{k+1}^{\intercal} \pi_{k+1})^2 \} \geq 0. \label{pen mat positive}
\end{equation}
Furthermore, we have that $b_c^{\intercal} \mathcal J b_c = \sum_{g_1 = 1}^{G_1} \dots \sum_{g_K = 1}^{G_K} b_c^{\intercal} \mathcal J_{g_1 \dots g_K} b_c$. For any $g_1, \dots, g_K$ with $g_k = 1, \dots, G_k \ (k = 1, \dots, K)$, in view of \eqref{I g_1, dots g_K}, we have that
\begin{equation}
b_c^{\intercal} \mathcal J_{g_1 \dots g_K} b_c = b^{\intercal} J^{\beta \beta}_{g_1 \dots g_K} b + 2 b^{\intercal} J^{\beta a}_{g_1 \dots g_K} \left( \sum_{k = 1}^K c_{k, g_k} \right) + J^{aa}_{g_1 \dots g_K} \left( \sum_{k = 1}^K c_{k, g_k} \right)^2. \label{g mat positive}
\end{equation}
Because $J^{\beta \beta}_{g_1 \dots g_K}$, $J^{\beta a}_{g_1 \dots g_K}$, and $J^{aa}_{g_1 \dots g_K}$ are submatrices of $J_{g_1 \dots g_K}$ in Assumption \ref{as m added}$(f)$, the above display is non-positive, and zero only if $b = 0$ and $\sum_{k = 1}^K c_{k, g_k} = 0$ by the negative definiteness of $J_{g_1 \dots g_K}$. Therefore, by \eqref{pen mat positive} and \eqref{g mat positive}, $b_c^{\intercal} (\mathcal J - \lambda \mathcal J_{Pen}) b_c \leq 0$, and equality holds only if $b = 0$ and $\sum_{k = 1}^K c_{k, g_k} = 0$ for any $g_1, \dots, g_K \ (g_k = 1, \dots, G_k, k = 1, \dots, K)$. Assume that $b_c^{\intercal} (\mathcal J - \lambda \mathcal J_{Pen}) b_c = 0$, and thus $b = 0$ and $\sum_{k = 1}^K c_{k, g_k} = 0$ for all $g_1, \dots, g_K$. By a straightforward calculation, observe that $c_{k, 1} = \dots = c_{k, G_k}$ for any $k = 1, \dots, K$. Because $\pi_{k, 1} + \dots + \pi_{k, G_k} = 1$ and $b_c^{\intercal} \lambda \mathcal J_{Pen} b_c = 0$ by assumption, \eqref{pen mat positive} leads to $\sum_{k = 1}^{K-1} (c_{k, 1} - c_{k+1, 1})^2 = 0$. This implies that all the elements of $c_1, \dots, c_K$ are the same. In view of $\sum_{k = 1}^K c_{k, g_k} = 0$, we conclude that $c_k = 0$ for all $k = 1, \dots, K$, and the negative definiteness of $\mathcal J - \lambda \mathcal J_{Pen}$ is proven.

Because $\mathcal J - \lambda \mathcal J_{Pen}$ is negative definite and thus nonsingular, \eqref{foc score} and steps $(i)$ and $(ii)$ in conjunction with Slutsky's theorem yield that 
\[
    \sqrt{N} (\hat \theta - \theta^0) \rightarrow_d N(0, (\mathcal J - \lambda \mathcal J_{Pen})^{-1} \mathcal L \mathcal D \mathcal L^{\intercal} (\mathcal J - \lambda \mathcal J_{Pen})^{-1} ).
\]
This finishes step $(iii)$.

We conclude the proof by showing that the asymptotic covariance matrix of $\sqrt{N} (\hat \beta - \beta^0)$ is positive definite. Note that this matrix is an upper left $p \times p$ submatrix of $(\mathcal J - \lambda \mathcal J_{Pen})^{-1} \mathcal L \mathcal D \mathcal L^{\intercal} (\mathcal J - \lambda \mathcal J_{Pen})^{-1}$. Let $b_0 := \left( b^{\intercal}, 0^{\intercal}_{\sum_{k=1}^{K} G_k \times 1} \right)$ be a $p + \sum_{k = 1}^{K} G_k$-dimensional real vector such that $b \neq 0$.
Suppose that $L_{1}$ and $L_{2}$ are the $p \times p$ upper-left and the $\sum_{k = 1}^K G_k \times p$ lower-left submatrices of $(\mathcal J - \lambda \mathcal J_{Pen})^{-1}$, respectively. Then we have
\begin{equation}
    \mathcal L^{\intercal} (\mathcal J - \lambda \mathcal J_{Pen})^{-1}b_0 = \mathcal L^{\intercal}((L_1 b)^{\intercal}, (L_2 b)^{\intercal})^{\intercal} = \mathcal L_{\beta}^{\intercal}  L_1 b + \mathcal L_{\alpha}^{\intercal} L_2 b. \label{matrix manipulation}
\end{equation}
As $(\mathcal J - \lambda \mathcal J_{Pen})^{-1}$ is negative definite, so is $L_1$. Hence, $L_1 b \neq 0$, by which $\mathcal L_{\beta}^{\intercal} L_1 b \neq 0$ follows because the columns of $\mathcal L_\beta^{\intercal}$ are linearly independent from the expression \eqref{lambda beta}. By the expressions \eqref{lambda beta} and \eqref{lambda alpha}, $\mathcal L_{\beta}^{\intercal} L_1 b$ is not in the column space of $\mathcal L_{\alpha}^{\intercal}$. Thus, $\mathcal L_{\beta}^{\intercal} L_1 b + \mathcal L_{\alpha}^{\intercal} L_2 b \neq 0$. 
In view of \eqref{matrix manipulation} and the positive definiteness of $\mathcal D$ from Assumption \ref{as m added}$(f)$, we have that $b_0^{\intercal} (\mathcal J - \lambda \mathcal J_{Pen})^{-1} \mathcal L \mathcal D \mathcal L^{\intercal} (\mathcal J - \lambda \mathcal J_{Pen})^{-1} b_0 > 0$. This shows the positive definiteness of the asymptotic covariance matrix of $\sqrt{N} (\hat \beta - \beta^0)$, and the proof is complete.
\end{proof}

\section{Auxiliary results}\label{app:lemma}

Lemmas \ref{inverse of W} is about the normalization matrix $W$ defined in Example 1 of the main text.
\begin{lem}\label{inverse of W}
    Let $\tilde W = (\tilde w_1^{\intercal}, \dots, \tilde w_K^{\intercal})^{\intercal}$ be a $k \times k$ matrix whose row vectors $\tilde w_k$ has the following expression: $\tilde w_1 = (1/K, 1 - 1/K, 1-2/K, \dots, 1 - (K-2)/K, 1/K)$,
    $\tilde w_k = (1/K, -1/K, \dots, -(k-1)/K, 1 - k/K, \dots, 1 - (K-2)/K, 1/K)$ for $k = 2, \dots, K-1$, and $\tilde w_K = (1/K, -1/K, \dots, -(K-1)/K)$. Then $\tilde W W = I_{K}$, where $W$ is defined in Example 1 in the main text.
\end{lem}
\begin{proof}
    This can be verified by simple matrix multiplication $\tilde W W$.
\end{proof}

Lemma \ref{normalization} suggests that the normalization imposed by the penalty term set $Pen (\alpha^0, \gamma^0)$ exactly to zero.
This lemma is useful in the proof of Theorem 1. 

\begin{lem} \label{normalization}
	Assume Assumptions \ref{as parameter space} and \ref{as m}. Then $Pen ( \alpha^0, \gamma^0) = 0$.
\end{lem}
\begin{proof} 
    For a real number $\eta$, we define a function:
    \begin{align}
        U_1 (\eta) := \ &P M_{i_1 \dots i_K}\left(\beta^0, \left( a^0_{1}(g^0_{1, i_1}) + \eta \right) + \left(a^0_{2}(g^0_{2, i_2}) - \eta \right) + a^0_{3}(g^0_{3, i_3}) + \dots + a^0_{K}(g^0_{K, i_K}) \right) \notag \\
        &- \frac{\lambda}{2} \left( \frac{1}{n_1} \sum_{i_1 = 1}^{n_1} \left( a^0_{1}(g^0_{1, i_1}) + \eta \right) - \frac{1}{n_2} \sum_{i_2 = 1}^{n_2} \left(a^0_{2}(g^0_{2, i_2}) - \eta\right)\right)^2 \notag \\
        &- \frac{\lambda}{2} \left( \frac{1}{n_2} \sum_{i_2 = 1}^{n_2} \left(a^0_{2}(g^0_{2, i_2}) - \eta\right) - \frac{1}{n_3} \sum_{i_3 = 1}^{n_3} a^0_{3} (g^0_{3, i_3})\right)^2. \notag
    \end{align}
    %where refer to the begining of Appendix A for the abbreviation $M_{i_1 \dots i_K}$.
    By the definition of $\theta^0$ and Assumption \ref{as parameter space}$(c)$, $U_1(\eta)$ attains its maximum at $\eta = 0$ in some neighborhood of $0$. Hence, the first order condition gives that
    \begin{equation}
        0 = 2 \left( \frac{1}{n_1} \sum_{i_1 = 1}^{n_1} a^0_{1}(g^0_{1, i_1})  - \frac{1}{n_2} \sum_{i_2 = 1}^{n_2} a^0_{2}(g^0_{2, i_2}) \right) 
        - \left( \frac{1}{n_2} \sum_{i_2  = 1}^{n_2} a^0_{2}(g^0_{2, i_2}) - \frac{1}{n_3} \sum_{i_3 = 1}^{n_3} a^0_{3}(g^0_{3, i_3}) \right). \label{FOC 12}
    \end{equation}
    For $k = 2, \dots, K -2$, define a function $U_k (\eta)$ as
    \begin{align}
        &P M_{i_1 \dots i_K} \left(\beta^0, a^0_{1}(g^0_{1, l_1}) + 
        \dots + 
        \left( a^0_{k} (g^0_{k, i_k}) + \eta \right) + \left(a^0_{k + 1}(g^0_{k + 1, i_{k + 1}}) - \eta\right) + \dots + a^0_{K}(g^0_{K, i_K}) \right) \notag \\
        &- \frac{\lambda}{2} \left( \frac{1}{n_{k-1}} \sum_{i_{k - 1} = 1}^{n_{k - 1}} a^0_{k - 1}(g^0_{k - 1, i_{k - 1}})  - \frac{1}{n_k}\sum_{i_k = 1}^{n_k} \left(a^0_{k}(g^0_{k, i_k}) + \eta\right)   \right)^2 \notag \\
        &- \frac{\lambda}{2} \left( \frac{1}{n_k}\sum_{i_{k } = 1}^{n_{k}} \left(a^0_{k}(g^0_{k, i_{k}}) + \eta\right)  - \frac{1}{n_{k+1}} \sum_{i_{k + 1} = 1}^{n_{k + 1}} \left(a^0_{k + 1}(g^0_{k + 1, i_{k + 1}}) - \eta\right)   \right)^2 \notag \\
        &- \frac{\lambda}{2} \left( \frac{1}{n_{k+1}}\sum_{i_{k + 1} = 1}^{n_{k + 1}} \left(a^0_{k + 1}(g^0_{k + 1, i_{k + 1}}) - \eta\right)  - \frac{1}{n_{k + 2}}\sum_{i_{k + 2} = 1}^{n_{k + 2}} \alpha^0_{k + 2}(g^0_{k + 2, i_{k + 2}})    \right)^2. \notag
    \end{align}
    Similarly to $U_{1}(\eta)$, we obtain the following equation for each $k = 2, \dots, K-2$: 
    \begin{align}
       0 = &- \left( \frac{1}{n_{k-1}}\sum_{i_{k - 1} = 1}^{n_{k - 1}} a^0_{k - 1}(g^0_{k - 1, i_{k - 1}})  - \frac{1}{n_{k}} \sum_{i_k = 1}^{_k} a^0_{k}(g^0_{k, i_k})   \right) \notag \\
        &+ 2 \left( \frac{1}{n_k}\sum_{i_{k } = 1}^{n_{k}} a^0_{k}(g^0_{k, i_{k}})  - \frac{1}{n_{k+1}}\sum_{i_{k + 1} = 1}^{n_{k + 1}} a^0_{k + 1}(g^0_{k + 1, i_{k + 1}})    \right) \notag \\
        &- \left( \frac{1}{n_{k+1}} \sum_{i_{k + 1} = 1}^{n_{k + 1}} a^0_{k + 1}(g^0_{k + 1, i_{k + 1}}) - \frac{1}{n_{k+2}} \sum_{i_{k + 2} = 1}^{n_{k + 2}} a^0_{k + 2}(g^0_{k + 2, i_{k + 2}})    \right). \label{FOC kk+1}
    \end{align}
    Likewise, we define $U_{K-1}(\eta)$ as
    \begin{align}
        &P M_{i_1 \dots i_K}\left(\beta^0, a^0_{1}(g^0_{1, i_1}) + \dots + \left(a^0_{K - 1}(g^0_{K - 1, i_{K - 1}}) + \eta\right)  + \left(a^0_{K}(g^0_{K, i_K}) - \eta\right) \right) \notag \\
        &- \frac{\lambda}{2} \left( \frac{1}{n_{K-2}} \sum_{i_{K - 2} = 1}^{n_{K - 2}} a^0_{K - 2}(g^0_{K - 2, i_{K - 2}}) - \frac{1}{n_{K-1}}\sum_{i_{K - 1} = 1}^{n_{K - 1}} (a^0_{K - 1}(g^0_{K - 1, i_{K - 1}}) + \eta)\right)^2 \notag \\ 
        &- \frac{\lambda}{2} \left( \frac{1}{n_{K-1}} \sum_{i_{K - 1} = 1}^{n_{K - 1}} (a^0_{K - 1}(g^0_{K - 1, i_{K - 1}}) + \eta) - \frac{1}{n_K} \sum_{i_K = 1}^{n_K} ( a^0_{K}(g^0_{K, i_K}) - \eta) \right)^2. \notag
    \end{align}
    By a similar reasoning to the above, we obtain the equation:
    \begin{align}
        0 = &- \left( \frac{1}{n_{K-2}}\sum_{i_{K - 2} = 1}^{n_{K - 2}} a^0_{K - 2}(g^0_{K - 2, i_{K - 2}}) - \frac{1}{n_{K-1}}\sum_{i_{K - 1} = 1}^{n_{K - 1}} a^0_{K - 1}(g^0_{K - 1, i_{K - 1}}) \right) \notag \\
        &+2 \left( \frac{1}{n_{K-1}}\sum_{i_{K - 1} = 1}^{n_{K - 1}} a^0_{K - 1}(g^0_{K - 1, i_{K - 1}}) - \frac{1}{n_K} \sum_{i_K = 1}^{n_K} a^0_{K}(g^0_{K, i_K})  \right). \label{FOC K-1K}
    \end{align}
    Stacking the equations in  \eqref{FOC 12}, \eqref{FOC kk+1}, and \eqref{FOC K-1K}, we obtain the system of equations $U \bar a = 0$, where $U$ is a $K-1$ by $K-1$ symmetric matrix in which the diagonal elements are all $2$, $(k, k + 1)$-th element is $-1$ for $k = 1, \dots, K - 2$, and the other $(l, m)$-th elements are zero $(l < m)$, while
    $\bar a$ is a $(K - 1)$-dimensional vector with its $k$-th element being $\frac{1}{n_k} \sum_{i_k = 1}^{n_k} a^0_{k}(g^0_{k, i_k}) - \frac{1}{n_{k + 1}}\sum_{i_{k + 1} = 1}^{n_{k + 1}} a^0_{k + 1}(g^0_{k + 1, i_{k + 1}})$.
    We show that $U$ is positive-definite, and thus nonsingular.
    Let $a^* = (a^*_1, \dots, a^*_{K - 1})'$ be a real vector such that $a^* \neq 0$. By a straightforward calculation, it can be shown that 
    \begin{align}
        {a^*}' U a^* &= 2 {a_1^*}^2 + \dots + 2 {a_{K - 1}^*}^2 - 2(a_{1}^* a_{2}^* + a_{2}^* a_{3}^* + \dots + a_{K-3}^* a_{K-2}^* + a_{K - 2}^* a_{K - 1}^*) \notag \\
        &= {a_1^*}^2 + (a_1^* - a_2^*)^2 + (a_2^* - a_3^*)^2 + \dots + (a_{K-3}^* - a_{K-2}^*)^2 + (a_{K - 2}^* - a_{K - 1}^*)^2 + {a_{K - 1}^*}^2, \notag
    \end{align}
    which is positive as long as $a^* \neq 0$. Hence, $U$ is non-singular and $a^* = 0$ follows.
    This complete the proof.
\end{proof}

\begin{lem}\label{normalization estimator}
    Assume Assumptions \ref{as parameter space} and \ref{as m}. If $\hat \alpha$ lies in an interior of $\mathcal A^{\prod}$, $Pen (\hat \gamma, \hat \alpha) = 0$.
\end{lem}
\begin{proof}
    The result follows by repeating the proof of Lemma \ref{normalization} with $P M_{i_1 \dots i_K}$ replaced by $\mathbb P_N M_{i_1 \dots i_K}$.
\end{proof}

\begin{lem}\label{lem covariance penalty}
    Assume the assumption of Theorem 2, and that $\mathcal D$ in the proof of Theorem 2 does not depend on the tuning parameter $\lambda$. Then the asymptotic distribution of $\sqrt{N} (\hat \theta - \theta^0)$ does not depend on the value of $\lambda$.
\end{lem}
\begin{proof}
    Fix $\lambda_1, \lambda_2 > 0$ with $\lambda_1 \neq \lambda_2$. Let $\hat \theta_{\lambda_j}$ denote the estimator $\hat \theta$ constructed with $\lambda_j$ as tuning parameters for $j = 1, 2$. Then the asymptotic distributions of $\sqrt{N} (\hat \theta_{\lambda_j} - \theta^0)$ is $N (0, (\mathcal J - \lambda_j H_{Pen})^{-1} \mathcal L D \mathcal L^{\intercal} (\mathcal J - \lambda_j \mathcal J_{Pen})^{-1})$ for $j = 1, 2$ by Theorem 2. Because $\mathcal D$ is invariant to the value of tuning parameter, it suffices to show that $\mathcal L^{\intercal} (\mathcal J - \lambda_1 \mathcal J_{Pen})^{-1} =  \mathcal L^{\intercal} (\mathcal J - \lambda_2 \mathcal J_{Pen})^{-1}$ for the desired result. Consider a matrix-valued function $\lambda \mapsto \mathcal L^{\intercal} (\mathcal J - \lambda \mathcal J_{Pen})^{-1}$. This function is differentiable in $\lambda > 0$ and its derivative is $\mathcal L^{\intercal}  (\mathcal J - \lambda \mathcal J_{Pen})^{-1} \mathcal J_{Pen}  (\mathcal J - \lambda \mathcal J_{Pen})^{-1}$. We show that this derivative is zero so that the function is constant.
    For any $\lambda > 0$, it follows from Lemma \ref{normalization} and \ref{normalization estimator} that $\mathcal S^{Pen} (\theta^0) = \mathcal S^{Pen} (\hat \theta_\lambda) = 0$ and, thus, $\lambda \mathcal J_{Pen, N} \sqrt{N} (\hat \theta_{\lambda}  - \theta^0) = 0$ from Taylor's theorem with probability approaching one. Combining this fact with the proof of Theorem 2 and the Slutsky's theorem, we have that $\lambda \mathcal J_{Pen, N} \sqrt{N} (\hat \theta_\lambda - \theta^0) \rightarrow_d N(0, \lambda \mathcal J_{Pen} (\mathcal J - \lambda \mathcal J_{Pen})^{-1} \mathcal L \mathcal D \mathcal L^{\intercal} (\mathcal J - \lambda \mathcal J_{Pen})^{-1} \lambda \mathcal J_{Pen}) = 0$. Because the asymptotic covariance matrix is zero and $\mathcal D$ is positive definite, $ \mathcal L^{\intercal} (\mathcal J - \lambda \mathcal J_{Pen})^{-1} \mathcal J_{Pen} = 0$, which shows that the derivative is zero. This completes the proof.
\end{proof}

\section{Sufficient conditions in the case of GLMs}\label{app:glm}
This section discusses how the assumptions for the theorems in the main text holds in three examples of GLMs (normal, logistic, and Poisson regression). The conditional independence assumption, the rate condition on the sample size, and Assumption \ref{as parameter space} and \ref{as group N} are what we need to impose regardless of the model. Hence, we here examine how Assumptions \ref{as m} and \ref{as m added} are satisfied.

\begin{exm}[Normal regression]  \label{exm glm normal as}
    The derivatives of the objective function up to the second order are given as follows:
    \begin{align}
        \nabla_{\beta_a}M(z_{i_1 \dots i_K} | \beta, a) 
        = \ & y_{i_1 \dots i_K} x^{+}_{i_1 \dots i_K}  - x^+_{i_1 \dots i_K} (x_{i_1 \dots i_K}^{\intercal} \beta + a), \notag \\
        \nabla_{\beta_a \beta_a} M(z_{i_1 \dots i_K} | \beta, a)
        = & - x^+_{i_1 \dots i_K} {x_{i_1 \dots i_K}^+}^{\intercal}, \notag
    \end{align}
    where $x_{i_1 \dots i_K}^+ = (x_{i_1 \dots i_K}^{\intercal}, 1)^{\intercal}$.
    In view of the above display and Example 2 in the main text, the boundedness of $x_{i_1 \dots i_K}$ in conjunction with Assumption \ref{as parameter space} is sufficient for the differentiability in Assumption \ref{as m}$(a)$, Assumption \ref{as m}$(c)$ and $(d)$, and \ref{as m added}$(a)$, $(b)$, $(c)$, $(d)$, $(e)$ and $(g)$. In particular, for Assumption \ref{as m}$(d)$, we can take an envelope function $M_{env} (z)$ such that $|M(z | \beta, a) - M(z|\bar \beta, \bar a)| \leq  M_{env}(z) \| (\beta^{\intercal}, a)^{\intercal} - (\bar \beta^{\intercal}, \bar a)^{\intercal} \|$ for any 
    $(\beta^{\intercal}, a)^{\intercal}$ and $(\bar \beta^{\intercal}, \bar a)^{\intercal}$ in $\mathcal B \times \mathcal A^{+}$. It follows from Lemma 26 of \cite{katoep} that $\sup_{R} \mathcal N(\ep \| M_{env} \|_{R, 2}, \mathcal M, L_2(R))$ is bounded by $\ep^{-(p + 1)}$ up to a universal finite constant, which verifies Assumption \ref{as m}$(d)$. By a similar argument, we can verify Assumption \ref{as m added}$(g)$. 
    The other part of Assumption \ref{as m}$(a)$ and Assumption \ref{as m}$(b)$ are satisfied if the maximum eigenvalue of $\mathbb E[- x_{i_1 \dots i_K}^+ {x_{i_1 \dots i_K}^+}^{\intercal}|\gamma^0] = \mathbb E[- x_{i_1 \dots i_K}^+ {x_{i_1 \dots i_K}^+}^{\intercal}|(g^0_{k, i_k})_{k = 1}^K]$ is bounded by zero uniformly over $i_1 \dots i_K$. Given that the number of points that $(g^0_{k, i_k})_{k = 1}^K$ can take is finite, this uniformity requirement is not so stringent as it may appear.
    For Assumption \ref{as m added}$(f)$, note that 
    \[-\psi^0 D^{\beta_a}_{i_1 \dots i_K} (\beta^0, a^0_{i_1 \dots i_K}) = J^{\beta_a}_{i_1 \dots i_K} (\beta^0, a^0_{i_1 \dots i_K})\] holds with $a^0_{i_1 \dots i_K} := a^0_{1}(g^0_{1, i_1}) + \dots + a^0_{K}(g^0_{K, i_K})$. 
    Hence, we only need to assume the convergence condition for $J^{\beta_a}_{i_1 \dots i_K} (\beta^0, a^0_{i_1 \dots i_K})$ in Assumption \ref{as m added}$(f)$, which is not too restrictive given that each $J^{\beta_a}_{i_1 \dots i_K} (\beta^0, a^0_{i_1 \dots i_K})$ is negative definite with the conditions already stated.
\end{exm}

\begin{exm}[Logistic regression] \label{exm glm logistic as}
The derivatives of the objective function up to the second order are given as follows:
    \begin{align}
        \nabla_{\beta_a} M(z_{i_1 \dots i_K} | \beta, a) 
        &=  y_{i_1 \dots i_K}  x^+_{i_1 \dots i_K}  - x^+_{i_1 \dots i_K} L(x_{i_1 \dots i_K}^{\intercal} \beta + a), \notag \\
        \nabla_{\beta_a \beta_a} M(z_{i_1 \dots i_K} | \beta, a)
        &= - x^+_{i_1 \dots i_K} {x_{i_1 \dots i_K}^+}^{\intercal} L' (x_{i_1 \dots i_K}^{\intercal} \beta + a), \notag
    \end{align}
    where $x^+_{i_1 \dots i_K} = (x_{i_1 \dots i_K}^{\intercal}, 1)^{\intercal}$.
    Similarly to Example \ref{exm glm normal as}, the boundedness of $x_{i_1 \dots i_K}$ in conjunction with Assumption \ref{as parameter space} is sufficient for the differentiability in Assumption \ref{as m}$(a)$, Assumption \ref{as m}$(c)$ and $(d)$, and \ref{as m added}$(a)$, $(b)$, $(c)$, $(d)$, $(e)$ and $(g)$. The other part of Assumption \ref{as m}$(a)$ and Assumption \ref{as m}$(b)$ are satisfied if the maximum eigenvalue of \[\mathbb E[- x^+_{i_1 \dots i_K}  {x_{i_1 \dots i_K}^+}^{\intercal} L' (x_{i_1 \dots i_K}^{\intercal} \beta + a)|\gamma^0] = \mathbb E[- x^+_{i_1 \dots i_K}  {x_{i_1 \dots i_K}^+}^{\intercal} L' (x_{i_1 \dots i_K}^{\intercal} \beta + a)|(g^0_{k, i_k})_{k = 1}^K]\] is bounded by zero uniformly over $i_1 \dots i_K$, and $\beta \in \mathcal B$ and $a \in \mathcal A$. Assumption \ref{as m added}$(f)$ can be verified in a similar manner to Example \ref{exm glm normal as}.
\end{exm}

\begin{exm}[Poisson regression] \label{exm glm poisson as} 
The derivatives of the objective function up to the second order are given as follows:
    \begin{align}
        \nabla_{\beta_a} M(z_{i_1 \dots i_K} | \beta, a)  
        &= y_{i_1 \dots i_K} x^+_{i_1 \dots i_K} - x^+_{i_1 \dots i_K} \exp(x_{i_1 \dots i_K}^{\intercal} \beta + a), \notag \\
        \nabla_{\beta_a \beta_a} M(z_{i_1 \dots i_K} | \beta, a) 
        &= -  x_{i_1 \dots i_K}^+ {x_{i_1 \dots i_K}^+}^{\intercal} \exp(x_{i_1 \dots i_K}^{\intercal} \beta + a), \notag
    \end{align}
    where $x^+_{i_1 \dots i_K} = (x_{i_1 \dots i_K}^{\intercal}, 1)^{\intercal}$.
    Similarly to Example \ref{exm glm normal as}, the boundedness of $x_{i_1 \dots i_K}$ in conjunction with Assumption \ref{as parameter space} is sufficient for the differentiability in Assumption \ref{as m}$(a)$, Assumption \ref{as m}$(c)$ and $(d)$, and \ref{as m added}$(a)$, $(b)$, $(c)$, $(d)$, $(e)$ and $(g)$. The other part of Assumption \ref{as m}$(a)$ and Assumption \ref{as m}$(b)$ are satisfied if the maximum eigenvalue of \[\mathbb E[- x^+_{i_1 \dots i_K}  {x_{i_1 \dots i_K}^+}^{\intercal} \exp(x_{i_1 \dots i_K}^{\intercal} \beta + a)|\gamma^0] = \mathbb E[- x^+_{i_1 \dots i_K}  {x_{i_1 \dots i_K}^+}^{\intercal} \exp(x_{i_1 \dots i_K}^{\intercal} \beta + a)|(g^0_{k, i_k})_{k = 1}^K]\] is bounded by zero uniformly over $i_1 \dots i_K$, and $\beta \in \mathcal B$ and $a \in \mathcal A$. Assumption \ref{as m added}$(f)$ can be verified in a similar manner to Example \ref{exm glm normal as}.
\end{exm}

\section{Optimization algorithm for ordered probit model with crossed effects}\label{app:algorithm}
We here provide details of the optimization algorithm for fitting the ordered probit model with crossed effects, used in Section~5 in the main document.
Instead of the maximization steps in Algorithm~1 in the main text, we use an one-step updating of the Newton-Raphson method. 
Note that the penalized likelihood function is given by 
\begin{align*}
Q(\Psi)=
\frac{1}{nm}\sum_{i=1}^N \log\left\{ \Phi(c_{y_i} - \eta_i) - \Phi(c_{y_i - 1} - \eta_i) \right\} - \frac{\lambda}{2}\left(\frac{1}{n}\sum_{\ell_{ai}=1}^n a(g_{\ell_{ai}}) - \frac{1}{m}\sum_{\ell_{bi}=1}^m b(h_{\ell_{bi}})  \right)^2,
\end{align*}
where $\eta_i = x_i^\top \beta + a (g_{\ell_{ai}}) + b (g_{\ell_{bi}})$, and $g_{\ell_{ai}}\in \{1,\ldots,G_a\}$ and $h_{\ell_{bi}}\in \{1,\ldots,G_b\}$ indicate which values in the supports random effects take.
Then, the first and second order partial derivatives of $Q(\Psi)$ with respect to $\beta$ are given by 
\begin{align*}
\frac{\partial Q(\Psi)}{\partial \beta} =
\sum_{i=1}^N u_i^{(0)}(\Psi)x_i,\ \ \ \ \ 
\frac{\partial^2 Q(\Psi)}{\partial \beta \partial \beta^\top} =
- \sum_{i=1}^N \left\{ u_i^{(0)}(\Psi)^2 + u_i^{(1)}(\Psi) \right\} x_i x_i^\top,
\end{align*}
where
$$
u_i^{(0)}(\Psi) 
= \frac{ \phi(c_{y_i - 1} - \eta_i) - \phi(c_{y_i} - \eta_i) }{ \Phi(c_{y_i} - \eta_i) - \Phi(c_{y_i - 1} - \eta_i) },
$$
and 
$$
u_i^{(1)}(\Psi)
=\frac{ (c_{y_i} - \eta_i)\phi(c_{y_i} - \eta_i) - (c_{y_i - 1} - \eta_i)\phi(c_{y_i - 1} - \eta_i) }{ \Phi(c_{y_i} - \eta_i) - \Phi(c_{y_i - 1} - \eta_i) }.
$$
Partial derivatives with respect to $a(1),\ldots,a(G_a)$ and $b(1),\ldots,b(G_b)$ can be obtained in the same way by using $u_i^{(0)}(\Psi)$ and $u_i^{(1)}(\Psi)$.  
Then, given the current value $\beta^{(s)}$, the new value $\beta^{(s+1)}$ is obtained as 
$$
\beta^{(s+1)}=\beta^{(s)} - \left\{\frac{\partial^2 Q(\Psi)}{\partial \beta \partial \beta^\top}\bigg|_{\Psi=\Psi^{(s)}}\right\}^{-1} \frac{\partial Q(\Psi)}{\partial \beta}\bigg|_{\Psi=\Psi^{(s)}},
$$
and the same for $a_g$ and $b_g$. 

\section{Rate condition}\label{app:rate}
In Theorem 2 of the main text, we assume the rate condition for sample sizes, that is, for any $k = 1, \dots, K$, and $l = 1, \dots, K$,
\begin{equation}
     \frac{\log n_l \log n_k}{\prod_{m=1, m \neq k}^K n_m} \rightarrow 0, \label{rate ours}
\end{equation}
and state that this condition is weaker than the rate assumption in Theorem 1 of \cite{Jiang2025aos}, and is not necessarily stronger than the rate condition in \cite{LyuSissonWelsh2024aos}. The purpose of this subsection is to give a justification for this statement.

We first discuss the comparison with \cite{Jiang2025aos}. Theorem 1 of \cite{Jiang2025aos} considers the case with $K = 2$ and assumes that $n_1, n_2 \rightarrow \infty$ such that
\begin{equation}
    \liminf \left( \frac{n_1}{n_2} \right) > 0, \quad \limsup \left( \frac{n_1}{n_2} \right) < \infty. \label{rate jiang 1}
\end{equation}
To show that the condition \eqref{rate ours} is weaker than the condition \eqref{rate jiang 1}, we assume \eqref{rate jiang 1} and are going to show that \eqref{rate ours} holds. Note that \eqref{rate ours} is shown if we show that all of the following four terms converge to zero
\begin{equation}
    \frac{(\log n_1)^2}{n_2}, \quad \frac{(\log n_2)^2}{n_1}, \quad 
    \frac{(\log n_1) (\log n_2)}{n_2}, \quad
     \frac{(\log n_1) (\log n_2)}{n_1}. \notag
\end{equation}
First, observe that $(\log n_1)^2/n_2 = (n_1 / n_2) (\log n_1 / \sqrt{n_1})^2$. By assumption \eqref{rate jiang 1}, $n_1/n_2$ is bounded, and $(\log n_1 / \sqrt{n_1})^2$ goes to zero as $n_1 \rightarrow \infty.$ Noting that \eqref{rate jiang 1} implies that 
\begin{equation}
    \liminf \left( \frac{n_2}{n_1} \right) > 0, \quad \limsup \left( \frac{n_2}{n_1} \right) < \infty, \notag
\end{equation}
$(\log n_2)^2/n_1 \rightarrow 0$ can be shown in the same manner.
Then, observe that $(\log n_1) (\log n_2)/n_2 = (n_1/n_2)^{1/2} (\log n_1 / \sqrt{n_1}) (\log n_2 / \sqrt{n_2})$. By assumption \eqref{rate jiang 1}, $(n_1/n_2)^{1/2}$ is bounded. Furthermore, $\log n_1 / \sqrt{n_1} \rightarrow 0$ and $\log n_2 / \sqrt{n_2} \rightarrow 0$ as $n_1, n_2 \rightarrow \infty$. This shows $(\log n_1) (\log n_2)/n_2 \rightarrow 0$. $(\log n_1) (\log n_2) / n_1 \rightarrow 0$ can be shown in the same manner.
We have now established \eqref{rate ours}.

For the comparison with \cite{LyuSissonWelsh2024aos}, part 2 of Condition B in the supplementary B of \cite{LyuSissonWelsh2024aos} states that, in our notation, the limit of the ratio $n_1 / n_2$ exists in $[0, \infty]$. We show that neither of this condition or our rate condition is weaker.
If we set $n_2 = \log n_1$, it is obvious to see that the condition in \cite{LyuSissonWelsh2024aos} is met while our condition does not hold. Conversely, set
\begin{equation}
    n_2
    = \begin{cases}
        n_1 & \text{if $n_1$ is odd.}\\
        n_1^2 & \text{if $n_1$ is even.}
    \end{cases} \notag
\end{equation}
Then $n_1/n_2$ does not have a limit as it moves between $1$ and $1/n_1$ so that the condition of \cite{LyuSissonWelsh2024aos} is violated. Meanwhile, it is not hard to see that our condition holds in this case.

\section{Additional Simulation Results}\label{app:addition}

\subsection{Sensitivity and adaptation of discretization}

We investigate the sensitivity of the proposed CDE approach to the choice of discretization levels by comparing estimation accuracy and computational time across several group-number specifications. 
We first focus on the two-way binary models setting and adopt the same simulation design as in Section~4.1. 
Specifically, we consider five candidate discretizations, $(G,H)=(\lfloor \sqrt{m} \rfloor,\lfloor \sqrt{n} \rfloor)$ (CDE1), $(\lfloor 0.7\sqrt{m} \rfloor,\lfloor 0.7\sqrt{n} \rfloor)$ (CDE2), $(\lfloor 1.3\sqrt{m} \rfloor,\lfloor 1.3\sqrt{n} \rfloor)$ (CDE3), $(20,20)$ (CDE4), and $(40,40)$ (CDE5). 
In addition, as an adaptive variant, we implement a data-driven selection that chooses the best candidate among these five specifications using the AIC, denoted by aCDE. 
Table~\ref{tab:sim-bin-sensitivity} reports the mean squared errors (MSE) of the point estimates and the coverage probabilities (CP) of nominal $95\%$ confidence intervals, based on 200 Monte Carlo replications. 

Overall, the estimation accuracy is fairly insensitive to the choice of $(G,H)$. 
In fact, the MSE and CP remain comparable across the five candidate discretizations over a wide range of sample sizes. If anything, excessively fine discretization (e.g., CDE5) can slightly deteriorate estimation accuracy, presumably due to increased variability induced by small sample sizes of each group. 
One possible explanation for this robustness with respect to $(G,H)$ is the role of Algorithm~2. 
Although the proposed CDE method first relies on a discrete approximation of the random effects, the final estimator is obtained after applying the pseudo-probability smoothing step. 
This additional smoothing would effectively mitigate the impact of the initial discretization, which makes the final performance insensitive to the specific choice of $(G,H)$. 
On the other hand, the computational time increases with the group numbers, and the adaptive approach by AIC incurs an additional cost since it requires fitting all candidates (hence the total time is essentially the sum of the five fits). 
Given that the five candidates yield only marginal differences in accuracy, the AIC-based adaptation brings little practical gain while substantially increasing computation. 
These results support our default recommendation, namely setting $(G, H)=(\lfloor{\sqrt{m}}\rfloor, \lfloor{\sqrt{n}}\rfloor)$ as a practically reasonable and robust choice.

We also conducted the same study for the three-way Poisson model using the same scenarios in Section~4.2. 
We consider five candidate discretizations, $(G_a, G_b, G_c)=(\lfloor \sqrt{n_a} \rfloor,\lfloor \sqrt{n_b} \rfloor, \lfloor \sqrt{n_c} \rfloor)$ (CDE1), $(\lfloor 0.7\sqrt{n_a} \rfloor,\lfloor 0.7\sqrt{n_b} \rfloor, \lfloor 0.7\sqrt{n_c} \rfloor)$ (CDE2), $(\lfloor 1.3\sqrt{n_a} \rfloor,\lfloor 1.3\sqrt{n_b} \rfloor, \lfloor 1.3\sqrt{n_c} \rfloor)$ (CDE3), $(20,20,20)$ (CDE4), and $(40,40,40)$ (CDE5), and the adaptive method by AIC (aCDE). 
The results based on 200 Monte Carlo replications are given in Table~\ref{tab:sim-Po-sensitivity}.
Similarly to the binary case, the results are generally insensitive to the choice of $(G_a, G_b, G_c)$, suggesting the default recommendation would be reasonable.

% Table 
\begin{table}[htbp]
\centering
{\small
\begin{tabular}{cccccccccccccccc}
\hline
Scenario & &{\footnotesize $N(\times 10^3)$} &  & {\footnotesize CDE1} & {\footnotesize CDE2} & {\footnotesize CDE3} & {\footnotesize CDE4} & {\footnotesize CDE5} & {\footnotesize aCDE}  \\
\hline
 &  & 5 &  & 1.75 & 1.91 & 1.70 & 1.69 & 1.81 & 1.90 \\
 &  & 10 &  & 0.64 & 0.66 & 0.64 & 0.65 & 0.79 & 0.66 \\
1 & MSE & 20 &  & 0.33 & 0.37 & 0.33 & 0.33 & 0.36 & 0.36 \\
 &  & 40 &  & 0.19 & 0.19 & 0.20 & 0.19 & 0.22 & 0.19 \\
 &  & 80 &  & 0.09 & 0.09 & 0.10 & 0.09 & 0.11 & 0.10 \\
 \hline
 &  & 5 &  & 93.0 & 91.8 & 92.4 & 91.4 & 91.0 & 91.8 \\
 &  & 10 &  & 95.2 & 95.8 & 95.2 & 95.4 & 93.2 & 95.8 \\
1 & CP & 20 &  & 94.4 & 94.4 & 95.0 & 94.8 & 94.4 & 94.2 \\
 &  & 40 &  & 92.6 & 92.8 & 92.0 & 92.6 & 90.6 & 92.8 \\
 &  & 80 &  & 93.2 & 93.6 & 92.2 & 93.4 & 90.6 & 93.6 \\
\hline
 &  & 5 &  & 1.56 & 1.68 & 1.56 & 1.63 & 1.86 & 1.67 \\
 &  & 10 &  & 0.72 & 0.76 & 0.74 & 0.75 & 0.90 & 0.76 \\
2 & MSE & 20 &  & 0.33 & 0.35 & 0.35 & 0.34 & 0.43 & 0.35 \\
 &  & 40 &  & 0.19 & 0.18 & 0.20 & 0.19 & 0.24 & 0.19 \\
 &  & 80 &  & 0.10 & 0.09 & 0.11 & 0.09 & 0.12 & 0.09 \\
 \hline
 &  & 5 &  & 93.0 & 93.6 & 92.6 & 92.2 & 90.8 & 93.4 \\
 &  & 10 &  & 94.8 & 95.2 & 94.2 & 94.0 & 92.0 & 95.0 \\
2 & CP & 20 &  & 96.2 & 95.6 & 94.8 & 95.6 & 93.0 & 95.4 \\
 &  & 40 &  & 94.0 & 94.8 & 93.2 & 94.0 & 91.0 & 94.4 \\
 &  & 80 &  & 94.4 & 94.6 & 92.8 & 94.6 & 90.2 & 95.0 \\
\hline
\end{tabular}
}
\caption{Mean squared errors (MSE) of the point estimates and coverage probability (CP) of $95\%$ confidence intervals for $\beta$ under two-way logistic CDE models with $(G,H)=(\lfloor{\sqrt{m}}\rfloor, \lfloor{\sqrt{n}})\rfloor$ (CDE1), $(\lfloor{0.7\sqrt{m}}\rfloor, \lfloor{0.7\sqrt{n}}\rfloor)$ (CDE2), $(\lfloor{1.3\sqrt{m}}\rfloor, \lfloor{1.3\sqrt{n}}\rfloor)$ (CDE3), $(20, 20)$ (CDE4), $(40, 40)$ (CDE5) and optimal choice selected by AIC among the five candidates (aCDE). 
}
\label{tab:sim-bin-sensitivity}
\end{table}

% Table 
\begin{table}[htbp]
\centering
{\small
\begin{tabular}{cccccccccccccccc}
\hline
Scenario & &{\footnotesize $N(\times 10^3)$} &  & {\footnotesize CDE1} & {\footnotesize CDE2} & {\footnotesize CDE3} & {\footnotesize CDE4} & {\footnotesize CDE5} & {\footnotesize aCDE}  \\
\hline
 &  & 5 &  & 0.131 & 0.133 & 0.130 & 0.130 & 0.138 & 0.131 \\
 &  & 10 &  & 0.062 & 0.063 & 0.062 & 0.062 & 0.063 & 0.062 \\
1 & MSE & 20 &  & 0.036 & 0.036 & 0.035 & 0.035 & 0.035 & 0.035 \\
 &  & 40 &  & 0.017 & 0.017 & 0.017 & 0.017 & 0.016 & 0.017 \\
 &  & 80 &  & 0.009 & 0.009 & 0.009 & 0.009 & 0.009 & 0.009 \\
 \hline
 &  & 5 &  & 92.0 & 92.2 & 93.0 & 93.4 & 92.6 & 92.2 \\
 &  & 10 &  & 93.4 & 93.2 & 94.0 & 93.8 & 93.0 & 93.6 \\
1 & CP & 20 &  & 91.6 & 90.8 & 91.6 & 91.2 & 91.4 & 91.2 \\
 &  & 40 &  & 92.4 & 92.8 & 93.2 & 92.8 & 92.6 & 92.8 \\
 &  & 80 &  & 93.2 & 92.2 & 92.8 & 93.2 & 93.4 & 92.8 \\
 \hline
 &  & 5 &  & 0.141 & 0.141 & 0.139 & 0.139 & 0.141 & 0.142 \\
 &  & 10 &  & 0.077 & 0.078 & 0.076 & 0.077 & 0.079 & 0.078 \\
2 & MSE & 20 &  & 0.032 & 0.031 & 0.031 & 0.031 & 0.030 & 0.032 \\
 &  & 40 &  & 0.019 & 0.019 & 0.019 & 0.019 & 0.018 & 0.019 \\
 &  & 80 &  & 0.009 & 0.009 & 0.009 & 0.009 & 0.009 & 0.009 \\
 \hline
 &  & 5 &  & 91.6 & 92.6 & 92.6 & 92.2 & 92.0 & 91.8 \\
 &  & 10 &  & 91.8 & 91.0 & 92.0 & 91.8 & 91.4 & 91.4 \\
2 & CP & 20 &  & 94.6 & 94.2 & 94.2 & 94.2 & 94.0 & 94.6 \\
 &  & 40 &  & 90.0 & 91.0 & 90.8 & 90.8 & 91.8 & 90.2 \\
 &  & 80 &  & 92.6 & 93.2 & 92.4 & 92.6 & 92.4 & 92.4 \\
\hline
\end{tabular}
}
\caption{Mean squared errors (MSE) of the point estimates and coverage probability (CP) of $95\%$ confidence intervals for $\beta$ under three-way Poisson CDE models with $(G_a,G_b,G_c)=(\lfloor{\sqrt{n_a}}\rfloor, \lfloor{\sqrt{n_b}}\rfloor, \lfloor{\sqrt{n_c}}\rfloor)$ (CDE1), $(\lfloor{0.7\sqrt{n_a}}\rfloor, \lfloor{0.7\sqrt{n_b}}\rfloor, \lfloor{0.7\sqrt{n_c}}\rfloor)$ (CDE2), $(\lfloor{1.3\sqrt{n_a}}\rfloor, \lfloor{1.3\sqrt{n_b}}\rfloor, \lfloor{1.3\sqrt{n_c}}\rfloor)$ (CDE3), $(20, 20, 20)$ (CDE4), $(40, 40, 40)$ (CDE5) and optimal choice selected by AIC among the five candidates (aCDE). 
}
\label{tab:sim-Po-sensitivity}
\end{table}

\subsection{Comparison of inference procedures}

% weighted bootstrapを提案
% -> 重み付き尤度の最小化は提案アルゴリズムの微修正で対応できる
% Bootstrapをするとややcoverageが改善する
% large sampleではほぼ変わらない
% supplement p.25のサンドイッチ型も比較に入れる
% -> plug-inよりも小さめなSEが計算される (おそらくplug-inを計算するときにsmoothed REを使っているからplug-inの性能が高い)

% 想定される結論
% - large Nで高速に計算する場合はplug-in
% - moderate Nでちゃんとinferenceをする場合はbootstrap

We here compare inference procedures based on the standard error estimation by ``Plug-in'', ``Sandwich'' and ``Bootstrap'' methods. 
The Sandwich method corresponds to computing the asymptotic covariance matrix provided in the proof of Theorem~2.

To compare the three methods, we first consider two-way logistic models considered in Section~4.1.
Using the same scenarios, we applied the three inference methods, where the number of bootstrap replications is set to 100. 
Based on 200 Monte Carlo replications, we compute the coverage probability (CP) and average interval length (AL) of the $95\%$ confidence intervals of regression coefficients, under $N\in\{5000, 10000, 20000\}$. 
The results are reported in Table~\ref{tab:sim-bin-inference}.

Based on the above considerations, we recommend the following strategy.
For large $N$ and computationally constrained settings, the plug-in variance estimator
provides fast and reliable inference.
For moderate sample sizes, where accurate finite-sample inference is desired, the
weighted bootstrap offers improved coverage at the cost of increased computation.

% Table 
\begin{table}[htbp]
\centering
{\small
\begin{tabular}{cccccccccccccccc}
\hline
 & {\footnotesize $N(\times 10^3)$} &  & Plug-in & Sandwich & Bootstrap & & Plug-in & Sandwich & Bootstrap \\
\hline
 & 5 &  & 92.6 & 92.7 & 93.8 &  & 93.4 & 92.3 & 95.2 \\
CP (\%) & 10 &  & 94.4 & 94.8 & 95.2 &  & 92.5 & 92.3 & 93.7 \\
 & 20 &  & 94.1 & 94.1 & 95.1 &  & 92.8 & 93.0 & 93.9 \\
\hline
 & 5 &  & 0.144 & 0.143 & 0.153 &  & 0.146 & 0.144 & 0.155 \\
AL & 10 &  & 0.102 & 0.102 & 0.107 &  & 0.103 & 0.103 & 0.108 \\
 & 20 &  & 0.072 & 0.072 & 0.075 &  & 0.073 & 0.073 & 0.076 \\
\hline
\end{tabular}
}
\caption{
Coverage probability (CP) and average interval length (AL) of $95\%$ confidence intervals of regression coefficients under two-way logistic models, based on 200 Monte Carlo replications.
}
\label{tab:sim-bin-inference}
\end{table}

We also conducted the same assessment under three-way Poisson model considered in Section~4.2.
Using the same scenarios, we applied the three inference methods, where the number of bootstrap replications is set to 100. 
Based on 200 Monte Carlo replications, we compute the coverage probability (CP) and average interval length (AL) of the $95\%$ confidence intervals of regression coefficients, under $N\in\{2500, 5000, 10000\}$. 
The results are reported in Table~\ref{tab:sim-Po-inference}.
The conclusion is mostly the same as the that in binary models.

% Table 
\begin{table}[htbp]
\centering
{\small
\begin{tabular}{cccccccccccccccc}
\hline
 & {\footnotesize $N(\times 10^3)$} &  & Plug-in & Sandwich & Bootstrap & & Plug-in & Sandwich & Bootstrap \\
\hline
 & 2.5 &  & 93.3 & 91.8 & 92.8 &  & 92.4 & 90.7 & 91.7 \\
CP (\%) & 5 &  & 93.6 & 93.1 & 93.4 &  & 91.7 & 90.7 & 91.6 \\
 & 10 &  & 93.8 & 93.4 & 93.5 &  & 92.7 & 92.6 & 92.6 \\
\hline
 & 2.5 &  & 0.043 & 0.040 & 0.043 &  & 0.043 & 0.041 & 0.043 \\
AL & 5 &  & 0.030 & 0.029 & 0.030 &  & 0.030 & 0.029 & 0.030 \\
 & 10 &  & 0.021 & 0.021 & 0.021 &  & 0.021 & 0.021 & 0.021 \\
\hline
\end{tabular}
}
\caption{
Coverage probability (CP) and average interval length (AL) of $95\%$ confidence intervals of regression coefficients under three-way Poisson models, based on 200 Monte Carlo replications.
}
\label{tab:sim-Po-inference}
\end{table}

\section{Additional Application Results}\label{app:three-way}

We here provide application results of the MovieLens data through three-way effects (classified by user, movie and zip code). 
In Figure~\ref{fig:app-RE-3way}, we provide the estimated function of age and scatter plots of estimates of regression coefficients for dummy variables.
The results under three-way effects are not much different from those under two-way effects, suggesting that introducing the additional random effect mainly refines the decomposition of the latent effects without substantially altering the inference for the global parameters.
In Figure~\ref{fig:app-RE-3way}, we present the estimated random effects of three classifications.
It shows that the distributional form of the user effect is quite different from that under two-way effects.
One possible explanation is that the classifications of user and zip code are closely related: zip codes provide a coarser grouping of users, so that the zip classification partially aggregates the user classification.
This nested or overlapping structure may lead to the observed difference in the distribution of the user effects.

%  Figure 
\begin{figure}[htbp!]
\centering
\includegraphics[width=\linewidth]{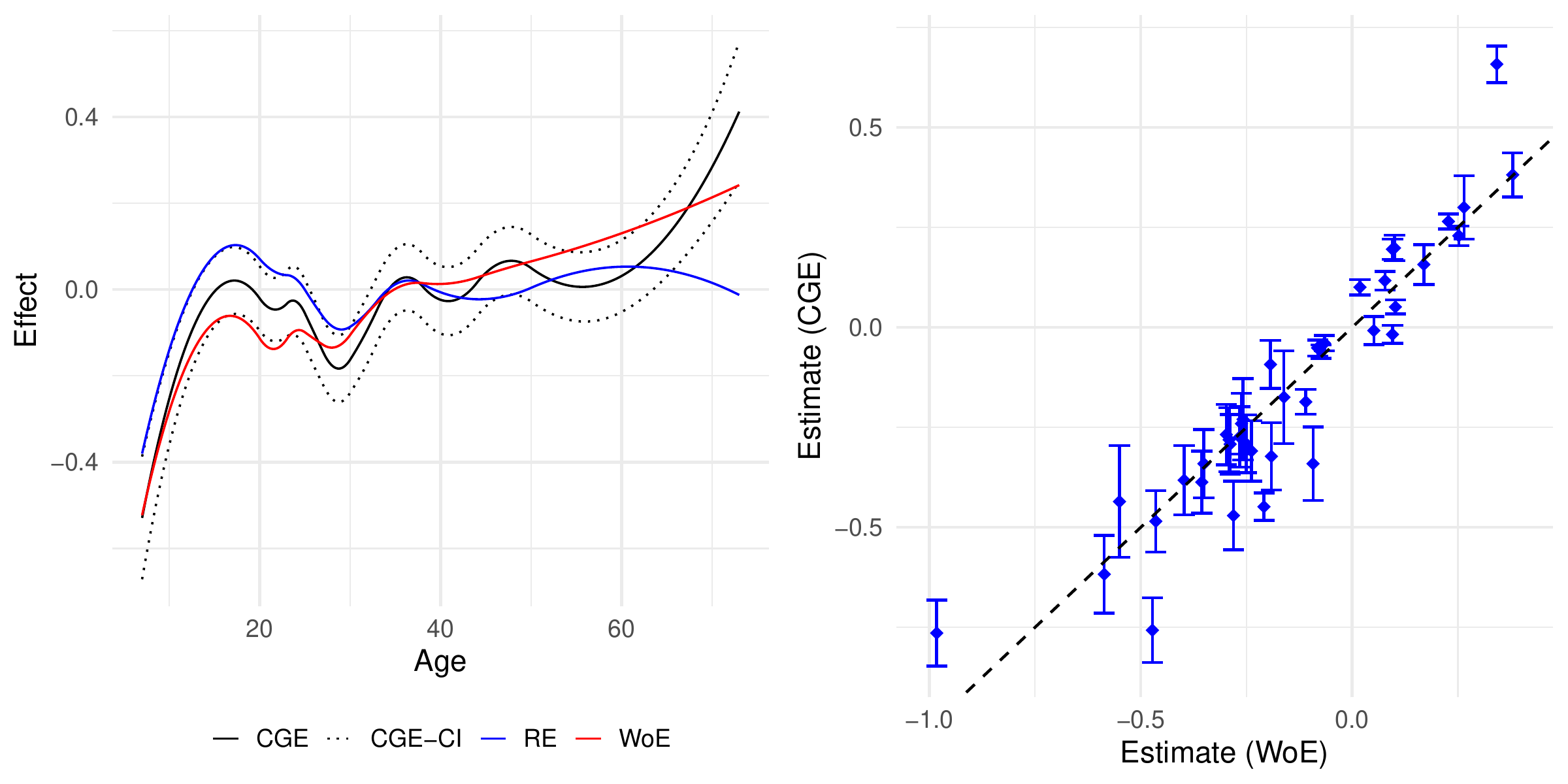}
\caption{Estimated regression function of age (left) and scatter plots of estimates of regression coefficients for dummy variables with $95\%$ confidence intervals for CDE (right) under three-way data.   } 
\label{fig:app-reg-3way}
\end{figure}

%  Figure 
\begin{figure}[htbp!]
\centering
\includegraphics[width=\linewidth]{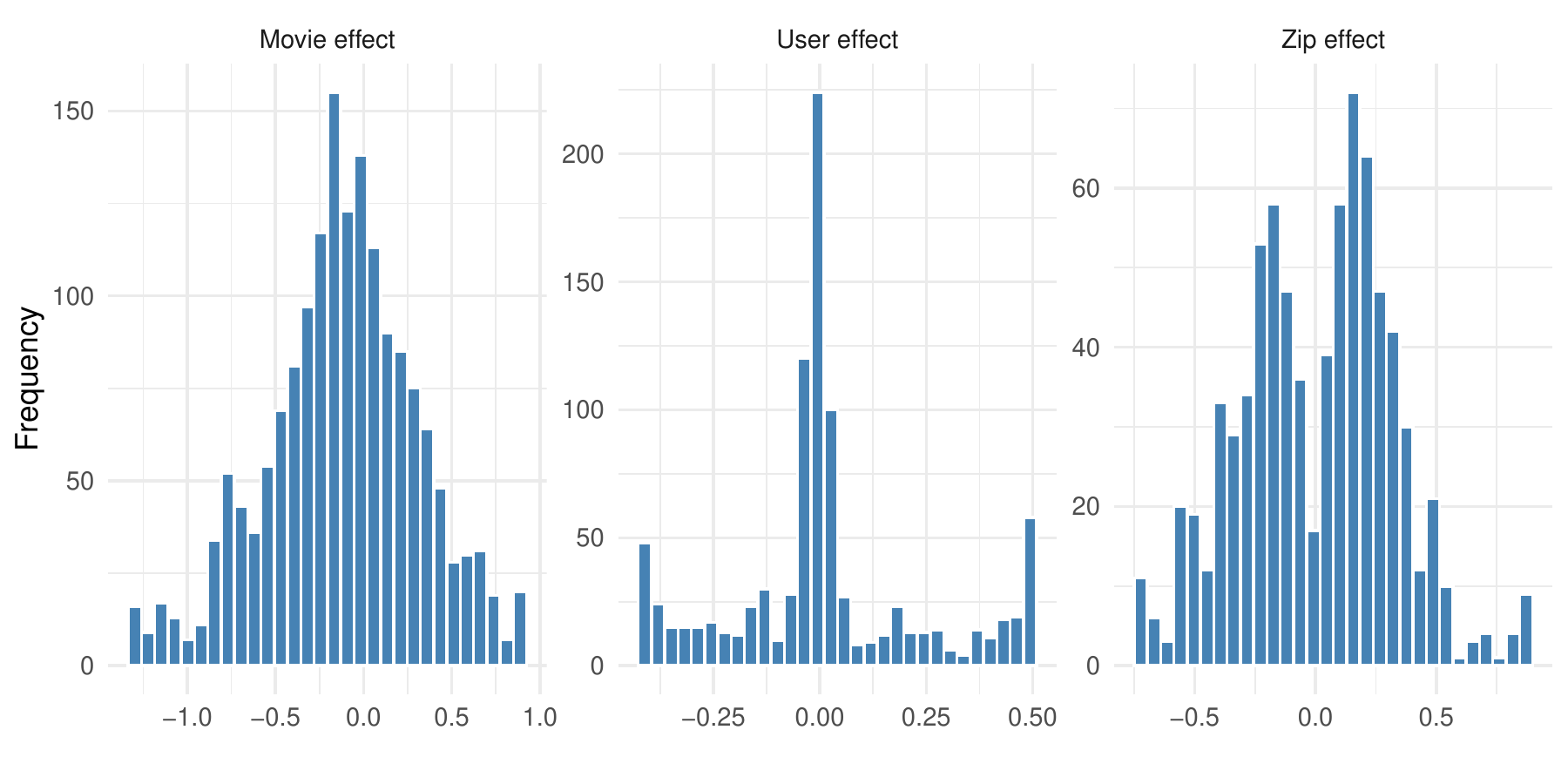}
\caption{Histograms of estimated user, movie and zip effects. } 
\label{fig:app-RE-3way}
\end{figure}

\end{document}